\documentclass[english, 12pt, letter]{article}
\usepackage{amsmath,amssymb,amsthm}
\usepackage{natbib}
\usepackage[nodisplayskipstretch]{setspace}
\usepackage{bm}
\usepackage{threeparttable}
\usepackage{graphicx}
\usepackage{enumerate}
\usepackage{commath}			
\usepackage{bibentry}
\usepackage{booktabs}
\usepackage[titletoc,title]{appendix}
\usepackage{multirow}
\usepackage{arydshln}
\usepackage{tabularx}
\usepackage{dsfont}
\usepackage[font={footnotesize}]{subcaption}
\usepackage[left=1in, right=1in, top=1in, bottom=1in]{geometry}

\usepackage[colorinlistoftodos,textsize=small]{todonotes}
\allowdisplaybreaks

\usepackage[
	breaklinks,
	colorlinks=true,
	pagebackref=false,
	pdfauthor={},%
	pdftitle={},%
	citecolor=blue,
	linkcolor=blue,
	urlcolor=blue
	]{hyperref}      

\newtheorem{lem}{Lemma}
\newtheorem{assumption}{Assumption}
\newtheorem{prop}{Proposition}
\newtheorem{cor}{Corollary}

\newtheorem{thm}{Theorem}

\newtheoremstyle{assumption}{1ex}{1ex}%
      {\it}
      {}
      {\bf}
      {.}
      {5pt}
      {\thmname{#1}\thmnumber{ #2}*\thmnote{ \slshape{(#3)}}} 
\theoremstyle{assumption}	
\newtheorem{assumptionA}{Assumption}
\newtheoremstyle{remark2}{1ex}{1ex}%
      {}
      {}
      {\bf}
      {.}
      {5pt}
      {\thmname{#1}\thmnumber{ #2}\thmnote{ \slshape{(#3)}}} 
\theoremstyle{remark2}
\newtheorem{rem}{Remark}

\newtheorem{example}{Example}

\newcommand{\Burr}{\operatorname{Burr}}

\usepackage{babel}
\input{ee.sty}

\makeatletter
\renewenvironment{proof}[1][\bfseries\proofname]{\par
   \pushQED{\qed}%
   \normalfont \topsep6\p@\@plus6\p@\relax
   \trivlist
   \item[\hskip\labelsep
     #1\@addpunct{:}]\ignorespaces
}{%
   \popQED\endtrivlist\@endpefalse
}

\providecommand{\leftsquigarrow}{%
  \mathrel{\mathpalette\reflect@squig\relax}%
}
\newcommand{\reflect@squig}[2]{%
  \reflectbox{$\m@th#1\rightsquigarrow$}%
}
\makeatother

\newcommand{\Comments}{1}
\newcommand{\mynote}[2]{\ifnum\Comments=1\textcolor{#1}{#2}\fi}
\newcommand{\mytodo}[2]{\ifnum\Comments=1%
  \todo[linecolor=#1!80!black,backgroundcolor=#1,bordercolor=#1!80!black]{#2}\fi}

\ifnum\Comments=1               
\fi

\ifnum\Comments=1               
\fi

\newcommand{\VaR}{\operatorname{VaR}}
\newcommand{\p}{\operatorname{P}}

\newcommand{\D}{\,\mathrm{d}}

\begin{document}

\baselineskip18pt
\renewcommand\floatpagefraction{.9}
\renewcommand\topfraction{.9}
\renewcommand\bottomfraction{.9}
\renewcommand\textfraction{.1}
\setcounter{totalnumber}{50}
\setcounter{topnumber}{50}
\setcounter{bottomnumber}{50}
\abovedisplayskip1.5ex plus1ex minus1ex
\belowdisplayskip1.5ex plus1ex minus1ex
\abovedisplayshortskip1.5ex plus1ex minus1ex
\belowdisplayshortskip1.5ex plus1ex minus1ex

\title{The Estimation Risk in Extreme Systemic Risk Forecasts\thanks{This work was funded by the Deutsche Forschungsgemeinschaft (DFG, German Research Foundation) through project 460479886.}
}

\author{
	Yannick Hoga\thanks{Faculty of Economics and Business Administration, University of Duisburg-Essen, Universit\"atsstra\ss e 12, D--45117 Essen, Germany, \href{mailto:yannick.hoga@vwl.uni-due.de}{yannick.hoga@vwl.uni-due.de}.}
}

\date{\today}
\maketitle

\begin{abstract}
\noindent Systemic risk measures have been shown to be predictive of financial crises and declines in real activity. Thus, forecasting them is of major importance in finance and economics. In this paper, we propose a new forecasting method for systemic risk as measured by the marginal expected shortfall (MES). It is based on first de-volatilizing the observations and, then, calculating systemic risk for the residuals using an estimator based on extreme value theory. We show the validity of the method by establishing the asymptotic normality of the MES forecasts. The good finite-sample coverage of the implied MES forecast intervals is confirmed in simulations. An empirical application to major US banks illustrates the significant time variation in the precision of MES forecasts, and explores the implications of this fact from a regulatory perspective. \\


\end{abstract}


\newpage

\section{Motivation}\label{Motivation}

The financial crisis of 2007--2009 has sparked regulatory and academic interest in assessing systemic risk. For instance, in April 2009 G20 leaders asked national regulators to develop guidelines for the assessment of the systemic importance of financial institutions, which were provided in a joint report by the International Monetary Fund (IMF), the Bank for International Settlements (BIS) and the Financial Stability Board (FSB) \citep{IBS09}. By now, these early developments have manifested themselves in official regulations. For example, the class of global systemically important banks (G-SIBs) is divided into buckets from 1 to 5, where banks in bucket 5 have to hold the highest additional capital buffers.\footnote{See \url{https://www.fsb.org/wp-content/uploads/P231121.pdf} for the 2021 list of the 30 G-SIBs.} The Basel framework, that sets out the methodology determining G-SIB membership, closely relies on systemic risk assessments \citep[SCO40]{BCBSBF19}.

Next to its regulatory importance, forecasting systemic risk is important in various other contexts. First, one hallmark of financial crises is that asset prices start to co-move. To measure the extent to which prices move in lockstep, several systemic risk measures may be used \citep{Bea12,AB16,Aea17}. Since the tendency of asset prices to co-move implies that diversification benefits are seriously reduced, it becomes important to predict systemic risk measures as indicators of diversification meltdown. Second, investigating the link between the real economy and the financial sector, \citet{ABT12}, \citet{GKP16} and \citet{BE17} find an increase in systemic risk to be predictive of future declines in real activity. These examples underscore the importance of accurately forecasting systemic risk.

Predictions of systemic risk may be produced from a variety of models, such as multivariate GARCH models or quantile regression models \citep{GT13,AB16}. However, very little is known about the asymptotic properties (e.g., consistency or asymptotic normality) of systemic risk forecasts issued from these or other models. This contrasts with the large literature developing asymptotic theory for forecasts of \textit{univariate} risk measures, such as the Value-at-Risk (VaR) or the expected shortfall (ES) \citep{Cea07,GS08,FZ15,Hog18+}. Therefore, it is the main aim of this paper to fill this gap for systemic risk forecasts. Specifically, we establish conditions under which systemic risk forecasts issued from a general class of multivariate GARCH-type models are (consistent and) asymptotically normal.

Of course, consistency is vital for the point forecasts to reflect actual levels of systemic risk. However, point forecasts of systemic risk are only of limited value, since they lack a measure of uncertainty. To illustrate the importance of confidence intervals around point risk forecasts, \citet{CG05} give the example of a portfolio manager allowed to take on portfolios with a VaR of at most 15\% of current capital. A VaR point estimate of 13\% would not indicate any need for rebalancing, yet a (say) 90\%-confidence interval of 10--16\% would induce the prudent portfolio manager to do so. Clearly, a similar case can be made for the importance of confidence intervals for systemic risk forecasts, which our asymptotic normality result allows us to construct.

As our systemic risk measure, we use the marginal expected shortfall (MES) of \citet{Aea17}. We do so, because of the ability of MES to identify key contributors to systemic risk during financial crises and due to its predictive content for a downturn in real economic activity \citep{GKP16,Aea17}. Furthermore, MES has an additivity property that is crucial for attributing systemic risk \citep{CIM13}. This may be useful for individual banks as well as the financial system as a whole. For the purposes of risk management or asset allocation, an individual bank may want to break down firm-wide losses into contributions from single units or trading desks. From the wider perspective of the financial system, the additivity property allows to decompose the system-wide ES into the sum of the MESs of all banks in the system. We also refer to the empirical application for an illustration of this property.

For MES to truly capture \textit{systemic} risk, it needs to be forecasted far out in the tail. For instance, \citet[p.~13]{Aea17} state that ``[w]e can think of systemic events [...] as extreme tail events that happen once or twice a decade''. Consequently, only few meaningful observations are available for forecasting MES. To deal with this, our MES estimator is motivated by extreme value theory (EVT). By imposing weak assumptions on the joint tail, EVT-based methods alleviate the problem of data scarceness outside the center of the distribution. Indeed, numerous studies show that EVT-based estimators improve the forecast quality of univariate risk measures, such as the VaR or the ES \citep{MF00,BLS06,KMP06,Bal07,Hog19+}. Hence, these methods have caught on in empirical work as well \citep{GL05}. We stress that the case for using EVT-based estimators for systemic risk measures, where joint extremes are of interest, seems to be even stronger than in the univariate case because data is even scarcer in the \textit{joint} tail.

A key ingredient of our MES forecast is the MES estimator of \citet{Cea15}, and in deriving asymptotic properties of MES forecasts we build on their work. However, \citet{Cea15} deal with unconditional MES estimation for independent, identically distributed (i.i.d.) random variables and, thus, their framework is inherently static. In contrast, we consider (conditional) MES forecasting in dynamic models. In particular, this requires taking volatility dynamics into account in the forecasts. From an econometric perspective, conditional MES has the advantage of incorporating current market conditions. Therefore, changes in market conditions are reflected in conditional MES, but not in its static version. This is akin to the standard deviation as a measure of risk. The conditional standard deviation (a.k.a.~volatility) is a good measure of current risk, while the unconditional standard deviation only provides an average measure of risk. 

We also extend our results (and the results of \citet{Cea15}) to a higher dimensional setting, where we consider MES forecasts for multiple variables jointly. This is important because one of the main purposes of systemic risk measures is to explore linkages in complex systems. We prove the joint asymptotic normality of MES forecasts in the higher-dimensional case. To enable inference, we propose an estimator of the asymptotic variance-covariance matrix and show its consistency. Then, we demonstrate how our results can be used to test for equal systemic risk contributions of the different units in the system. This test is later on used in the empirical application. From a technical perspective, our higher-dimensional results draw on \citet{Hog17a+}, who explores tail index estimation for multivariate time series.

We confirm the good finite-sample coverage of our asymptotic confidence intervals for MES in simulations. We do so for the constant conditional correlation (CCC) GARCH model of \citet{Bol90}. Our main findings are that coverage improves the more extreme the risk level of the MES forecasts. Also, the better the (marginal and joint) tail can be approximated using extreme value methods, the more precise the forecasts tend to be in terms of root mean square error (RMSE) and the lengths of the forecast intervals.

Our empirical application considers the 8 G-SIBS from the US. As expected, there is significant time variation in the levels of systemic risk as measured by the MES. Significant peaks in systemic risk can be observed during the financial crisis of 2007--2009, the European sovereign debt crisis in 2011 and the Corona crash of March 2020. Computing our MES forecast intervals over time shows that it is particularly during times of crises---when accurate systemic risk assessments are needed most---that forecasts tend to be least precise (as measured by the lengths of the forecast intervals). We also apply our test for equal systemic risk contributions of each of the 8 banks. Not surprisingly, the null of equality can be rejected for every single time point in our sample. This is consistent with the fact that the 8 banks are assigned to different buckets in the G-SIB classification.

The remainder of the paper is structured as follows. Section~\ref{Preliminaries} introduces the multivariate volatility model and our MES forecasts together with all required regularity conditions. Section~\ref{Asymptotic normality of MES forecasts} derives limit theory for MES forecasts and Section~\ref{Higher-Dimensional Extensions} extends this to multiple MES forecasts. Coverage of the confidence intervals for MES is assessed in the simulations in Section~\ref{Simulations}. The empirical application in Section~\ref{Empirical Application} investigates MES forecasts for the 8 US G-SIBs. The final Section~\ref{Conclusion} concludes. All proofs are relegated to the Appendix.

\section{Preliminaries}\label{Preliminaries}

We adopt the following notational conventions. Throughout, we use bold letters to denote vectors and matrices. In particular, $\mI$ is the $(2\times2)$-identity matrix. For any matrix $\mA=(a_{ij})$, we will use the norm defined by $\Vert\mA\Vert=\sum_{i,j}|a_{ij}|$. The transpose of a matrix $\mA$ is denoted by $\mA^\prime$ and its vectorization by $\vec(\mA)$, where the columns of $\mA$ are stacked on top of each other. The diagonal matrix containing the elements of the vector $\vv$ on the main diagonal is $\diag(\vv)$. We let $\overset{d}{=}$ stand for equality in distribution. For some random variable $X$, put $X^{+}=\max\{X,0\}$ and $X^{-}=X-X^{+}$. For scalar sequences $a_n$ and $b_n$, we write $a_n\asymp b_n$ if both $a_n=O(b_n)$ and $b_n=O(a_n)$ hold, as $n\to\infty$, and we write $a_n\sim b_n$ if $a_n/b_n\to1$.

\subsection{The data-generating process}\label{The data-generating process}

Consider a sample $\{(X_t, Y_t)^\prime\}_{t=-\ell_n+1,\ldots, n}$ from the random variables of interest. In a forecasting situation, interest focuses on predicting systemic risk based on the current state of the market, which is captured by the information set $\mathcal{F}_{n}=\sigma\big((X_n, Y_n)^\prime, (X_{n-1}, Y_{n-1})^\prime,\ldots\big)$. Define $F_{n}(x,y)=\p\{X_{n+1}\leq x, Y_{n+1}\leq y\mid \mathcal{F}_{n}\}$ to be the conditional joint distribution function (d.f.) with marginals $F_{n,0}(x)=F_{n}(x,\infty)$ and $F_{n,1}(y)=F_{n}(\infty,y)$. Then, the (conditional) MES is defined as
\[
	\theta_{n,p}=\E\big[Y_{n+1}\mid X_{n+1}>\VaR_{n}(p), \mathcal{F}_n\big],
\]
where $\VaR_{n}(p)=F_{n,0}^{\leftarrow}(1-p)$, with ``$\leftarrow$'' indicating the left-continuous inverse, is the VaR of the reference position for some ``small'' $p\in(0,1)$. Thus, under current market conditions, MES measures next period's average loss $Y_{n+1}$ given that $X_{n+1}$ is in distress.

An MES forecast (based on $\mathcal{F}_n$) may be of interest in a number of situations.
For instance, when the $Y_t$ are losses of one's own portfolio, and the $X_t$ denote losses of some reference index, such as the S\&P~500. The $Y_t$ may also denote the losses of a single trading desk, and the $X_t$ firm-wide losses. Alternatively, the $Y_t$ may be the losses of a financial institution with $X_t$ standing for system-wide losses. In each case, it may be of interest to understand how the two risk factors are connected, e.g., for the purposes of stress testing or to assess portfolio sensitivities. In these situations, the MES can serve as a (real-valued) summary of the degree of connectedness.

Throughout, we suppose that the losses are generated from the model
\begin{equation}\label{eq:model}
	\begin{pmatrix}X_t\\ Y_t\end{pmatrix} = \mSigma_t(\vtheta^\circ)\vvarepsilon_t,
\end{equation}
where the true parameter vector $\vtheta^\circ$ is an element of some parameter space $\mTheta$ and the diagonal matrix $\mSigma_t(\vtheta^\circ)=\mSigma_t=\diag(\vsigma_t)$ with $\vsigma_t=(\sigma_{t,X}, \sigma_{t,Y})^\prime$ is $\mathcal{F}_{t-1}$-measurable. Moreover, the $\vvarepsilon_t=(\varepsilon_{t,X}, \varepsilon_{t,Y})^\prime$ are independent of $\mathcal{F}_{t-1}$ and i.i.d.~with mean zero, unit variance and correlation matrix $\mR$. Thus, $\mSigma_t$ contains the individual volatilities on the main diagonal, since
\begin{equation*}
	\Var\big((X_t, Y_t)^\prime\mid \mathcal{F}_{t-1}\big)=\mSigma_t\Var(\vvarepsilon_t)\mSigma_t^\prime=\mSigma_t\mR\mSigma_t^\prime=\begin{pmatrix}\sigma_{t,X}^2 & \rho_{X,Y} \sigma_{t,X} \sigma_{t,Y}\\ \rho_{X,Y} \sigma_{t,X} \sigma_{t,Y} & \sigma_{t,Y}^2\end{pmatrix},
\end{equation*}
where $\rho_{X,Y}=\corr(\varepsilon_{t,X}, \varepsilon_{t,Y})$ is the constant conditional correlation of $(X_t, Y_t)^\prime\mid\mathcal{F}_{t-1}$. We already mention here that estimating the correlation $\rho_{X,Y}$ of $\vvarepsilon_t$ is not required for conditional MES forecasting. Rather, it is the unconditional MES of $\vvarepsilon_t$ that will be required; see \eqref{eq:MES true}.

The best known among the class of models in \eqref{eq:model} is the CCC--GARCH model of \citet{Bol90}. But our framework also covers models incorporating volatility spillover, such as the extended (E)CCC--GARCH of \citet{Jea98}.

\begin{rem}
We work with CCC--GARCH-type models here. However, DCC--GARCH models, due to \citet{Eng02}, have attained benchmark status among multivariate GARCH models because of their forecasting accuracy \citep{LRV12}. We focus here on the former class of models for several reasons. First, the former models continue to be studied extensively in the literature \citep{Jea98,HT04,NT09,CK10}. Second, as \citet[p.~620]{FZ16b} point out, a full estimation theory for DCC--GARCH models is not available. Since MES forecasts necessarily require a parameter estimator, developing limit theory for MES forecasts based on DCC--GARCH models is beyond the scope of the present paper. Third, our EVT-based estimation method exploits dependence between $\varepsilon_{t,X}$ and $\varepsilon_{t,Y}$. However, in DCC--GARCH models, the innovations are decorrelated (to ensure identification), in which case our MES estimator cannot be expected to work well. Fourth, higher-order measures of risk, such as MES, are not properly identified in DCC--GARCH models \citep{HHM22}. To see this, recall that in a DCC--GARCH framework, $(X_t, Y_t)^\prime=\mSigma_t\vvarepsilon_t$ for i.i.d.~$\vvarepsilon_t\sim(\vzero,\mI)$ and not necessarily diagonal $\mSigma_t$. However, only $\Var\big((X_t, Y_t)^\prime\mid \mathcal{F}_{t-1}\big)=\mSigma_t\mSigma_t^\prime=:\mH_t$ is modeled, while the model stays silent on the choice of (the non-unique) $\mSigma_t$. However, different $\mSigma_t$ imply different conditional distributions $(X_t, Y_t)^\prime\mid\mathcal{F}_{t-1}$ and, hence, different values for MES. For instance, $\mSigma_t$ may denote the symmetric square root implied by the eigenvalue decomposition ($\mSigma_{t}^{s}$), or it may be the lower triangular matrix of the Cholesky decomposition ($\mSigma_{t}^{l}$). Both decompositions are equivalent in the sense that both imply the same dynamics in second-order moments (as given in $\mH_t$). Yet, when it comes to higher-order measures of risk (such as MES), the two decompositions imply different values for MES, as the next example illustrates.
\end{rem}

\begin{example}\label{ex:sq}
Suppose that $\varepsilon_{t,X}$ and $\varepsilon_{t,Y}$ have a (standardized) Student's $t_{5}$-distribution, independently of each other. Assume for simplicity that 
\[
	\mH_{n+1}=\begin{pmatrix}5& 4\\ 4 & 5\end{pmatrix},\quad\text{implying}\quad\mSigma_{n+1}^{s}=\begin{pmatrix}2& 1\\ 1 & 2\end{pmatrix}\quad\text{and}\quad\mSigma_{n+1}^{l}=\begin{pmatrix}2.23...& 0\\ 1.78... & 1.34...\end{pmatrix}.
\]
Then, when the underlying structural model is $\mSigma_t^{s}\vvarepsilon_t$ (resp.~$\mSigma_t^{l}\vvarepsilon_t$), we have that $\theta_{n,p}=3.82...$ (resp.~$\theta_{n,p}=4.00...$). Thus, while second-order (cross-)moments of $X_t$ and $Y_t$ are identical for both $\mSigma_t^{s}\vvarepsilon_t$ and $\mSigma_t^{l}\vvarepsilon_t$, MES depends on the assumed structural model. The underlying reason for this is that the conditional distribution $\p\{X_{n+1}\leq\cdot, Y_{n+1}\leq \cdot\mid\mathcal{F}_{n}\}$ depends on the decomposition of the variance-covariance matrix.

From a structural perspective, the triangular structure of $\mSigma_t^{l}$ implies that $\varepsilon_{t,X}$ is the idiosyncratic shock pertaining to (say) the market return $X_t$, which also affects the (say) portfolio return $Y_t$. However, the market is not affected by $\varepsilon_{t,Y}$. This contrasts with, e.g., a symmetric assumption on $\mSigma_t^{s}$, where a unit shock in $\varepsilon_{t,X}$ has the same effect on $Y_t$ as a unit shock in $\varepsilon_{t,Y}$ on $X_t$.
\end{example}

\subsection{Model Assumptions}

Denote the d.f.~of the innovations in \eqref{eq:model} by $F(x,y)=\p\{\varepsilon_{t,X}\leq x,\ \varepsilon_{t,Y}\leq y\}$ (which we assume to be continuous) and the marginal d.f.s by $F_0(x)=F(x,\infty)$ and $F_1(y)=F(\infty,y)$. For model \eqref{eq:model}, we have that $X_{n+1}=\sigma_{n+1,X}\varepsilon_{n+1,X}$, such that $\VaR_{n}(p)=\sigma_{n+1,X}F_{0}^{\leftarrow}(1-p)$. Therefore, the MES becomes
\begin{equation}\label{eq:MES true}
	\theta_{n,p} = \sigma_{n+1,Y}\E\big[\varepsilon_{t,Y}\mid \varepsilon_{t,X}>F_{0}^{\leftarrow}(1-p)\big].
\end{equation}
Hence, a forecast of $\theta_{n,p}$ consists of two parts. First, volatility $\sigma_{n+1,Y}$ must be forecasted and, second, we must estimate $\theta_{p}:=\E[\varepsilon_{t,Y}\mid \varepsilon_{t,X}>F_{0}^{\leftarrow}(1-p)]$, i.e., the unconditional MES of $(\varepsilon_{t,X}, \varepsilon_{t,Y})^\prime$. Thus, in the following, we have to impose some regularity conditions on the parameter estimator and the volatility model (to forecast volatility via some $\widehat{\sigma}_{n+1,Y}(\widehat{\vtheta})$) and on the joint tail of the $\vvarepsilon_t=(\varepsilon_{t,X}, \varepsilon_{t,Y})^\prime$ (to estimate $\theta_{p}$).

We begin by imposing assumptions on the estimator and the volatility model.
Regarding the former, we work with a generic parameter estimator $\widehat{\vtheta}$ that satisfies

\begin{assumption}\label{ass:param est}
The parameter estimator $\widehat{\vtheta}$ satisfies $n^{\xi}\big(\widehat{\vtheta}-\vtheta^\circ\big)=O_{\p}(1)$, as $n\to\infty$, for some $\xi\in(0,1/2]$.
\end{assumption}

The standard case of $\sqrt{n}$-consistent estimators is covered by $\xi=1/2$. For some examples of such estimators in multivariate volatility models, we refer to \citet{FZ10}. When errors are heavy-tailed, Assumption~\ref{ass:param est} may only hold for $\xi<1/2$. E.g., in standard univariate GARCH models, the QMLE satisfies Assumption~\ref{ass:param est} with $\xi=1/2$ ($\xi<1/2$), when the innovations have finite (infinite) 4th moments \citep{HY03}. Thus, the generality afforded by Assumption~\ref{ass:param est} is not vacuous.

Next, we introduce some assumptions on the volatility model, i.e., $\mSigma_t$.
The $\mSigma_t$'s in sufficiently general volatility models often depend on the infinite past $(X_{t-1},Y_{t-1})^\prime$, $(X_{t-2}, Y_{t-2})^\prime$, $\ldots$. Therefore, to approximate the $\mSigma_t$'s, we use fixed artificial initial values $(\widehat{X}_{-\ell_n},\widehat{Y}_{-\ell_n})^\prime$, $(\widehat{X}_{-\ell_n-1}, \widehat{Y}_{-\ell_n-1})^\prime$, $\ldots$ in
\[
	\widehat{\mSigma}_t(\vtheta) = \mSigma_t\big((X_{t-1},Y_{t-1})^\prime,\ldots,(X_{-\ell_n+1}, Y_{-\ell_n+1})^\prime, (\widehat{X}_{-\ell_n}, \widehat{Y}_{-\ell_n})^\prime, (\widehat{X}_{-\ell_n-1}, \widehat{Y}_{-\ell_n-1})^\prime,\ldots;\vtheta\big),
\]
where $\vtheta\in\mTheta$.
This suggests to approximate $\mSigma_t$ by $\widehat{\mSigma}_t:=\widehat{\mSigma}_t(\widehat{\vtheta})=\diag\big(\widehat{\vsigma}_t(\widehat{\vtheta})\big)$, where $\widehat{\vsigma}_t(\widehat{\vtheta}) = \big(\widehat{\sigma}_{t,X}(\widehat{\vtheta}), \widehat{\sigma}_{t,Y}(\widehat{\vtheta})\big)^\prime$. 

We impose the following assumptions on the initialization and on the volatility model:

\begin{assumption}\label{ass:vola est}
For any $M>0$, there exists a neighborhood $\mathcal{N}(\vtheta^\circ)$ of $\vtheta^\circ$, and $p_{\ast}>0$, $q_{\ast}>0$ with $p_{\ast}^{-1} + q_{\ast}^{-1}=1$, such that for all $i=1,\ldots,\dim(\mTheta)$,
\begin{align*}
	\E\bigg[\sup_{\vtheta\in\mathcal{N}(\vtheta^\circ)}\Big\Vert \mSigma_t^{-1}(\vtheta)\frac{\partial\mSigma_t(\vtheta)}{\partial \theta_i}   \Big\Vert^{Mp_{\ast}}\bigg]&<\infty,\\
	\E\bigg[\sup_{\vtheta\in\mathcal{N}(\vtheta^\circ)}\Big\Vert \mSigma_t(\vtheta)\mSigma_t^{-1}(\vtheta^\circ)   \Big\Vert^{Mq_{\ast}}\bigg]&<\infty.
\end{align*}
\end{assumption}

\begin{assumption}\label{ass:vola init}
There exist some constants $C>0$ and $\rho\in(0,1)$, and some random variable $C_0>0$, such that for all $t\in\mathbb{N}$ it holds almost surely (a.s.) that
\begin{align*}
	\sup_{\vtheta\in\mTheta}\big\Vert\mSigma_t^{-1}(\vtheta)\big\Vert&\leq C,\\
	\sup_{\vtheta\in\mTheta}\big\Vert\mSigma_t(\vtheta) - \widehat{\mSigma}_t(\vtheta)\big\Vert&\leq C_0\rho^{t+\ell_n-1}.
\end{align*}
\end{assumption}

Assumption~\ref{ass:vola est} for $M=4$ is almost identical to assumption A8 in \citet{FJM17}. Together with Assumption~\ref{ass:param est}, it can be used to show that parameter estimation effects in the MES forecasts vanish asymptotically. At first sight, the H\"{o}lder conjugate exponents $p_{\ast}$ and $q_{\ast}$ in Assumption~\ref{ass:vola est} do not appear useful, since $M$ is allowed to be arbitrarily large. However, when (e.g.) $q_{\ast}=C/M$ for some finite constant $C$, one may accommodate an infinite $(C+\varepsilon)$th-moment of $\sup_{\vtheta\in\mathcal{N}(\vtheta^\circ)}\big\Vert \mSigma_t(\vtheta)\mSigma_t^{-1}(\vtheta^\circ) \big\Vert$, which, however, comes at the price of requiring moments of arbitrary order to exist for $\sup_{\vtheta\in\mathcal{N}(\vtheta^\circ)}\big\Vert \mSigma_t^{-1}(\vtheta)\partial\mSigma_t(\vtheta)/\partial \theta_i\big\Vert$.

Assumption~\ref{ass:vola init} is almost identical to assumptions A1 and A2 in \citet{FJM17}, and ensures that initialization effects vanish asymptotically at a suitable rate. Initialization effects may occur because $\mSigma_{t}$ often depends on the infinite past in $\mathcal{F}_{t-1}$, yet for estimation via $\widehat{\mSigma}_{t}$ only the truncated information set $\widehat{\mathcal{F}}_{t-1}=\sigma\big((X_{t-1}, Y_{t-1})^\prime,\ldots, (X_{-\ell_n+1}, Y_{-\ell_n+1})^\prime\big)$ is available.

Note that in Assumptions~\ref{ass:vola est}--\ref{ass:vola init} we have tacitly assumed that $\mSigma_t(\vtheta)$ is differentiable (in a neighborhood of the true parameter) and invertible, where the latter is equivalent to $\sigma_{t,X}(\vtheta)>0$ and $\sigma_{t,Y}(\vtheta)>0$, since $\mSigma_t(\vtheta)=\diag\big(\vsigma_t(\vtheta)\big)$. Of course, assuming positive volatilities is rather innocuous.

\begin{example}\label{ex:ex}
This example gives an instance of a model that satisfies Assumptions~\ref{ass:vola est}--\ref{ass:vola init}. Consider the ECCC--GARCH model of \citet{Jea98}, which models the squared volatilities $\vsigma_t^2=(\sigma_{t,X}^2, \sigma_{t,Y}^2)^\prime$ as
\begin{equation}\label{eq:vola ECCC}
	\vsigma_t^2=\vsigma_t^2(\vtheta) = \vomega + \sum_{j=1}^{\overline{p}}\mB_j \begin{pmatrix}X_{t-j}^2\\ Y_{t-j}^2\end{pmatrix} + \sum_{j=1}^{\overline{q}}\mGamma_j\vsigma^2_{t-j}(\vtheta),\qquad t\in\mathbb{Z},
\end{equation}
where $\vtheta=\big(\vomega^\prime, \vec^\prime(\mB_1),\ldots, \vec^\prime(\mB_{\overline{p}}), \vec^\prime(\mGamma_1), \ldots, \vec^\prime(\mGamma_{\overline{q}})\big)^\prime$. 
The classical CCC--GARCH model of \citet{Bol90} only allows for diagonal $\mB_j$'s and $\mGamma_j$'s, while non-diagonal matrices---and, hence, volatility spillovers---are accommodated only by ECCC--GARCH models.
Under the conditions of their Theorem~5.1, \citet{FJM17} show that a solution to the stochastic recurrence equations \eqref{eq:model} and \eqref{eq:vola ECCC} exists, and Assumptions~\ref{ass:vola est} and \ref{ass:vola init} are satisfied for 
\[
	\widehat{\vsigma}_t^2(\vtheta) = \vomega + \sum_{j=1}^{\overline{p}}\mB_j \begin{pmatrix}X_t^2\\ Y_t^2\end{pmatrix} + \sum_{j=1}^{\overline{q}}\mGamma_j\widehat{\vsigma}_{t-j}^2(\vtheta),\qquad t\geq1,
\]
with fixed initial values $(X_{-\ell_n}, Y_{-\ell_n})^\prime=(\widehat{X}_{-\ell_n}, \widehat{Y}_{-\ell_n})^\prime$, $(X_{-\ell_n-1}, Y_{-\ell_n-1})^\prime=(\widehat{X}_{-\ell_n-1}, \widehat{Y}_{-\ell_n-1})^\prime,\ldots$ and fixed initial $\widehat{\vsigma}_{-\ell_n}^2(\vtheta), \widehat{\vsigma}_{-\ell_n-1}^2(\vtheta),\ldots$.
\end{example}

As pointed out above, we also need to impose some regularity conditions on the tail of the $\vvarepsilon_t$ (to estimate $\theta_p$). Specifically, we assume that the following limit exists for all $(x,y)^\prime\in[0,\infty]^2\setminus\{(\infty,\infty)\}$:
\begin{equation}\label{eq:R}
	\lim_{s\to\infty}s\p\big\{1-F_{0}(\varepsilon_{t,X})\leq x/s,\ 1-F_{1}(\varepsilon_{t,Y})\leq y/s\big\}=:R(x,y).
\end{equation}
\citet{SS06} call $R(\cdot,\cdot)$ the (upper) \textit{tail copula}. Much like a copula, $R(x,y)$ only depends on the (extremal) dependence structure of $\vvarepsilon_t$, and not on the marginal distributions. Regarding the marginals, we assume heavy right tails in the sense that there exist $\gamma_i>0$, such that
\begin{equation}\label{eq:U}
	\lim_{s\to\infty}U_i(sx)/U_i(s)=x^{\gamma_i}\qquad\text{for all }x>0,\ i=0,1,
\end{equation}
where $U_i=(1/[1-F_i])^\leftarrow$ \citep{HF06}. This condition means that far out in the tail, the distribution can roughly be modeled as a Pareto distribution (for which \eqref{eq:U} holds even without the limit). Ever since the work of \citet{Bol87}, heavy-tailed innovations are standard ingredients of volatility models. For instance, the popular Student's $t_{\nu}$-distribution satisfies \eqref{eq:U} with $\gamma_i=1/\nu$. In EVT, $\gamma_i$ is known as the \textit{extreme value index}.

\subsection{MES estimator}\label{MES estimator}

To estimate $\theta_{p}$, we build on \citet{Cea15}. These authors use an extrapolation argument that is often applied in EVT, such as in estimating high quantiles \citep{Wei78}. The general idea is to first estimate the quantity of interest at a less extreme level (say, $\theta_{k/n}$ for $k/n\gg p$) and then, in a second step, to extrapolate to the desired level (say, $\theta_{p}$) by exploiting the tail shape. These arguments rely on $p=p_n$ tending to zero, as $n\to\infty$. 

Specifically, under \eqref{eq:R} and \eqref{eq:U} with $\gamma_1\in(0,1)$, \citet{Cea15} show that 
\begin{equation}\label{eq:MES}
	\lim_{p\downarrow0}\frac{\theta_{p}}{U_1(1/p)}=\int_{0}^{\infty}R(1, y^{-1/\gamma_1})\textup{d}y.
\end{equation}
Now, let $k=k_n$ be an \textit{intermediate sequence} of integers, such that $k\to\infty$ and $k/n\to0$, as $n\to\infty$. Then, for $n\to\infty$,
\begin{equation}\label{eq:(7+)}
	\theta_{p}\overset{\eqref{eq:MES}}{\sim}\frac{U_1(1/p)}{U_1(n/k)}\theta_{k/n}\overset{\eqref{eq:U}}{\sim}\Big(\frac{k}{np}\Big)^{\gamma_1}\theta_{k/n}.
\end{equation}
This relation suggests the following two-step procedure to estimate $\theta_{p}$. First, estimate the less extreme (``within-sample'') $\theta_{k/n}$ and, then, use the tail shape---characterized here by $\gamma_1$---to extrapolate to the desired (possibly ``beyond the sample'') $\theta_{p}$. 

One key difference to \citet{Cea15} is that the $\vvarepsilon_{t}$ are not available for estimation, but need to be approximated by the standardized residuals 
\[
	\widehat{\vvarepsilon}_t=\widehat{\vvarepsilon}_t(\widehat{\vtheta})=\widehat{\mSigma}_{t}^{-1}(X_t, Y_t)^\prime.
\]
To estimate $\theta_{k/n}$, we then use the non-parametric estimator
\[
	\widehat{\theta}_{k/n} = \frac{1}{k}\sum_{t=1}^{n}\widehat{\varepsilon}_{t,Y}^{+} I_{\big\{\widehat{\varepsilon}_{t,X}>\widehat{\varepsilon}_{(k+1),X}\big\}},
\]
where $\widehat{\varepsilon}_{(1),Z}\geq\ldots\geq \widehat{\varepsilon}_{(n),Z}$ denote the order statistics of $\widehat{\varepsilon}_{1,Z},\ldots,\widehat{\varepsilon}_{n,Z}$ ($Z\in\{X,Y\}$), such that $\widehat{\varepsilon}_{(k+1),X}$ estimates $F_0^{\leftarrow}(1-k/n)$ in $\theta_{k/n}$, and $I_{\{\cdot\}}$ denotes the indicator function. 

\begin{rem}\label{rem:trunc}
The estimator $\widehat{\theta}_{k/n}$ uses $\widehat{\varepsilon}_{t,Y}^{+}$ instead of $\widehat{\varepsilon}_{t,Y}$. Thus, it is in fact an estimator of $\theta_{k/n}^{+}=\E[\varepsilon_{t,Y}^{+}\mid \varepsilon_{t,X}>F_{0}^{\leftarrow}(1-k/n)]$. This is because the proofs closely exploit the relation that $\E[Z]=\int_{0}^{\infty}\p\{Z>x\}\D x$ for $Z\geq0$, such that
\begin{equation}\label{eq:(p.10)}
	\theta_{k/n}^{+}=\frac{n}{k}\int_{0}^{\infty}\p\{\varepsilon_{t,X}> U_{0}(n/k),\ \varepsilon_{t,Y}^{+}>y\}\D y.
\end{equation}
However, since \citet[Proof of Theorem~2]{Cea15} show that $\theta_{p}/\theta_{p}^{+}=1+o\big(1/\sqrt{k}\big)$, this does not impair the asymptotic validity of the estimator $\widehat{\theta}_p$ based on $\widehat{\theta}_{k/n}$.
\end{rem}

To estimate $\gamma_1$, we use the \citet{Hil75} estimator 
\[
	\widehat{\gamma}_1 = \frac{1}{k_1}\sum_{t=1}^{k_1}\log\big(\widehat{\varepsilon}_{(t),Y} / \widehat{\varepsilon}_{(k_1+1),Y}\big),
\]
where $k_1$ is another intermediate sequence of integers. Plugging the estimators $\widehat{\theta}_{k/n}$ and $\widehat{\gamma}_1$ into \eqref{eq:(7+)}, we obtain the desired estimator
\[
	\widehat{\theta}_{p} = \Big(\frac{k}{np}\Big)^{\widehat{\gamma}_1}\widehat{\theta}_{k/n}.
\]

To establish the asymptotic normality of $\widehat{\theta}_{p}$, we impose essentially the same regularity conditions as \citet{Cea15}. First, we specify the speed of convergence in \eqref{eq:R} via Assumption~\ref{ass:R} and that in \eqref{eq:U} via Assumption~\ref{ass:U}. Here and elsewhere, $x\wedge y=\min\{x,y\}$.

\begin{assumption}\label{ass:R}
There exist $\beta>\gamma_1$ and $\tau<0$ such that, as $s\to\infty$,
\[
	\sup_{\substack{x\in[1/2, 2]\\ y\in(0,\infty)}}\frac{\big|s\p\{1-F_{0}(\varepsilon_{t,X})\leq x/s,\ 1-F_{1}(\varepsilon_{t,Y})\leq y/s\}-R(x,y)\big|}{y^{\beta}\wedge 1}=O(s^{\tau}).
\]
\end{assumption}

\begin{assumption}\label{ass:U}
For $i=0,1$ there exist $\rho_i<0$ and an eventually positive or negative function $A_i(\cdot)$ such that, as $s\to\infty$, $A_i(sx)/A_i(s)\rightarrow x^{\rho_i}$ for all $x>0$ and, for any $x_0>0$,
\[
	\sup_{x\geq x_0}\Big|x^{-\gamma_i}\frac{U_i(sx)}{U_i(s)}-1\Big|=O\{A_i(s)\}.
\]
\end{assumption}

Assumptions~\ref{ass:R} and \ref{ass:U} provide second-order refinements of the convergences in \eqref{eq:R} and \eqref{eq:U}. They ensure that bias terms arising from extrapolation vanish asymptotically. Note that the more negative $\tau$ ($\rho_i$) in Assumption~\ref{ass:R} (Assumption~\ref{ass:U}), the better the approximation.

\begin{rem}
\begin{enumerate}
	\item[(i)] Replacing $s$ with $n/k$ in the probability in the numerator, Assumption~\ref{ass:R} requires $(n/k)\p\big\{\varepsilon_{t,X}>U_0(n/[kx]),\ \varepsilon_{t,Y}>U_1(n/[ky])\big\}$ to converge uniformly to its limit. In view of \eqref{eq:(p.10)} it is therefore sufficient to impose uniformity only in a neighborhood of $1$ for $x$ (which, following \citet{Cea15}, we take to be $[1/2,2]$ here). However, uniformity in $y$ over $(0,\infty)$ is required, as the integration in \eqref{eq:(p.10)} extends over all positive $y$-values. For a specific dependence structure of $(\varepsilon_{t,X}, \varepsilon_{t,Y})^\prime$, Assumption~\ref{ass:R} may be checked by drawing on the results in \citet[Section 4]{FHM15}; see also Remark~3 in that paper. Some specific distributions for which Assumption~\ref{ass:R} holds are also given by \citet[Sec.~3]{Cea15}. 
	
\item[(ii)] As pointed out by \citet[Remark~2]{Cea15}, Assumption~\ref{ass:R} excludes the case of asymptotically independent $(\varepsilon_{t,X}, \varepsilon_{t,Y})^\prime$, where $R\equiv0$ (to see this let $y=s$). This rules out Gaussian copulas, but covers $t$-copulas \citep{Hef00}, which seem to be empirically more relevant for financial data \citep{BDE03}. Building on \citet{CM20}, it may be possible to allow asymptotically independent innovations in MES estimation. This is, however, beyond the scope of the present paper. 
\end{enumerate}
\end{rem}

\begin{rem}
\begin{enumerate}
	\item[(i)]
\citet{Cea15} also impose Assumption~\ref{ass:U} on the tail of $\varepsilon_{t,Y}$. Together with $\sqrt{k_1}A_1(n/k_1)\to0$ (see Assumption~\ref{ass:k} below) it ensures that $\sqrt{k_1}(\widehat{\gamma}_1-\gamma_1)$ does not have any asymptotic bias terms; see Lemma~\ref{lem:gamma} or also \citet[Example~5.1.5]{HF06}.
 
	\item[(ii)]
	Note that \citet{Cea15} do not have to impose Assumption~\ref{ass:U} on the tail of $\varepsilon_{t,X}$, because the MES does not depend on its distribution. However, in estimating MES based on the estimated residuals $(\widehat{\varepsilon}_{t,X}, \widehat{\varepsilon}_{t,Y})^\prime$ we have to justify the replacement of the unobservable $\varepsilon_{t,X}$ by the feasible $\widehat{\varepsilon}_{t,X}$ even in the tails. To that end, we require a sufficiently well-behaved tail also of the $\varepsilon_{t,X}$. 

We stress that imposing Assumption~\ref{ass:U} for $\varepsilon_{t,X}$ (i.e., for $i=0$) is mainly a convenience. It can be replaced by any other condition ensuring the conclusion of Lemma~\ref{lem:r n pm} holds, as a careful reading of the proofs reveals. For instance, it can easily be shown that Lemma~\ref{lem:r n pm} remains valid, e.g., for light-tailed (standardized) exponentially distributed $\varepsilon_{t,X}$. However, ever since the work of \citet{Bol87}, heavy-tailed errors satisfying Assumption~\ref{ass:U} (such as $t_{\nu}$-distributed errors with $\gamma_i=1/\nu$ and $\rho_i=-2$) are regarded as more suitable in volatility modeling. We refer to \citet[Examples~1--3]{HJ11} for further heavy-tailed distributions with corresponding values for $\gamma_i$ and $\rho_i$.
\end{enumerate}
\end{rem}

We mention that the problem of estimating a MES based on approximated conditioning variables (here the $\widehat{\varepsilon}_{t,X}$) also appears in the work of \citet{DP18}, who deal with unconditional MES estimation for i.i.d.~random variables. They have to ensure that replacing their (latent) conditioning variable $Z_j$ with some feasible $\widetilde{Z}_j$ has no asymptotic impact. Instead of imposing a distributional assumption (as we do) on the conditioning variables, \citet{DP18} assume that the $\widetilde{Z}_j$'s are estimated from an initial pre-sample of length $n_2$ to rule out any asymptotic effects. Their Assumption~1~(a.3) (with $p_0=q_0=1/2$ in their notation) then requires $n_1=o(n_2^{(1-\varepsilon)/2})$ for the actual estimation sample size $n_1$. This allows them to justify the replacement of the latent $Z_j$ with the $\widetilde{Z}_j$ without imposing distributional assumptions, as we do here. However, adopting a similar approach in our present time series context would be highly unnatural.

We need two additional assumptions.

\begin{assumption}\label{ass:k}
As $n\to\infty$, 
\begin{align*}
&\sqrt{k}A_0(n/k)\rightarrow0,\quad \sqrt{k_1}A_1(n/k_1)\rightarrow0,\\
&k=O(n^{\alpha})\quad\text{for some }\alpha<-2\tau/(-2\tau+1)\wedge 2\gamma_1\rho_1/(2\gamma_1\rho_1+\rho_1-1),\\
&k=o(p^{2\tau(1-\gamma_1)}),\\
& \min\big\{\sqrt{k}, \sqrt{k_1}/\log(d_n)\big\} = O(n^{\widetilde{\alpha}})\quad\text{and}\quad \sqrt{k_1}=O(n^{\widetilde{\alpha}})\quad\text{for some }\widetilde{\alpha}<\xi,
\end{align*}
with $A_0(\cdot)$ and $A_1(\cdot)$ from Assumption~\ref{ass:U}, and $\xi\in(0, 1/2]$ from Assumption~\ref{ass:param est}.
\end{assumption}

\begin{assumption}\label{ass:mom}
$\E|\varepsilon_{t,Y}^{-}|^{1/\gamma_1}<\infty$.
\end{assumption}

The purpose of Assumption~\ref{ass:k} is to restrict the speed of divergence of $k$ and $k_1$. While this is obvious for most items, it is less clear for the first conditions involving $A_0(\cdot)$ and $A_1(\cdot)$. However, \citet[p.~77]{HF06} show that these conditions imply that $k=o(n^{-2\rho_0/(1-2\rho_0)})$ and $k_1=o(n^{-2\rho_1/(1-2\rho_1)})$, respectively. While large values of the intermediate sequences $k$ and $k_1$ imply a small asymptotic variance, a bias is incurred by using possibly ``non-tail'' observations in the estimates. Therefore, a bound on the growth of $k$ and $k_1$ is required for asymptotically unbiased estimates. The requirement that $k=o(p^{2\tau(1-\gamma_1)})$ together with Assumption~\ref{ass:mom} is only used to show that $\theta_p/\theta_p^{+}=1+o(1/\sqrt{k})$; cf.~Remark~\ref{rem:trunc}. This relation ensures that $\widehat{\theta}_{p}$, which actually estimates $\theta^{+}_{p}$, also estimates $\theta_{p}$. The final condition in Assumption~\ref{ass:k} ensures that the parameter estimator (which is $n^{\xi}$-consistent) converges sufficiently fast relative to our MES estimator (which is $\min\big\{\sqrt{k}, \sqrt{k_1}/\log(d_n)\big\}$-consistent by Proposition~\ref{prop:MES}). In the standard case where $\xi=1/2$, this condition is redundant, because $\min\big\{\sqrt{k}, \sqrt{k_1}/\log(d_n)\big\}\leq\sqrt{k}=o(n^{1/2})$ and $\sqrt{k_1}=o(n^{1/2})$ as $k$ and $k_1$ are intermediate sequences.

\section{Asymptotic normality of MES forecasts}\label{Asymptotic normality of MES forecasts}

With the MES estimator of the previous subsection, our MES forecast becomes
\begin{equation}\label{eq:MES forc}
	\widehat{\theta}_{n,p}=\widehat{\sigma}_{n+1,Y}\widehat{\theta}_{p},
\end{equation}
where $\widehat{\sigma}_{n+1,Y}=\widehat{\sigma}_{n+1,Y}(\widehat{\vtheta})$.
We can now state our first main theoretical result.

\begin{thm}\label{thm:main result}
Let $(X_t, Y_t)^\prime$ be a strictly stationary solution to \eqref{eq:model} that is measurable with respect to the sigma-field generated by $\{\vvarepsilon_t,\vvarepsilon_{t-1},\ldots\}$. Suppose Assumptions~\ref{ass:param est}--\ref{ass:mom} hold, and $\gamma_1\in(0,1/2)$. Suppose further that $d_n:=k/(np)\geq1$, $r:=\lim_{n\to\infty}\sqrt{k}\log(d_n)/\sqrt{k_1}\in[0,\infty]$, $q:=\lim_{n\to\infty}k_1/k\in(0,\infty)$ and $\lim_{n\to\infty}\log(d_n)/\sqrt{k_1}=0$. Moreover, suppose that the truncation sequence $\ell_n$ satisfies $\ell_n/\log n\to\infty$. Then, as $n\to\infty$,
\[
	\min\big\{\sqrt{k}, \sqrt{k_1}/\log(d_n)\big\}\log\Bigg(\frac{\widehat{\theta}_{n,p}}{\theta_{n,p}}\Bigg)\overset{d}{\longrightarrow}\begin{cases}\Theta+r\Gamma,& \text{if }r\leq1,\\
	(1/r)\Theta+\Gamma,&\text{if }r>1,\end{cases}
\]
where $\Theta$ and $\Gamma$ are zero mean Gaussian random variables with 
\begin{align*}
\Var(\Theta) &= \gamma_1^2 - 1-b^2\int_{0}^{\infty}R(1, s)\D s^{-2\gamma_1},\qquad b=1/\int_{0}^{\infty}R(1, s)\D s^{-\gamma_1},\\
\Var(\Gamma) &= \gamma_1^2,\\
\Cov(\Gamma,\Theta) &= \frac{\gamma_1}{\sqrt{q}}\Big(1-\gamma_1+\frac{b}{q^{\gamma_1}}\Big)R(1,q)\\
&\hspace{2cm}-\frac{\gamma_1}{\sqrt{q}}\int_{0}^{q}\Big[(1-\gamma_1)+bs^{-\gamma_1}\big\{1-\gamma_1-\gamma_1\log(s/q)\big\}\Big]R(1, s)s^{-1}\D s.
\end{align*}
\end{thm}

\begin{proof}
See Appendix~\ref{sec:Proof of Theorem 1}.
\end{proof}

The assumptions of Theorem~\ref{thm:main result} can be roughly divided into two parts. Assumptions~\ref{ass:param est}--\ref{ass:vola init}, which are similar to conditions maintained by \citet{FJM17}, ensure that the innovations of the volatility model can be recovered from the observations with sufficient precision (see Proposition~\ref{prop:approx}). Assumptions~\ref{ass:R}--\ref{ass:mom}, which closely resemble the conditions in \citet{Cea15}, imply the asymptotic normality of the MES estimator for the innovations. Then, Assumptions~\ref{ass:param est}--\ref{ass:mom} jointly ensure that the MES estimator $\widehat{\theta}_p$ based on the filtered residuals is also asymptotically normal (see Proposition~\ref{prop:MES}). Here, the requirement that the truncation sequence $\ell_n$ be sufficiently long (i.e., $\ell_n/\log n\to\infty$) together with Assumption~\ref{ass:vola init} ensures that initialization effects are negligible in the limit.

The case most often considered in EVT is that where $d_n\to\infty$; see, e.g., \citet[Theorem~4.3.1]{HF06} for high quantile estimation. When additionally $k\asymp k_1$, we have that $r=\infty$, implying that $\Gamma$ is the asymptotic limit. The proof of Theorem~\ref{thm:main result} shows that $\widehat{\gamma}_1$ consistently estimates $\gamma_1$; see Lemma~\ref{lem:gamma} in Appendix~\ref{Proof of Proposition 2}. Hence, a feasible asymptotic $(1-\iota)$-confidence interval for $\theta_{n,p}$ in case $r=\infty$ is given by
\begin{equation}\label{eq:CI}
	\Big[\widehat{\theta}_{n,p}\exp\big\{\mp\Phi^{-1}(1-\iota/2)\widehat{\gamma}_1\log(d_n)/\sqrt{k_1}\big\}\Big],
\end{equation}
where $\Phi^{-1}(\cdot)$ denotes the inverse of the standard normal d.f. 

This is not a ``classical'' confidence interval for some unknown parameter because $\theta_{n,p}$ is random. However, it can be interpreted as such, since (by Theorem~\ref{thm:main result} and Lemma~\ref{lem:gamma})
\begin{multline*}
	\p\Big\{\widehat{\theta}_{n,p}\exp\big\{-\Phi^{-1}(1-\iota/2)\widehat{\gamma}_1\log(d_n)/\sqrt{k_1}\big\} \leq \theta_{n,p} \\
	\leq \widehat{\theta}_{n,p}\exp\big\{\Phi^{-1}(1-\iota/2)\widehat{\gamma}_1\log(d_n)/\sqrt{k_1}\big\}\Big\}\underset{(n\to\infty)}{\longrightarrow}1-\iota.
\end{multline*}
Such a straightforward interpretation of confidence intervals for random quantities is often not possible \citep{BHS21}. Yet, it is possible here, because in Theorem~\ref{thm:main result} the asymptotic estimation uncertainty comes solely from the non-random $\theta_p$ component in
\[
	\log(\widehat{\theta}_{n,p}/\theta_{n,p}) = \log(\widehat{\sigma}_{n+1,Y}/\sigma_{n+1,Y}) + \log(\widehat{\theta}_p/\theta_p).
\]
Specifically, the proof of Theorem~\ref{thm:main result} shows that parameter estimation effects vanish because volatility in $\widehat{\theta}_{n,p}$ can be estimated $n^{\xi}$-consistently, yet the unconditional MES estimate has a slower rate of convergence, thus dominating asymptotically. In the context of EVT-based VaR and ES forecasting, this was noted before by, e.g., \citet{Cea07}, \citet{MYT18} and \citet{Hog18+}. 

Inference when $r<\infty$ remains an unsolved issue, even for unconditional MES estimation; see \citet{Cea15} and \citet{DP18}. However, the case $k\asymp k_1$ and $d_n\to\infty$ seems to be the case of most practical interest, because $d_n\to\infty$ corresponds to situations of strong extrapolation, where $k/n\gg p$. Furthermore, we demonstrate the good finite-sample coverage of \eqref{eq:CI} in simulations in Section~\ref{Simulations}. Nonetheless, it may be possible to explicitly deal with the case $r<\infty$. This may be possible by using self-normalization as in \citet{Hog18+} or by employing suitable bootstrap methods along the lines of \citet{LPS23}. We leave these challenging extensions for future research.

\section{Higher-Dimensional Extensions}\label{Higher-Dimensional Extensions}

One desirable property of MES as a systemic risk measure is its additivity property. To illustrate, suppose that $Y_{t,1},\ldots, Y_{t,D}$ denote the losses of all trading desks of a business unit. The weighted losses of the business unit then sum to $X_t=\sum_{d=1}^{D}w_{t-1,d}Y_{t,d}$, where the weights $w_{t-1,d}$ are determined by how much capital is allocated to each trading desk in advance and, thus, are known at time $t-1$. The total riskiness of the business unit, as measured by the ES, can then be decomposed as $\E[X_t\mid X_t>\VaR_t(p),\ \mathcal{F}_{t-1}]=\sum_{d=1}^{D}w_{t-1,d}\E[Y_{t,d}\mid X_t>\VaR_t(p),\ \mathcal{F}_{t-1}]$. In allocating capital among the trading desks, one may want to ensure an equal risk contribution of each trading desk, such that $w_{t-1,1}\E[Y_{t,1}\mid X_t>\VaR_t(p),\ \mathcal{F}_{t-1}]=\ldots=w_{t-1,D}\E[Y_{t,D}\mid X_t>\VaR_t(p),\ \mathcal{F}_{t-1}]$. Therefore, it becomes important to develop tools for the joint inference on different MES forecasts. 

Clearly, drawing inferences on many MES forecasts jointly is also important in other contexts. For instance, suppose the $Y_{t,1},\ldots, Y_{t,D}$ denote individual losses of all banks in the financial system. Then, the regulator seeks to control the system's total risk as measured by the ES $\E[X_t\mid X_t>\VaR_t(p),\ \mathcal{F}_{t-1}]$; see \citet{QZ21}. Since $\E[X_t\mid X_t>\VaR_t(p),\ \mathcal{F}_{t-1}]=\sum_{d=1}^{D}w_{t-1,d}\E[Y_{t,d}\mid X_t>\VaR_t(p),\ \mathcal{F}_{t-1}]=:\sum_{d=1}^{D}w_{t-1,d}\theta_{t-1,p,d}$, it becomes clear that regulators should take into account the estimation risk of the individual MES forecasts.

To enable joint hypothesis testing, we consider a high-dimensional extension of model \eqref{eq:model}, viz.
\begin{equation}\label{eq:model h}
	(X_t, Y_{t,1},\ldots,Y_{t,D})^\prime= \mSigma_t\vvarepsilon_t.
\end{equation}
We take the diagonal matrix $\mSigma_t=\diag(\vsigma_t)$ with $\vsigma_t=(\sigma_{t,X}, \sigma_{t,Y_1},\ldots,\sigma_{t,Y_D})^\prime$ to be measurable with respect to $\mathcal{F}_{t}=\sigma\big((X_t, Y_{t,1},\ldots,Y_{t,D})^\prime, (X_{t-1}, Y_{t-1,1},\ldots,Y_{t-1,D})^\prime,\ldots\big)$, and the $\vvarepsilon_t=(\varepsilon_{t,X}, \varepsilon_{t,Y_1},\ldots,\varepsilon_{t,Y_D})^\prime$ to be independent of $\mathcal{F}_{t-1}$ and i.i.d.~with mean zero, unit variance and correlation matrix $\mR$. We again assume the innovations $\vvarepsilon_t$ to have a continuous d.f.

To forecast MES, we use the same estimator as before, i.e., $\widehat{\theta}_{n,p,d}=\widehat{\sigma}_{n+1,Y_d}\widehat{\theta}_{p,d}$ with
\[
	\widehat{\theta}_{p,d}=\Big(\frac{k}{np}\Big)^{\widehat{\gamma}_{d}}\widehat{\theta}_{k/n,d},
\]
where $\widehat{\gamma}_{d}= \frac{1}{k_d}\sum_{t=1}^{k_d}\log\big(\widehat{\varepsilon}_{(t),Y_d} / \widehat{\varepsilon}_{(k_d+1),Y_d}\big)$ and $\widehat{\theta}_{k/n,d}=\frac{1}{k}\sum_{t=1}^{n}\widehat{\varepsilon}_{t,Y_d}^{+} I_{\big\{\widehat{\varepsilon}_{t,X}>\widehat{\varepsilon}_{(k+1),X}\big\}}$ are defined in the expected way. 

To derive the asymptotic limit of $\big(\widehat{\theta}_{n,p,1}, \ldots, \widehat{\theta}_{n,p,D}\big)^\prime$, Assumptions~\ref{ass:param est}--\ref{ass:vola init} do not have to be changed. The remaining Assumptions~\ref{ass:R}--\ref{ass:mom} have to be generalized slightly as follows. To that end, we extend the notation in the obvious way. E.g., we denote the d.f.~of $Y_{t,d}$ by $F_d(\cdot)$, and set $U_d=(1/[1-F_d])^\leftarrow$. The extreme value index of the $Y_{t,d}$ is denoted by $\gamma_d$.

\begin{assumptionA}\label{ass:R*}
For all $d=1,\ldots,D$, there exist $\beta_d>\gamma_d$, $\tau_d<0$ and $R_d(\cdot,\cdot)$ such that, as $s\to\infty$,
\[
	\sup_{\substack{x\in[1/2, 2]\\ y\in(0,\infty)}}\frac{\big|s\p\{1-F_{0}(\varepsilon_{t,X})\leq x/s,\ 1-F_{d}(\varepsilon_{t,Y_d})\leq y/s\}-R_d(x,y)\big|}{y^{\beta_d}\wedge 1}=O(s^{\tau_d}).
\]
Moreover, for all $i,j=1,\ldots,D$ there exists a function $R_{i,j}(\cdot,\cdot)$, such that for all $x,y\in[0,\infty]^2\setminus\{(\infty,\infty)\}$,
\begin{equation}\label{eq:Rij}
	\lim_{s\to\infty}s\p\{1-F_{i}(\varepsilon_{t,Y_{i}})\leq x/s,\ 1-F_{j}(\varepsilon_{t,Y_{j}})\leq y/s\}=R_{i,j}(x,y).
\end{equation}
\end{assumptionA}

\begin{assumptionA}\label{ass:U*}
For all $d=0,\ldots,D$ there exist $\rho_d<0$ and an eventually positive or negative function $A_d(\cdot)$ such that, as $s\to\infty$, $A_d(sx)/A_d(s)\rightarrow x^{\rho_d}$ for all $x>0$ and, for any $x_0>0$,
\[
	\sup_{x\geq x_0}\Big|x^{-\gamma_d}\frac{U_d(sx)}{U_d(s)}-1\Big|=O\{A_d(s)\}.
\]
\end{assumptionA}

\begin{assumptionA}\label{ass:k*}
As $n\to\infty$, for each $d=1,\ldots,D$
\begin{align*}
&\sqrt{k}A_0(n/k)\rightarrow0,\quad \sqrt{k_d}A_1(n/k_d)\rightarrow0,\\
&k=O(n^{\alpha})\quad\text{for some }\alpha<-2\tau_d/(-2\tau_d+1)\wedge 2\gamma_d\rho_d/(2\gamma_d\rho_d+\rho_d-1),\\
&k=o(p^{2\tau_d(1-\gamma_d)}),\\
& \min\big\{\sqrt{k}, \sqrt{k_d}/\log(d_n)\big\} = O(n^{\widetilde{\alpha}})\quad\text{and}\quad \sqrt{k_d}=O(n^{\widetilde{\alpha}})  \quad\text{for some }\widetilde{\alpha}<\xi,
\end{align*}
with $A_0(\cdot)$ and $A_d(\cdot)$ from Assumption~\ref{ass:U*}*, and $\xi\in(0, 1/2]$ from Assumption~\ref{ass:param est}.
\end{assumptionA}

\begin{assumptionA}\label{ass:mom*}
$\E|\varepsilon_{t,Y_d}^{-}|^{1/\gamma_d}<\infty$ for all $d=1,\ldots,D$.
\end{assumptionA}

The only non-trivial extension of Assumptions~\ref{ass:R}--\ref{ass:mom} relative to Assumptions~\ref{ass:R*}*--\ref{ass:mom*}* is condition \eqref{eq:Rij}, which is closely related to \eqref{eq:R}. It is needed to derive the joint convergence of $(\widehat{\gamma}_1,\ldots,\widehat{\gamma}_D)^\prime$, and the limit function $R_{i,j}(\cdot,\cdot)$ features prominently in its asymptotic limit and that of Theorem~\ref{thm:main result2}. For simplicity, Theorem~\ref{thm:main result2} only considers the case where $r_d:=\lim_{n\to\infty}\sqrt{k}\log(d_n)/\sqrt{k_d}=\infty$ for all $d=1,\ldots,D$, which implies confidence intervals of the form in \eqref{eq:CI}.

\begin{thm}\label{thm:main result2}
Let $(X_t, Y_{t,1},\ldots,Y_{t,D})^\prime$ be a strictly stationary solution to \eqref{eq:model h} that is measurable with respect to the sigma-field generated by $\{\vvarepsilon_t,\vvarepsilon_{t-1},\ldots\}$.
Suppose Assumptions~\ref{ass:param est}--\ref{ass:vola init} and Assumptions~\ref{ass:R*}*--\ref{ass:mom*}* hold, and $\gamma_d\in(0,1/2)$ for all $d=1,\ldots,D$. Suppose further that $d_n:=k/(np)\geq1$, $r_d=\lim_{n\to\infty}\sqrt{k}\log(d_n)/\sqrt{k_1}=\infty$, $q_d:=\lim_{n\to\infty}k_d/k\in(0,\infty)$ and $\lim_{n\to\infty}\log(d_n)/\sqrt{k_d}=0$ for all $d=1,\ldots,D$. Moreover, suppose that the truncation sequence $\ell_n$ satisfies $\ell_n/\log n\to\infty$. Then, as $n\to\infty$,
\[
	\Bigg(\frac{\sqrt{k_1}}{\log(d_n)}\log\Bigg(\frac{\widehat{\theta}_{n,p,1}}{\theta_{n,p,1}}\Bigg),\ldots, \frac{\sqrt{k_D}}{\log(d_n)}\log\Bigg(\frac{\widehat{\theta}_{n,p,D}}{\theta_{n,p,D}}\Bigg)\Bigg)^\prime\overset{d}{\longrightarrow}(\Gamma_1,\ldots,\Gamma_D)^\prime,
\]
where $(\Gamma_1,\ldots, \Gamma_D)^\prime$ is a zero mean Gaussian random vector with variance-covariance matrix $\mSigma=(\sigma_{i,j})_{i,j=1,\ldots,D}$, 
\begin{equation*}
	\sigma_{i,j}:=\Cov(\Gamma_i,\Gamma_j) = \frac{\gamma_i\gamma_j}{\sqrt{q_i q_j}}\frac{ R_{i,j}(q_i,q_j)+R_{i,j}(q_j,q_i)}{2}.
\end{equation*}
\end{thm}

\begin{proof}
See Appendix~\ref{sec:Proof of Theorem 2}. 
\end{proof}

For joint inference on the MES forecasts, we have to estimate $\mSigma$, i.e., the $\gamma_d$'s, $q_d$'s, $R_{i,j}(q_i,q_j)$'s and $R_{i,j}(q_j,q_i)$'s. We propose to estimate $R_{i,j}(q_i,q_j)$ and $R_{i,j}(q_j,q_i)$ via
\begin{align*}
	\widehat{R}_{i,j}(q_i,q_j)&=\frac{1}{k}\sum_{t=1}^{n}I_{\big\{\widehat{\varepsilon}_{t,Y_i}>\widehat{\varepsilon}_{(k_i+1),Y_i},\ \widehat{\varepsilon}_{t,Y_j}>\widehat{\varepsilon}_{(k_j+1),Y_j}\big\}},\\
		\widehat{R}_{i,j}(q_j,q_i)&=\frac{1}{k}\sum_{t=1}^{n}I_{\big\{\widehat{\varepsilon}_{t,Y_i}>\widehat{\varepsilon}_{(k_j+1),Y_i},\ \widehat{\varepsilon}_{t,Y_j}>\widehat{\varepsilon}_{(k_i+1),Y_j}\big\}}.
\end{align*}
The next proposition shows that these estimates are consistent:
\begin{prop}\label{prop:cons}
Under the conditions of Theorem~\ref{thm:main result2}, it holds for all $i,j=1,\ldots,D$ that $\widehat{R}_{i,j}(q_i,q_j)\overset{\p}{\longrightarrow}R_{i,j}(q_i,q_j)$ and $\widehat{R}_{i,j}(q_j,q_i)\overset{\p}{\longrightarrow}R_{i,j}(q_j,q_i)$, as $n\to\infty$.
\end{prop}

\begin{proof}
See Appendix~\ref{sec:Proof of Proposition 1}.
\end{proof}

Since the $q_d$'s appearing in $\sigma_{i,j}$ can easily be ``estimated'' via $k_d/k$, it remains to estimate the extreme value indices $\gamma_d$. These can, however, be consistently estimated via $\widehat{\gamma}_d$; see Lemma~\ref{lem:gamma}. In sum, the asymptotic variance-covariance matrix $\mSigma$ from Theorem~\ref{thm:main result2} may be estimated consistently via $\widehat{\mSigma}=(\widehat{\sigma}_{i,j})_{i,j=1,\ldots,D}$ with typical element 
\begin{equation*}
\widehat{\sigma}_{i,j}=k\frac{\widehat{\gamma}_i\widehat{\gamma}_j}{\sqrt{k_i k_j}}\frac{\widehat{R}_{i,j}(q_i,q_j) + \widehat{R}_{i,j}(q_j,q_i)}{2}.
\end{equation*} 
Note that for $i=j=d$, we get---as expected from Theorem~\ref{thm:main result}---that $\widehat{\sigma}_{d,d}=\widehat{\gamma}_d^2$.

We mention that \citet{DP18} also consider MES estimation (albeit in an i.i.d.~static setting) in a higher-dimensional framework. However, they focus on limit theory for individual MES estimates. This contrasts with our joint convergence result in Theorem~\ref{thm:main result2}, with appertaining inference tools provided by Proposition~\ref{prop:cons} and Lemma~\ref{lem:gamma}.

As outlined above, in applications one may want to test the equality of (value-weighted) risk contributions by several trading desks. Similarly, one may want to test the equality of the risk contributions by banks in a financial system. This latter application is further explored in Section~\ref{Empirical Application}. In each case, the null is that 
\[
	H_0\ :\ w_{n,1}\theta_{n,p,1}=\ldots=w_{n,D}\theta_{n,p,D}.
\]

We test this null by comparing each of $\log(w_{n,1}\widehat{\theta}_{n,p,1}),\ldots,\log(w_{n,D-1}\widehat{\theta}_{n,p,D-1})$ with the ``average'' forecast $(1/D)\sum_{d=1}^{D}\log(w_{n,d}\widehat{\theta}_{n,p,d})$. Then, any ``large'' difference is evidence against the null. With the $(D-1)\times D$-matrix
\[
	\mT=\begin{pmatrix}1 &   & 0 &0\\
	 & \ddots & & \vdots\\
	0 & & 1 & 0		
	\end{pmatrix} - (1/D)
	\begin{pmatrix}1 & \ldots &1 \\
	 \vdots & \ddots & \vdots\\
	1 & \ldots& 1 		
	\end{pmatrix}
\]
this suggests using
\[
	\mT\begin{pmatrix} \log(w_{n,1}\widehat{\theta}_{n,p,1})\\ \vdots\\ \log(w_{n,D}\widehat{\theta}_{n,p,D})\end{pmatrix}=\begin{pmatrix} \log(w_{n,1}\widehat{\theta}_{n,p,1})-(1/D)\sum_{d=1}^{D}\log(w_{n,d}\widehat{\theta}_{n,p,d})\\ \vdots\\ \log(w_{n,D-1}\widehat{\theta}_{n,p,D-1})-(1/D)\sum_{d=1}^{D}\log(w_{n,d}\widehat{\theta}_{n,p,d})\end{pmatrix}
\]
in a Wald-type test. Specifically, we use the test statistic
\begin{multline*}
	\mathcal{T}_n=k\begin{pmatrix} \log(w_{n,1}\widehat{\theta}_{n,p,1}),\ldots,\log(w_{n,D}\widehat{\theta}_{n,p,D})\end{pmatrix}\mT^\prime\mL^\prime\\
	\times\mL \mT \begin{pmatrix} \log(w_{n,1}\widehat{\theta}_{n,p,1}),\ldots,\log(w_{n,D}\widehat{\theta}_{n,p,D})\end{pmatrix}^\prime
\end{multline*}
with $\mL$ the lower triangular matrix from the Cholesky decomposition satisfying $\mL\mL^\prime=(\mT\widehat{\mSigma}\mT^\prime)^{-1}$, where the inverse is assumed to exist. 

\begin{cor}\label{cor:test}
Suppose that the conditions of Theorem~\ref{thm:main result2} hold, and that $\mT\mSigma\mT^\prime$ is positive definite. Then, it holds under $H_0$ that $\mathcal{T}_n\overset{d}{\longrightarrow}\chi_{D-1}^2$, as $n\to\infty$.
\end{cor}

\begin{proof}
Combine the continuous mapping theorem with Theorem~\ref{thm:main result2}, Proposition~\ref{prop:cons} and the consistency of $\widehat{\gamma}_d$ for $d=1,\ldots,D$ (from Lemma~\ref{lem:gamma}).
\end{proof}

\section{Simulations}\label{Simulations}

Here, we investigate the coverage of the confidence interval in \eqref{eq:CI} for $\iota=0.05$. We do so for $n\in\{500,\ 1000\}$ and $p\in\{1\%,\ 0.5\%,\ 0.1\%,\ 0.05\%,\ 0.01\%,\ 0.005\%,\ 0.001\%\}$. Throughout, we use 1,000 replications and we clip off the first $\ell_n=10$ residuals to reduce the impact of initialization effects on our MES estimator. Following \citet[Footnote~6]{QZ21}, we use identical $k=k_1$ in the simulations and the empirical application. Specifically, we set $k=k_1=\lfloor 0.1\cdot \log(n)^4 \rfloor$ to satisfy Assumption~\ref{ass:k}. Since the probabilities $p$ are quite small relative to the sample sizes $n$, our confidence intervals---valid when $d_n=k/(np)$ tends to infinity---should provide reasonable coverage.

We simulate from the simple CCC--GARCH model $(X_t, Y_t)^\prime=\diag(\vsigma_t)\vvarepsilon_t$, where $\vsigma_t=(\sigma_{t,X}, \sigma_{t,Y})^\prime$ with
\begin{equation}\label{eq:model params}
\begin{split}
		\sigma_{t,X}^2 &= 0.001 + 0.2 \cdot X_t^2 + 0.75 \cdot\sigma_{t-1,X}^2,\\
		\sigma_{t,Y}^2 &= 0.001 + 0.1 \cdot Y_t^2 + 0.85 \cdot\sigma_{t-1,Y}^2.
\end{split}
\end{equation}
The parameter values of the volatility equations are chosen to resemble typical estimates obtained for financial data. The innovations $\vvarepsilon_t$ are i.i.d.~draws from a $t$-copula with $\nu$ degrees of freedom and correlation coefficient $\rho_{X,Y}=0.95$. The marginals of $\vvarepsilon_t$ are from a (standardized and symmetrized) $\Burr(a, b)$-distribution with d.f.~$F(x)=1-(1+x^b)^{-a}$, $x>0$, $a,b>0$. Hence, $\varepsilon_{t,Z}\overset{d}{=}R_tB_t/\sqrt{\E[B_t^2]}$ for $Z\in\{X,Y\}$, where $R_t$ are Rademacher random variables (i.e., equal to $\pm1$ with probability 1/2), independent of the $B_t\sim\Burr(a, b)$.

We choose the marginal Burr distribution because it allows to vary the quality of the approximation in Assumption~\ref{ass:U} via the parameters $a$ and $b$ without changing the extreme value index $\gamma$. Example~2 in \citet{HJ11} shows that the extreme value index of a $\Burr(a, b)$-distribution is given by $1/(ab)$ and its second-order parameter by $-b$. Thus, we have that $\gamma_0=\gamma_1=1/(ab)$ and $\rho_0=\rho_1=-b$ in the notation of Assumption~\ref{ass:U}. This implies that the larger $b$, the faster the convergence of $U_i(sx)/U(s)$ to the Pareto-type limit $x^{\gamma_i}$ takes place in Assumption~\ref{ass:U}. By suitable choices of $a$ and $b$, this allows us to assess the implications of a better tail approximation on the precision of our MES estimator, while keeping the tail index constant. Specifically, we choose $(a,b)\in\{(0.25, 20),\ (0.2, 25)\}$ to always obtain $\gamma_0=\gamma_1=1/5$. However, the Pareto approximation is better for $(a,b)=(0.2, 25)$ because $b$ is larger.

The choice of the $t$-copula for $\vvarepsilon_t$ implies that Assumption~\ref{ass:R} is satisfied with $\tau=-2/\nu$, $\beta=1+2/\nu$ and $R(x,y)=x\overline{F}_{\nu+1}\Big(\frac{(x/y)^{1/\nu}-\rho_{X,Y}}{\sqrt{1-\rho_{X,Y}^2}}\sqrt{\nu+1}\Big) + y\overline{F}_{\nu+1}\Big(\frac{(y/x)^{1/\nu}-\rho_{X,Y}}{\sqrt{1-\rho_{X,Y}^2}}\sqrt{\nu+1}\Big)$, where $\overline{F}_{\nu+1}(\cdot)=1-F_{\nu+1}(\cdot)$ denotes the survivor function of the $t_{\nu+1}$-distribution; see \citet[Remark~3 and Section~4.1]{FHM15}. A more negative $\tau$ (smaller degrees of freedom $\nu$ of the $t$-copula) implies a better approximation in Assumption~\ref{ass:R}. Our copula construction allows us to vary the dependence structure (in particular the quality of the approximation in Assumption~\ref{ass:R} as a function of $\nu$) without changing the marginals of the innovations. Specifically, we choose $\nu\in\{3,5\}$.

To estimate the GARCH parameters in \eqref{eq:model params}, we use the quasi-maximum likelihood estimator (QMLE). Since we choose $a$ and $b$ to give $\gamma_0=\gamma_1=1/(ab)=1/5$, the innovations have finite fourth moments, implying $\sqrt{n}$-consistency of the QMLE \citep{FZ10}. Thus, Assumption~\ref{ass:param est} is met for $\xi=1/2$. By Example~\ref{ex:ex}, Assumptions~\ref{ass:vola est} and \ref{ass:vola init} are also met.

\begin{table}[t!]
	\centering
		\begin{tabular}{lllrrrrrrr}
			\toprule
  $n$ 	&$\nu$& $(a,b)$ 			&$p$			& Bias	 & RMSE	 	 & $\substack{\text{RMSE}\\ \text{ML}}$  & $\substack{\text{RMSE}\\ \text{NP}}$	& Length  & Coverage  \\
  \midrule	                                                                              
	$500$	&	$3$	& $(0.25, 20)$  & 1\%		  &	  0.3  &    2.6  &    2.1   &    4.8  &    6.6  &  83.9    \\
				&			&								&	0.5\%		&	  0.6  &    3.4  &    2.7   &    9.6  &    9.2  &  86.9    \\
				&			& 	 		      	&	0.1\%		&	  1.3  &    6.0  &    4.9   &   46.1  &   18.0  &  90.5    \\
				&			& 	 		      	&	0.05\%	&   1.6  &    7.6  &    6.3   &   52.9  &   23.4  &  91.6    \\
				&			& 	 		      	&	0.01\%	&	  3.2  &   12.9  &   10.8   &   72.7  &   41.3  &  93.1    \\
				&			& 	 		      	&	0.005\%	&   4.3  &   16.1  &   13.7   &   83.8  &   52.0  &  93.4    \\
				&			& 	 		      	&	0.001\%	&  11.0  &   27.7  &   23.1   &  114.0  &   86.9  &  93.9    \\
						\cmidrule(lr){2-10}                                        
				&	$3$	& $(0.2, 25)$   & 1\%			&	  0.1  &    2.5  &    2.0   &    4.6  &    6.4  &  83.2    \\ 
				&			&								&	0.5\%		&	  0.3  &    3.3  &    2.6   &    9.3  &    8.9  &  85.2    \\ 
				&			& 	 		      	&	0.1\%		&	  0.8  &    5.8  &    4.5   &   46.1  &   17.4  &  90.2    \\ 
				&			& 	 		      	&	0.05\%	&   1.2  &    7.3  &    5.6   &   53.1  &   22.6  &  01.4    \\ 
				&			& 	 		      	&	0.01\%	&	  2.6  &   12.3  &    9.3   &   73.7  &   39.6  &  93.1    \\ 
				&			& 	 		      	&	0.005\%	&   3.9  &   15.4  &   11.4   &   84.9  &   49.8  &  93.6    \\ 
			  &			& 	 		      	&	0.001\%	&   6.0  &   24.8  &   18.6   &  117.4  &   82.9  &  94.7    \\ 
					\cmidrule(lr){2-10}                                          
				&	$5$	& $(0.25, 20)$  & 1\%		  &	  0.7  &    2.7  &    2.0   &    4.6  &    6.7  &  82.7    \\ 
				&			&								&	0.5\%		&	  1.1  &    3.6  &    2.7   &    9.2  &    9.3  &  85.7    \\ 
				&			& 	 		      	&	0.1\%		&	  2.6  &    6.6  &    4.7   &   45.9  &   18.2  &  88.9    \\ 
				&			& 	 		      	&	0.05\%	&   3.5  &    8.4  &    5.9   &   52.6  &   23.7  &  89.8    \\ 
				&			& 	 		      	&	0.01\%	&	  6.8  &   14.5  &    9.9   &   72.6  &   41.8  &  90.4    \\ 
				&			& 	 		      	&	0.005\%	&   8.1  &   17.8  &   12.3   &   83.5  &   52.6  &  91.6    \\ 
				&			& 	 		      	&	0.001\%	&  10.2  &   27.6  &   20.5   &  117.6  &   88.0  &  94.1    \\ 
			\bottomrule
		\end{tabular}
	\caption{Bias, RMSEs, average interval lengths of \eqref{eq:CI} and coverage (in \%) for $n=500$ and $1-\iota=95\%$. Bias, RMSEs and interval lengths are all multiplied by 100 for better readability.}
	\label{tab:500}
\end{table}

\begin{table}[t!]
	\centering
		\begin{tabular}{lllrrrrrrr}
			\toprule
  $n$ 	&$\nu$& $(a,b)$ 			&$p$			& Bias	 & RMSE	   & $\substack{\text{RMSE}\\ \text{ML}}$& $\substack{\text{RMSE}\\ \text{NP}}$& Length   & Coverage \\
	\midrule                                                                               
$1000$	&	$3$	& $(0.25, 20)$  & 1\%		  &	0.2    &   1.9   &   1.4  &   3.3  &      4.8  &   83.1     \\
				&			&								&	0.5\%		&	0.4    &   2.5   &   1.8  &   5.3  &      6.7  &   86.3     \\
				&			& 	 		      	&	0.1\%		&	1.0    &   4.4   &   3.2  &  14.5  &     13.4  &   90.4     \\
				&			& 	 		      	&	0.05\%	& 1.3    &   5.6   &   4.1  &  53.2  &     17.4  &   90.8     \\
				&			& 	 		      	&	0.01\%	&	2.2    &   9.4   &   6.8  &  73.5  &     30.7  &   92.5     \\
				&			& 	 		      	&	0.005\%	& 2.4    &  11.6   &   8.5  &  83.3  &     38.6  &   93.1     \\
				&			& 	 		      	&	0.001\%	& 2.3    &  18.5   &  14.2  & 114.2  &     64.1  &   93.5     \\
						\cmidrule(lr){2-10}                                      
				&	$3$	& $(0.2, 25)$   & 1\%			&	0.1    &   1.9   &   1.5  &   3.3  &      4.7  &   82.1     \\
				&			&								&	0.5\%		&	0.2    &   2.4   &   2.0  &   5.4  &      6.7  &   85.8     \\
				&			& 	 		      	&	0.1\%		&	0.5    &   4.3   &   3.4  &  15.5  &     13.2  &   89.6     \\
				&			& 	 		      	&	0.05\%	& 0.8    &   5.5   &   4.3  &  52.6  &     17.2  &   91.1     \\
				&			& 	 		      	&	0.01\%	&	1.9    &   9.3   &   7.0  &  72.2  &     30.2  &   92.5     \\
				&			& 	 		      	&	0.005\%	& 3.0    &  11.7   &   8.7  &  82.7  &     37.9  &   92.9     \\
			  &			& 	 		      	&	0.001\%	& 7.9    &  19.9   &  14.2  & 114.5  &     62.9  &   92.9     \\
					\cmidrule(lr){2-10}                                        
				&	$5$	& $(0.25, 20)$  & 1\%		  &	0.6    &   2.0   &   1.4  &   3.2  &      4.8  &   82.0     \\
				&			&								&	0.5\%		&	0.8    &   2.7   &   1.9  &   5.2  &      6.8  &   85.2     \\
				&			& 	 		      	&	0.1\%		&	1.6    &   4.8   &   3.3  &  14.3  &     13.5  &   89.6     \\
				&			& 	 		      	&	0.05\%	& 2.2    &   6.1   &   4.2  &  52.3  &     17.5  &   90.6     \\
				&			& 	 		      	&	0.01\%	&	3.8    &  10.3   &   7.0  &  71.7  &     30.9  &   92.3     \\
				&			& 	 		      	&	0.005\%	& 4.8    &  12.8   &   8.7  &  82.1  &     38.8  &   92.7     \\
				&			& 	 		      	&	0.001\%	& 8.7    &  21.2   &  15.7  & 107.9  &     64.6  &   93.3     \\
		\bottomrule
		\end{tabular}
	\caption{Bias, RMSEs, average interval lengths of \eqref{eq:CI} and coverage (in \%) for $n=1000$ and $1-\iota=95\%$. Bias, RMSEs and interval lengths are all multiplied by 100 for better readability.}
	\label{tab:1000}
\end{table}

We draw the following conclusions from the simulation results in Tables~\ref{tab:500}--\ref{tab:1000}:

\begin{enumerate}
	\item The more extreme $p$, the better the coverage of the confidence intervals \eqref{eq:CI}. This may be explained as follows. For more extreme $p$, the quantity $\sqrt{k}\log(d_n) / \sqrt{k_1}$ is larger, suggesting that the limit in Theorem~\ref{thm:main result} can be better approximated by $\Gamma$, i.e., the limit distribution exploited in the construction \eqref{eq:CI}. As pointed out in the Motivation, it is precisely the small $p$'s for which systemic risk measures are of most interest. For these small $p$'s our confidence intervals are reasonably accurate, particularly when compared with other EVT-based confidence intervals \citep[see, e.g.,][Fig.~1]{Cea07}. To improve coverage for less extreme $p$, one may entertain self-normalized confidence intervals as in \citet{Hog18+} or one may use bootstrap approximations, as explored by \citet{LPS23} in the context of extreme VaR and ES forecasting. In principle, coverage for larger $p$ could also be improved by using forecast intervals based on the limiting distribution $(1/r)\Theta+\Gamma$. However, to the best of our knowledge, no consistent estimator has been proposed for its asymptotic variance, which seems difficult to estimate.
	
	\item As expected, the more extreme $p$, the larger the bias and the RMSE of the MES forecasts. Also, the larger $b$ and the smaller $\nu$, the lower the bias and RMSE tend to be for fixed risk level $p$. This may be explained by the fact that for larger $b$ the Pareto approximation is more accurate for the marginals. Similarly, for smaller $\nu$, the approximation in Assumption~\ref{ass:R}---that is used in extrapolating---is more precise, leading to lower bias and RMSE. Of course, bias and particularly the RMSE are also reduced when the sample size $n$ increases.
	
	\item Similarly as the RMSE, the lengths of the confidence intervals also decrease the larger $n$ and the better the approximations in Assumptions~\ref{ass:R} and \ref{ass:U} (i.e., the smaller $\nu$ and the larger $b$). This is also as expected as confidence intervals provide a measure of the statistical uncertainty around the point MES forecasts.
	
	\item For purposes of comparison, we have also included the RMSEs of two additional MES forecasts with quite different robustness-efficiency tradeoff: one based on a parametric maximum likelihood (ML) estimator and another based on a non-parametric (NP) estimator.\footnote{The ML estimator separately estimates the copula parameters ($\nu,\rho_{X,Y}$) and the marginal parameters ($a,b$) via ML. It then computes the MES $\theta_p$ implied by the estimated values. In contrast, the NP estimate is simply $\widehat{\theta}_p^{\textup{NP}}=\frac{1}{np}\sum_{t=1}^{n}\widehat{\varepsilon}_{t,Y}I_{\{\widehat{\varepsilon}_{t,X}>\widehat{\varepsilon}_{(\lfloor np\rfloor+1),X}\}}$. In both cases, the MES forecasts are obtained by pre-multiplying the MES estimate with $\widehat{\sigma}_{n+1}^{Y}$; cf.~\eqref{eq:MES forc}.} Our semi-parametric estimator $\widehat{\theta}_p$ provides a balance between the advantages of these two estimators by producing reasonably robust estimates with good efficiency. The results are as expected. Because the assumptions underlying the use of the ML estimator and our $\widehat{\theta}_p$ are met, they are more efficient than the NP estimates in terms of RMSE, with the ML estimator being the favorite. But of course the ML estimator may be dangerous to use under misspecification. As \citet[Sec.~2.2]{Dre08} warns, the consequences of misspecification of model-based estimates are even magnified in the tail and, therefore, such estimates are not recommended in applications of EVT. Indeed, unreported simulations with the ML estimator based on (misspecified) marginal $t$-densities show an inferior performance compared to $\widehat{\theta}_p$ with RMSEs 2--4 times as large.
\end{enumerate}

\section{Empirical Application}\label{Empirical Application}

Here, we consider the 8 US G-SIBs as determined by the \citet{FSB21}.\footnote{Specifically, the 8 G-SIBs are Bank of America, Bank of New York Mellon, Citigroup, Goldman Sachs, JP Morgan Chase, Morgan Stanley, State Street, and Wells Fargo.} We use daily data from 2001 to 2021, because data for JP Morgan Chase---one of the G-SIBs---is only available after the merger of JP Morgan and Chase Manhattan in 2000. This gives us $N=5283$ observations. All data are downloaded from Datastream.

Since we consider 8 banks, we have $D=8$ in the notation of Section~\ref{Higher-Dimensional Extensions}. The $Y_{t,d}$ ($d=1,\ldots,D$) then denote the log-losses of the individual institutions' shares (calculated based on adjusted closing prices). The value-weighted losses of the whole financial system are $X_t=\sum_{t=1}^{D}w_{t-1,d}Y_{t,d}$. Specifically, the weights sum to one (i.e., $\sum_{d=1}^{D}w_{t-1,d}=1$) and are calculated as $w_{t-1,d}=M_{t-1,d}/(\sum_{d=1}^{D}M_{t-1,d})$ with $M_{t-1,d}$ the market capitalization of institution $d$ at time $t-1$; see \citet{QZ21} for a similar approach. To forecast MES, we use a CCC--GARCH model of the form \eqref{eq:model h} with standard GARCH(1,1) volatility models for the marginals. Note that a model only for $(Y_{t,1},\ldots,Y_{t,D})^\prime$ is not sufficient for this purpose, as $X_t$ is not only composed of the $Y_{t,d}$'s but also of the time-varying weights $w_{t-1,d}$. We use rolling-window MES forecasts, where a window of length $(n+\ell_n)$ is rolled through the $N=5283$ observations to yield $N-(n+\ell_n)$ one-step ahead forecasts. We pick $n=1000$ and $\ell_n=10$, corresponding to roughly four years of daily returns in each moving window. To reflect the scarcity of systemic events, we choose $p=1/n=0.001$. Results for different values of $p$ are available upon request. As in the simulations, we set $k=k_1=\lfloor 0.1 \cdot \log(1000)^4\rfloor$.

\begin{figure}[t!]
	\centering
		\includegraphics[width=\textwidth]{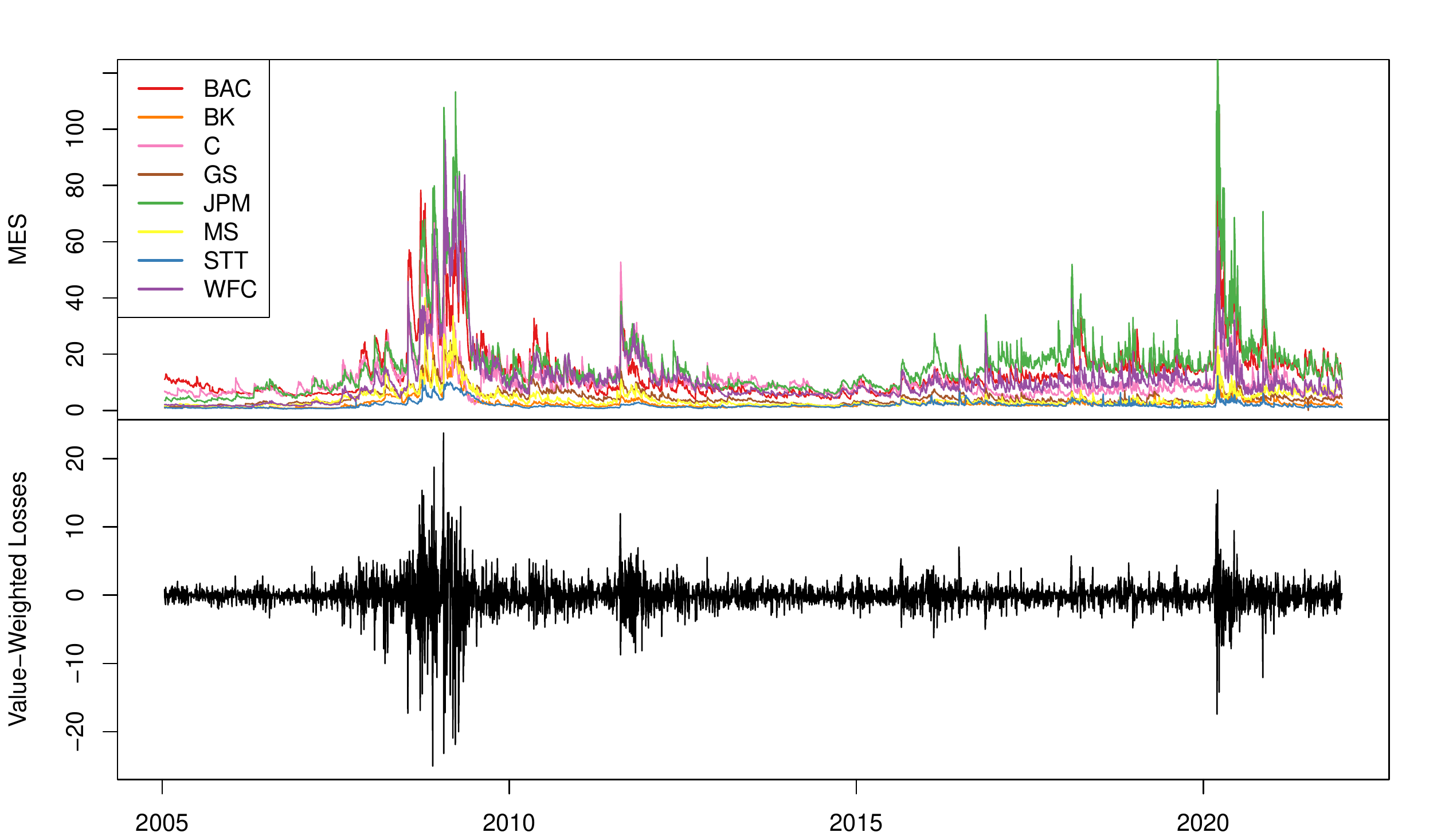}
	\caption{Top panel: Value-weighted MES forecasts for eight US G-SIBs Bank of America (BAC), Bank of New York Mellon (BK), Citigroup (C), Goldman Sachs (GS), JP Morgan Chase (JPM), Morgan Stanley (MS), State Street (STT), and Wells Fargo (WFC). Bottom panel: Time series of value-weighted losses $X_t$.}
	\label{fig:1}
\end{figure}

From a regulatory perspective, it is important---at each point $t$ in time---to limit the risk of the overall banking system as measured by the ES, viz.
\[
	\E[X_t\mid X_t>\VaR_t(p), \mathcal{F}_{t-1}]=\sum_{d=1}^{D}w_{t-1,d}\E[Y_{t,d}\mid X_t>\VaR_t(p),\mathcal{F}_{t-1}]=\sum_{d=1}^{D}w_{t-1,d}\theta_{t-1,p,d}. 
\]
The bottom panel of Figure~\ref{fig:1} shows a plot of the value-weighted losses $X_t$. During periods of high volatility, aggregate risk of the system is high. Such periods are noticeable during the financial crisis of 2007--2009, the European sovereign debt crisis in the early 2010s, and the Corona crash of March 2020.

The top panel of Figure~\ref{fig:1} plots the systemic risk contributions of each bank as measured by the value-weighted MES $w_{t-1,d}\theta_{t-1,p,d}$. Clearly, the contributions vary through time. Notice also that the MES forecasts during 2021 largely reflect the systemic risk contributions of the different institutions as determined by the \citet{FSB21}. For instance, JP Morgan Chase is listed as the systemically most risky bank by the \citet{FSB21}, which is mirrored by the high MES forecasts in the top panel of Figure~\ref{fig:1}. At the other end of the spectrum, State Street, Bank of New York Mellon, and Morgan Stanley are deemed least risky by the \citet{FSB21}, similarly as their low MES forecasts suggest.

The placement of the institutions in different buckets by the \citet{FSB21} (with bucket 5 containing the most systemically risky institutions and bucket 1 the least systemically risky ones) suggests that the systemic risk contributions of the different banks may also be distinguishable in statistical terms. Thus, for each point in time $t$ for which MES forecasts are issued, we test equality of the value-weighted MES forecasts, i.e., $w_{t-1,1}\theta_{t-1,p,1}=\ldots=w_{t-1,D}\theta_{t-1,p,D}$. To do so, we use the test statistic $\mathcal{T}_n$ from Corollary~\ref{cor:test}. For each point in time, the $p$-value is virtually zero, indicating---as expected---large heterogeneity of systemic risk contributions among banks.

\begin{figure}[t!]
	\centering
		\includegraphics[width=\textwidth]{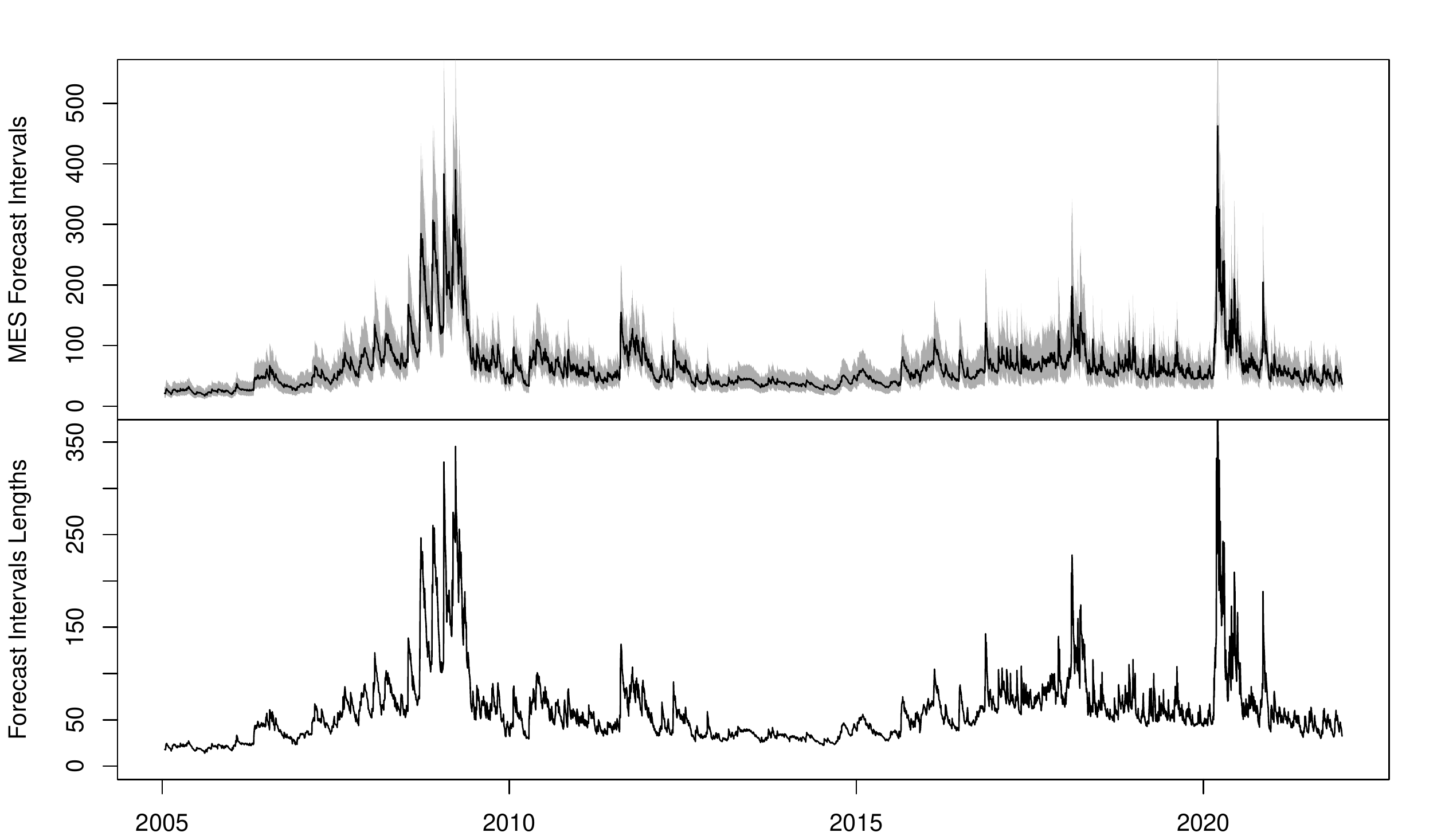}
	\caption{Top panel: MES forecasts (solid black line) together with 95\%-confidence intervals (shaded grey area) for JP Morgan Chase. Bottom panel: Lengths of 95\%-confidence intervals.}
	\label{fig:2}
\end{figure}

To limit the overall riskiness of the system, it seems prudent for the regulator to not only rely on the point MES forecasts, but also to take into account the estimation risk associated with those forecasts. To illustrate the significant impact of estimation risk, Figure~\ref{fig:2} displays the MES forecasts together with the 95\%-confidence intervals given in \eqref{eq:CI}. It does so exemplarily for JP Morgan Chase, which is ranked as the world's most systemically risky bank by the \citet{FSB21}. The financial crisis of 2007--2009 and the Corona crash of March 2020 are associated with the largest spikes in systemic riskiness.  It is precisely during these times, where precise
risk assessments are needed most, that our intervals suggest that systemic riskiness is most
poorly forecasted. This can be seen most clearly from the bottom panel of Figure~\ref{fig:2}, where lengths of the 95\%-forecast intervals are shown. During crises (e.g., around 2008) the estimation risk in the forecasts is highest. In contrast, during calmer times, the estimation risk is comparatively smaller. Formula \eqref{eq:CI} suggests that the main driver of these differences in forecast quality is the volatility of the JP Morgan Chase shares.

\section{Conclusion}\label{Conclusion}

For univariate risk measures, such as VaR and ES, there is a voluminous literature on asymptotic properties of forecasts. In contrast, asymptotic properties of systemic risk forecasts are largely unexplored. We fill this gap by deriving limit theory for EVT-based MES forecasts. In doing so, we extend the unconditional MES estimator of \citet{Cea15} along two dimensions. First, we prove its validity when applied to residuals of multivariate volatility models, thus allowing it to be used for conditional MES forecasting and confidence interval construction. In simulations, we illustrate the good finite-sample coverage of the forecast intervals, which provide valuable information beyond the mere point forecast of MES. Second, we derive limit theory also in higher dimensional systems, therefore enabling joint inference on multiple MES forecasts. This may be beneficial as illustrated in the empirical application to the losses on the 8 US G-SIBs.

The following avenues may be worth exploring in future research. First, one may develop bootstrap-based confidence intervals for MES to improve coverage, particularly for not so extreme $p$. Second, it may be interesting to explore the properties of our forecasts for data-adaptive choices of $k$ and $k_1$, which may improve finite-sample properties. The results of \citet{Dea20} in the context of extreme value index estimation suggest that this may influence the asymptotic behavior. Third, we have exclusively focused on MES as a measure of systemic risk in this paper. However, there are many other popular systemic risk measures in the literature, such as the CoVaR and CoES of \citet{AB16}. Thus, future research could develop limit theory for (EVT-based) forecasts of these measures as well.

\onehalfspacing

\appendix
\appendixpage

If not specified otherwise, all limits and all $o_{(\p)}$- and $O_{(\p)}$-symbols are to be interpreted with respect to $n\to\infty$. We denote by $K>0$ a large positive constant, that may change from line to line, and by $\mI$ the identity matrix of appropriate dimension.

\section{Proof of Theorem~\ref{thm:main result}}
\label{sec:Proof of Theorem 1}

\renewcommand{\theequation}{A.\arabic{equation}}	
\setcounter{equation}{0}

The proof of Theorem~\ref{thm:main result} requires Propositions~\ref{prop:approx} and \ref{prop:MES}.

\begin{prop}\label{prop:approx}
Suppose Assumptions~\ref{ass:param est}--\ref{ass:vola init} hold, and the truncation sequence $\ell_n$ satisfies $\ell_n/\log n\to\infty$. Then, we have for any $\iota>0$ that, as $n\to\infty$,
\begin{align*}
	\widehat{\vvarepsilon}_t=\vvarepsilon_t\big\{1 + o_{\p}(n^{\iota-\xi})\big\},\\
	\widehat{\vsigma}_t=\vsigma_t\big\{1 + o_{\p}(n^{\iota-\xi})\big\}
\end{align*}
uniformly in $t=1,\ldots,n+1$. 
\end{prop}

\begin{proof}
See Appendix~\ref{Proof of Proposition 1}.
\end{proof}

\begin{prop}\label{prop:MES}
Under the conditions of Theorem~\ref{thm:main result}, we have that, as $n\to\infty$,
\[
	\min\big\{\sqrt{k}, \sqrt{k_1}/\log(d_n)\big\}\log\Bigg(\frac{\widehat{\theta}_{p}}{\theta_{p}}\Bigg)\overset{d}{\longrightarrow}\begin{cases}\Theta+r\Gamma,& \text{if }r\leq1,\\
	(1/r)\Theta+\Gamma,&\text{if }r>1,\end{cases}
\]
with $r$, $\Gamma$ and $\Theta$ defined as in Theorem~\ref{thm:main result}.
\end{prop}

\begin{proof}
See Appendix~\ref{Proof of Proposition 2}.
\end{proof}

\begin{proof}[\textbf{Proof of Theorem~\ref{thm:main result}}]
Use \eqref{eq:MES true} and \eqref{eq:MES forc} to write
\begin{align}
\min&\big\{\sqrt{k}, \sqrt{k_1}/\log(d_n)\big\}\log\Bigg(\frac{\widehat{\theta}_{n,p}}{\theta_{n,p}}\Bigg)\notag\\
&=\min\big\{\sqrt{k}, \sqrt{k_1}/\log(d_n)\big\}\log\Bigg(\frac{\widehat{\sigma}_{n+1,Y}}{\sigma_{n+1,Y}}\cdot \frac{\widehat{\theta}_p}{\theta_p}\Bigg)\notag\\
&=\min\big\{\sqrt{k}, \sqrt{k_1}/\log(d_n)\big\}\log\Bigg(\frac{\widehat{\sigma}_{n+1,Y}}{\sigma_{n+1,Y}}\Bigg)+\min\big\{\sqrt{k}, \sqrt{k_1}/\log(d_n)\big\}\log\Bigg(\frac{\widehat{\theta}_{p}}{\theta_{p}}\Bigg).\label{eq:decomp}
\end{align}
From Proposition~\ref{prop:approx} and Assumption~\ref{ass:k} it follows that for $\iota<\xi-\widetilde{\alpha}$,
\begin{align}
	\min\big\{\sqrt{k}, \sqrt{k_1}/\log(d_n)\big\}\log\Bigg(\frac{\widehat{\sigma}_{n+1,Y}}{\sigma_{n+1,Y}}\Bigg)&=\min\big\{\sqrt{k}, \sqrt{k_1}/\log(d_n)\big\}\log\Big(1+o_{\p}(n^{\iota-\xi})\Big)\notag\\
	&=\min\big\{\sqrt{k}, \sqrt{k_1}/\log(d_n)\big\}o_{\p}(n^{\iota-\xi})\notag\\
	&=O(n^{\widetilde{\alpha}})o_{\p}(n^{\iota-\xi})\notag\\
	&=o_{\p}(n^{\widetilde{\alpha}+\iota-\xi})\notag\\
	&=o_{\p}(1),\label{eq:vola negl}
\end{align}
because $\log x\sim x-1$, $x\to1$.
For the second right-hand side term in \eqref{eq:decomp} Proposition~\ref{prop:MES} applies. Overall, the conclusion follows.
\end{proof}

\section{Proof of Proposition~\ref{prop:approx}}\label{Proof of Proposition 1}

\renewcommand{\theequation}{B.\arabic{equation}}	
\setcounter{equation}{0}

\begin{proof}[\textbf{Proof of Proposition~\ref{prop:approx}}]
Fix $\iota>0$. Choose $M>0$ from Assumption~\ref{ass:vola est} sufficiently large, such that $1/M<\iota$. Recall the definition of $\mathcal{N}(\vtheta^\circ)$ from Assumption~\ref{ass:vola est}. For $\vtheta\in\mathcal{N}(\vtheta^\circ)$, write
\begin{align}
	\widehat{\vvarepsilon}_t(\vtheta) &= \widehat{\mSigma}_t^{-1}(\vtheta)(X_t,Y_t)^\prime = \widehat{\mSigma}_t^{-1}(\vtheta)\mSigma_t(\vtheta^\circ)\vvarepsilon_t\notag\\
	&= \big[ \mSigma_t^{-1}(\vtheta) \widehat{\mSigma}_t(\vtheta)\big]^{-1}\big[\mSigma_t^{-1}(\vtheta^\circ)\mSigma_t(\vtheta)\big]^{-1}\vvarepsilon_t\notag\\
	&= \Big\{\mI + \mSigma_t^{-1}(\vtheta) \big[\widehat{\mSigma}_t(\vtheta) - \mSigma_t(\vtheta)\big]\Big\}^{-1}\Big\{\mI + \mSigma_t^{-1}(\vtheta^\circ)\big[\mSigma_t(\vtheta) - \mSigma_t(\vtheta^\circ)\big]\Big\}^{-1} \vvarepsilon_t.\label{eq:eps decomp}
\end{align}

Use Assumption~\ref{ass:vola init} (where $C_0$ is random) to conclude that
\begin{align}
\max_{t=1,\ldots,n+1}\norm{\mSigma_t^{-1}(\vtheta) \big[\widehat{\mSigma}_t(\vtheta) - \mSigma_t(\vtheta)\big]} &\leq \max_{t=1,\ldots,n+1}\norm{\mSigma_t^{-1}(\vtheta)} \max_{t=1,\ldots,n+1}\big\Vert\widehat{\mSigma}_t(\vtheta) - \mSigma_t(\vtheta)\big\Vert \notag\\
&\leq C C_0 \rho^{\ell_n}=o_{\p}(n^{-K})\label{eq:fi}
\end{align}
for any $K>0$, because
\begin{align*}
	\rho^{\ell_n} &= \exp\Big\{\log\big(\rho^{(\ell_n/\log n)\log n}\big)\Big\}\\
	&= \exp\Big\{\log (n) \log\big(\rho^{\ell_n/\log n}\big)\Big\}=\big(\exp\{\log n\}\big)^{\log\big(\rho^{\ell_n/\log n}\big)}\\
	&=k^{\log\big(\rho^{\ell_n/\log n}\big)}=o\big(n^{-K}\big),
\end{align*}
since $\ell_n/\log n\to\infty$ by assumption, such that $\rho^{\ell_n/\log n}\to0$ and, hence, $\log\big(\rho^{\ell_n/\log n}\big)\to-\infty$.

For the second right-hand side term in \eqref{eq:eps decomp}, we obtain that 
\begin{multline}\label{eq:param est decomp}
	\max_{t=1,\ldots,n+1}\Big\Vert\mSigma_t^{-1}(\vtheta^\circ)\big[\mSigma_t(\vtheta) - \mSigma_t(\vtheta^\circ)\big]\Big\Vert \leq \max_{t=1,\ldots,n+1}\bigg|\frac{\sigma_{t,X}(\vtheta) - \sigma_{t,X}(\vtheta^\circ)}{\sigma_{t,X}(\vtheta^\circ)}\bigg|\\
	+ \max_{t=1,\ldots,n+1}\bigg|\frac{\sigma_{t,Y}(\vtheta) - \sigma_{t,Y}(\vtheta^\circ)}{\sigma_{t,Y}(\vtheta^\circ)}\bigg|.
\end{multline}
Consider the first term on the right-hand side of \eqref{eq:param est decomp}. Define $\dot{\sigma}_{t,X}(\vtheta)=\partial \sigma_{t,X}(\vtheta) / \partial \vtheta$. Use the mean value theorem to deduce that for some $\vtheta^\ast$ on the line connecting $\vtheta$ and $\vtheta^\circ$,
\begin{align}
\max_{t=1,\ldots,n+1}\bigg|\frac{\sigma_{t,X}(\vtheta) - \sigma_{t,X}(\vtheta^\circ)}{\sigma_{t,X}(\vtheta^\circ)}\bigg| &= \max_{t=1,\ldots,n+1}\bigg|\frac{\dot{\sigma}_{t,X}(\vtheta^\ast)(\vtheta - \vtheta^\circ)}{\sigma_{t,X}(\vtheta^\circ)}\bigg|\notag\\
&= \max_{t=1,\ldots,n+1}\bigg|\frac{\dot{\sigma}_{t,X}(\vtheta^\ast)}{\sigma_{t,X}(\vtheta^\ast)}\cdot\frac{\sigma_{t,X}(\vtheta^\ast)}{\sigma_{t,X}(\vtheta^\circ)}\cdot(\vtheta - \vtheta^\circ)\bigg|\notag\\
&\leq \max_{t=1,\ldots,n+1}\bigg\{\sup_{\vtheta\in\mathcal{N}(\vtheta^\circ)}\norm{\frac{\dot{\sigma}_{t,X}(\vtheta)}{\sigma_{t,X}(\vtheta)}}\cdot\sup_{\vtheta\in\mathcal{N}(\vtheta^\circ)}\bigg|\frac{\sigma_{t,X}(\vtheta)}{\sigma_{t,X}(\vtheta^\circ)}\bigg|\bigg\}\norm{\vtheta - \vtheta^\circ}.\label{eq:(3.1ti)}
\end{align}
We obtain that
\begin{align*}
	\p&\bigg\{\max_{t=1,\ldots,n+1}\bigg\{\sup_{\vtheta\in\mathcal{N}(\vtheta^\circ)}\norm{\frac{\dot{\sigma}_{t,X}(\vtheta)}{\sigma_{t,X}(\vtheta)}}\cdot\sup_{\vtheta\in\mathcal{N}(\vtheta^\circ)}\bigg|\frac{\sigma_{t,X}(\vtheta)}{\sigma_{t,X}(\vtheta^\circ)}\bigg|\bigg\}>\varepsilon n^\iota\bigg\} \\
	& \leq\sum_{t=1}^{n+1}\p\bigg\{\sup_{\vtheta\in\mathcal{N}(\vtheta^\circ)}\norm{\frac{\dot{\sigma}_{t,X}(\vtheta)}{\sigma_{t,X}(\vtheta)}}\cdot\sup_{\vtheta\in\mathcal{N}(\vtheta^\circ)}\bigg|\frac{\sigma_{t,X}(\vtheta)}{\sigma_{t,X}(\vtheta^\circ)}\bigg|>\varepsilon n^{\iota}\bigg\}\\
	&\leq \sum_{t=1}^{n+1}\frac{1}{\varepsilon^M n^{M\iota}}\E\bigg[\sup_{\vtheta\in\mathcal{N}(\vtheta^\circ)}\norm{\frac{\dot{\sigma}_{t,X}(\vtheta)}{\sigma_{t,X}(\vtheta)}}^M\cdot\sup_{\vtheta\in\mathcal{N}(\vtheta^\circ)}\bigg|\frac{\sigma_{t,X}(\vtheta)}{\sigma_{t,X}(\vtheta^\circ)}\bigg|^M\bigg]\\
	&\leq \sum_{t=1}^{n+1}\frac{1}{\varepsilon^M n^{M\iota}}\Bigg\{\E\bigg[\sup_{\vtheta\in\mathcal{N}(\vtheta^\circ)}\norm{\frac{\dot{\sigma}_{t,X}(\vtheta)}{\sigma_{t,X}(\vtheta)}}^{Mp_{\ast}}\Bigg\}^{1/p_{\ast}}\cdot\Bigg\{\E\bigg[\sup_{\vtheta\in\mathcal{N}(\vtheta^\circ)}\bigg|\frac{\sigma_{t,X}(\vtheta)}{\sigma_{t,X}(\vtheta^\circ)}\bigg|^{Mq_{\ast}}\bigg]\Bigg\}^{1/q_{\ast}}\\
	&\leq K n^{1-M\iota},
\end{align*}
where we used subadditivity in the first step, Markov's inequality in the second step, H\"{o}lder's inequality in the third step, and Assumption~\ref{ass:vola est} for the final inequality. Therefore, since $1/M<\iota$,
\[
	\max_{t=1,\ldots,n+1}\bigg\{\sup_{\vtheta\in\mathcal{N}(\vtheta^\circ)}\norm{\frac{\dot{\sigma}_{t,X}(\vtheta)}{\sigma_{t,X}(\vtheta)}}\cdot\sup_{\vtheta\in\mathcal{N}(\vtheta^\circ)}\bigg|\frac{\sigma_{t,X}(\vtheta)}{\sigma_{t,X}(\vtheta^\circ)}\bigg|\bigg\}=o_{\p}(n^{\iota}),
\]
whence from \eqref{eq:(3.1ti)} and Assumption~\ref{ass:param est},
\begin{equation}\label{eq:(B.4)}
	\max_{t=1,\ldots,n+1}\bigg|\frac{\sigma_{t,X}(\vtheta) - \sigma_{t,X}(\vtheta^\circ)}{\sigma_{t,X}(\vtheta^\circ)}\bigg|=o_{\p}(n^{\iota})n^{-\xi}n^{\xi}\norm{\vtheta-\vtheta^\circ} =o_{\p}(n^{\iota-\xi})O_{\p}(1)=o_{\p}(n^{\iota-\xi}).
\end{equation}
Using identical arguments, we may also show that
\[
	\max_{t=1,\ldots,n+1}\bigg|\frac{\sigma_{t,Y}(\vtheta) - \sigma_{t,Y}(\vtheta^\circ)}{\sigma_{t,Y}(\vtheta^\circ)}\bigg|=o_{\p}(n^{\iota-\xi}).
\]
Thus, from \eqref{eq:param est decomp},
\begin{equation}\label{eq:se}
	\max_{t=1,\ldots,n+1}\norm{\mSigma_t^{-1}(\vtheta^\circ)\big[\mSigma_t(\vtheta) - \mSigma_t(\vtheta^\circ)\big]} = o_{\p}(n^{\iota-\xi}).
\end{equation}
Plugging \eqref{eq:fi} and \eqref{eq:se} into \eqref{eq:eps decomp}, we get that for all $\vtheta\in\mathcal{N}(\vtheta^\circ)$
\[
	\widehat{\vvarepsilon}_t(\vtheta) = \vvarepsilon_t\big\{1+ o_{\p}(n^{\iota-\xi})\big\},
\]
where we recall that the matrices $\mSigma(\cdot)$ and $\widehat{\mSigma}(\cdot)$ in \eqref{eq:eps decomp} are diagonal.
Since, by Assumption~\ref{ass:param est}, $\widehat{\vtheta}$ is an element of $\mathcal{N}(\vtheta^\circ)$ with probability approaching 1, as $n\to\infty$, the conclusion for the residuals follows.

The uniform approximability of the volatilities follows from
\begin{align*}
	\bigg|\frac{\widehat{\sigma}_{t,X}(\vtheta)}{\sigma_{t,X}(\vtheta^\circ)}-1\bigg|& =\bigg|\frac{\widehat{\sigma}_{t,X}(\vtheta) - \sigma_{t,X}(\vtheta)}{\sigma_{t,X}(\vtheta^\circ)} + \frac{\sigma_{t,X}(\vtheta)}{\sigma_{t,X}(\vtheta^\circ)}-1\bigg|\\
	& \leq \bigg|\frac{\widehat{\sigma}_{t,X}(\vtheta) - \sigma_{t,X}(\vtheta)}{\sigma_{t,X}(\vtheta^\circ)}\bigg| + \bigg|\frac{\sigma_{t,X}(\vtheta)}{\sigma_{t,X}(\vtheta^\circ)}-1\bigg|\\
	&= o_{\p}(n^{-K}) + o_{\p}(n^{\iota-\xi})
\end{align*}
(from \eqref{eq:fi} and \eqref{eq:(B.4)}) together with the $n^{\xi}$-consistency of $\widehat{\vtheta}$ (from Assumption~\ref{ass:param est}). A similar result obtains for $\widehat{\sigma}_{t,Y}$.
\end{proof}

\section{Proof of Proposition~\ref{prop:MES}}\label{Proof of Proposition 2}

\renewcommand{\theequation}{C.\arabic{equation}}	
\setcounter{equation}{0}	

Let $F_{+}(\cdot)$ denote the d.f.~of $\varepsilon_{t,Y}^{+}$, and set $U_{+}=[1/(1-F_{+})]^{\leftarrow}$. Then, the proof of Theorem~2 in \citet{Cea15} shows that Assumptions~\ref{ass:R} and~\ref{ass:U} (phrased in terms of $(\varepsilon_{t,X}, \varepsilon_{t,Y})^\prime$) continue to hold for $(\varepsilon_{t,X}, \varepsilon_{t,Y}^{+})^\prime$ with the same constants and functions $A_i(\cdot)$. This will be exploited in some of the following proofs without further mention.

For $(x,y)^\prime\in[0,\infty]^2\setminus\{(\infty,\infty)^\prime\}$, we define
\begin{align*}
	R_n(x,y) &= \frac{n}{k}\p\big\{F_{0}(\varepsilon_{t,X})>1-kx/n,\ F_+(\varepsilon_{t,Y}^{+})>1-ky/n\big\},\\
	T_n(x,y) &= \frac{1}{k}\sum_{t=1}^{n}I_{\big\{F_{0}(\varepsilon_{t,X})>1-kx/n,\ F_+(\varepsilon_{t,Y}^{+})>1-ky/n\big\}},\\
	\widehat{T}_n(x,y) &= \frac{1}{k}\sum_{t=1}^{n}I_{\big\{F_{0}(\widehat{\varepsilon}_{t,X})>1-kx/n,\ F_+(\widehat{\varepsilon}_{t,Y}^{+})>1-ky/n\big\}}.
\end{align*}

The limit distribution of $\widehat{\theta}_{p}$ is characterized by the zero-mean Gaussian process
\[
	\big\{W_R(x,y)\big\}_{(x,y)^\prime\in[0,\infty]^2\setminus\{(\infty,\infty)^\prime\}}
\]
with covariance structure given by
\[
	\E\big[W_R(x_1,y_1)W_R(x_2,y_2)\big] = R(x_1\wedge x_2, y_1\wedge y_2).
\]
Then, 
\begin{align*}
	\Theta &= (\gamma_1-1)W_R(1, \infty) + \Big\{\int_{0}^{\infty}R(1, s)\D s^{-\gamma_1}\Big\}^{-1}\int_{0}^{\infty}W_R(1,s)\D s^{-\gamma_1},\\
	\Gamma &= \frac{\gamma_1}{\sqrt{q}}\Big\{-W_R(\infty,q) + \int_{0}^{q}s^{-1}W_R(\infty,s)\D s\Big\}
\end{align*}
are the zero-mean Gaussian random variables from Theorem~\ref{thm:main result}.

The proof of Proposition~\ref{prop:MES} requires Lemmas~\ref{lem:Lemma 1}--\ref{lem:MES final}. These lemmas build on Proposition~3.1 in \citet{EHL06}. Invoking a Skorohod construction, the limit processes involved in that proposition may be assumed to be defined on the same probability space. This leads to an easier presentation of some of the subsequent results. We state the version of Proposition~3.1 as given in \citet[Lemma~1]{Cea15}:

\begin{lem}\label{lem:Lemma 1}
Suppose that \eqref{eq:R} holds. Then, for any $\eta\in[0,1/2)$ and $T>0$, it holds that, as $n\to\infty$,
\begin{align}
\sup_{x,y\in(0,T]}&y^{-\eta}\Big|\sqrt{k}\big\{T_n(x,y)-R_n(x,y)\big\}-W_R(x,y)\Big|\overset{a.s.}{\longrightarrow}0,\notag\\
\sup_{x\in(0,T]}&x^{-\eta}\Big|\sqrt{k}\big\{T_n(x,\infty)-x\big\}-W_R(x,\infty)\Big|\overset{a.s.}{\longrightarrow}0,\label{eq:help e n 1}\\
\sup_{y\in(0,T]}&y^{-\eta}\Big|\sqrt{k}\big\{T_n(\infty,y)-y\big\}-W_R(\infty,y)\Big|\overset{a.s.}{\longrightarrow}0.\notag
\end{align}
\end{lem}

\begin{lem}\label{lem:gamma}
Under the conditions of Theorem~\ref{thm:main result}, we have that, as $n\to\infty$,
\begin{equation*}
	\sqrt{k_1}\big( \widehat{\gamma}_{1} - \gamma_1\big)\overset{\p}{\longrightarrow}\Gamma.
\end{equation*}
\end{lem}

\begin{proof}
See Appendix~\ref{Proof of lemmas}.
\end{proof}

For the next lemma, we introduce the following (with the exception of $\widehat{\theta}_{k/n}$) infeasible estimators of $\theta_{k/n}$:
\begin{align}
\widetilde{\theta}_{k/n} &= \frac{1}{k}\sum_{t=1}^{n}\varepsilon_{t,Y}^{+}I_{\big\{\varepsilon_{t,X}>\varepsilon_{(k+1),X}\big\}},\notag\\
\widetilde{\theta}_{k/n}^{\ast} &= \frac{1}{k}\sum_{t=1}^{n}\varepsilon_{t,Y}^{+}I_{\big\{\varepsilon_{t,X}>U_0(n/k)\big\}},\notag\\
\widehat{\theta}_{k/n} &= \frac{1}{k}\sum_{t=1}^{n}\widehat{\varepsilon}_{t,Y}^{+}I_{\big\{\widehat{\varepsilon}_{t,X}>\widehat{\varepsilon}_{(k+1),X}\big\}},\notag\\
\widehat{\theta}_{k/n}^{\ast} &= \frac{1}{k}\sum_{t=1}^{n}\widehat{\varepsilon}_{t,Y}^{+}I_{\big\{\widehat{\varepsilon}_{t,X}>U_0(n/k)\big\}}.\notag\\
\intertext{Moreover, define}
e_n & =(n/k)\big\{1-F_0(\varepsilon_{(k+1),X})\big\},\label{eq:(e.n)}\\
\widehat{e}_n & =(n/k)\big\{1-F_0(\widehat{\varepsilon}_{(k+1),X})\big\},\label{eq:(e.hat.n)}
\end{align}
such that $\widehat{e}_n\widehat{\theta}_{k\widehat{e}_n/n}^{\ast}=\widehat{\theta}_{k/n}$ and $e_n\widetilde{\theta}^{\ast}_{ke_n/n}=\widetilde{\theta}_{k/n}$.

\begin{lem}\label{lem:MES final}
Under the conditions of Theorem~\ref{thm:main result}, we have that, as $n\to\infty$,
\begin{equation*}
	\frac{\sqrt{k}}{U_1(n/k)}\big( \widehat{\theta}_{k/n} - \widetilde{\theta}_{k/n} \big)=o_{\p}(1).
\end{equation*}
\end{lem}

\begin{proof}
See Appendix~\ref{Proof of lemmas}.
\end{proof}

Now, we can prove Proposition~\ref{prop:MES}.

\begin{proof}[\textbf{Proof of Proposition~\ref{prop:MES}}]
Write 
\begin{equation*}
\frac{\widehat{\theta}_{p}}{\theta_{p}} = \frac{d_n^{\widehat{\gamma}_1}}{d_n^{\gamma_1}}\cdot\frac{\widehat{\theta}_{k/n}}{\theta_{k/n}^{+}}\cdot\frac{d_n^{\gamma_1}\theta_{k/n}^{+}}{\theta_{p}^{+}}\cdot\frac{\theta_{p}^{+}}{\theta_{p}}=:B_1\cdot B_2\cdot B_3\cdot B_4.
\end{equation*}
First consider $B_1$. By the mean value theorem and $(\partial/\partial x)d^x=d^x\log d$, there exists $\vartheta\in(0,1)$, such that
\begin{align}
	\frac{\sqrt{k_1}}{\log d_n}(B_1-1) &= \frac{\sqrt{k_1}}{\log d_n}(d_n^{\widehat{\gamma}_1-\gamma_1}-d_n^0)\notag\\
	&=\frac{\sqrt{k_1}}{\log d_n}d_n^{\vartheta(\widehat{\gamma}_1-\gamma_1)}\log(d_n)(\widehat{\gamma}_1-\gamma_1)\notag\\
	&=\sqrt{k_1}(\widehat{\gamma}_1-\gamma_1)\exp\big\{\log\big(d_n^{\vartheta(\widehat{\gamma}_1-\gamma_1)}\big)\big\}\notag\\
	&=\sqrt{k_1}(\widehat{\gamma}_1-\gamma_1)\exp\big\{\vartheta(\widehat{\gamma}_1-\gamma_1)\log d_n\big\}\notag\\
	&=\sqrt{k_1}(\widehat{\gamma}_1-\gamma_1)\big\{1+o_{\p}(1)\big\}\notag\\
	&=\Gamma+o_{\p}(1),\label{eq:first Prop 2}
\end{align}
where the second-to-last line follows from $\vartheta(\widehat{\gamma}_1-\gamma_1)\log d_n=O_{\p}(\log d_n/\sqrt{k_1})=o_{\p}(1)$ (from Lemma~\ref{lem:gamma} and our assumption that $\log d_n/\sqrt{k_1}=o(1)$), and the final line follows from Lemma~\ref{lem:gamma}. 

Combine our Lemma~\ref{lem:MES final} with Proposition~3 in \citet{Cea15} to get that
\[
	\sqrt{k}(B_2-1)=\sqrt{k}\Bigg(\frac{\widehat{\theta}_{k/n}}{\theta_{k/n}^{+}}-1\Bigg)\overset{\p}{\longrightarrow}\Theta.
\]

Equation~(32) of \citet{Cea15} yields that
\[
	B_3=1+o\big(1/\sqrt{k}\big).
\]

Finally, the proof of Theorem~2 in \citet{Cea15} shows that
\[
	B_4=1+o\big(1/\sqrt{k}\big).
\]

The rest of the proof follows as that of Theorem~1 in \citet{Cea15}. We give it here for the sake of completeness. Combining the above displays leads to
\begin{align}
&\frac{\widehat{\theta}_{p}}{\theta_{p}}-1 = B_1\cdot B_2\cdot B_3 \cdot B_4 - 1\notag\\
&=\bigg[1+\frac{\log d_n}{\sqrt{k_1}}\Gamma+o_{\p}\Big(\frac{\log d_n}{\sqrt{k_1}}\Big)\bigg]\bigg[1+\frac{\Theta}{\sqrt{k}}+o_{\p}\Big(\frac{1}{\sqrt{k}}\Big)\bigg]\bigg[1+o\Big(\frac{1}{\sqrt{k}}\Big)\bigg]\bigg[1+o\Big(\frac{1}{\sqrt{k}}\Big)\bigg]-1\notag\\
&=\frac{\log d_n}{\sqrt{k_1}}\Gamma+\frac{\Theta}{\sqrt{k}}+o_{\p}\Big(\frac{\log d_n}{\sqrt{k_1}}\Big)+o_{\p}\Big(\frac{1}{\sqrt{k}}\Big).\label{eq:final Prop 2}
\end{align}
Since $\log x\sim x-1$ as $x\to1$, the claimed convergence follows. The variances and covariances of $\Gamma$ and $\Theta$ can be computed by using their definition and exploiting the covariance structure of $W_R(\cdot,\cdot)$ together with Fubini's theorem.
\end{proof}

\section{Proofs of Lemmas~\ref{lem:gamma}--\ref{lem:MES final}}\label{Proof of lemmas}

\renewcommand{\theequation}{D.\arabic{equation}}	
\setcounter{equation}{0}	

Define
\begin{equation*}
	s_n(y) = s_{n,k_1}(y)= \frac{n}{k_1}\Big[1-F_{+}\big(y^{-\gamma_1}U_{+}(n/k_1)\big)\Big],\qquad y>0.
\end{equation*}

\begin{proof}[\textbf{Proof of Lemma~\ref{lem:gamma}}]
As a first step, we show that for any $T>0$
\begin{equation}\label{eq:sn}
\sup_{y\in(0,T]}\Big|\frac{s_n(y)}{y}-1\Big|=o\big(1/\sqrt{k_1}\big).
\end{equation}
Since $U_{1}(s)=U_{+}(s)$ for $s>1/\{1-F_{1}(0)\}$, we have from Assumption~\ref{ass:U} that for any $y_0>0$,
\[
	\sup_{y\geq y_0}\Big|y^{-\gamma_1}\frac{U_{+}(sy)}{U_{+}(s)}-1\Big|=O\{A_1(s)\},\qquad s\to\infty.
\]
Insert $s=n/k_1$ and $y=1/s_n(y)$ in that relation to obtain that
\[
	\sup_{y\in(0,T]}\bigg|\Big(\frac{s_n(y)}{y}\Big)^{\gamma_1}-1\bigg|=O\{A_1(n/k_1)\},\qquad n\to\infty.
\]
From $A_1(n/k_1)=o\big(1/\sqrt{k_1}\big)$ by Assumption~\ref{ass:k} and from a Taylor expansion, \eqref{eq:sn} follows.

Lemma~\ref{lem:Lemma 1} implies for any $\eta\in[0,1/2)$ and any $T>0$ that
\begin{equation}\label{eq:(p.30)}
	\sup_{y\in(0,T]}y^{-\eta}\Big|\sqrt{k}\big\{T_n(\infty,y)-y\big\} - W_R(\infty,y)\Big|\overset{a.s.}{\longrightarrow}0.
\end{equation}
Since $k_1/k\to q\in(0,\infty)$, we may replace $y$ by $y(k_1/k)$ in that relation to obtain
\begin{equation*}
	\sup_{y\in(0,T]}\Big(y\frac{k_1}{k}\Big)^{-\eta}\bigg|\sqrt{k}\Big\{\frac{1}{k}\sum_{t=1}^{n}I_{\big\{F_{+}(\varepsilon_{t,Y}^{+})>1-k_1y/n\big\}}-y\frac{k_1}{k}\Big\} - W_R\Big(\infty,y\frac{k_1}{k}\Big)\bigg|\overset{a.s.}{\longrightarrow}0
\end{equation*}
or, multiplying through with $\sqrt{k/k_1}$,
\begin{equation*}
	\sup_{y\in(0,T]}y^{-\eta}\Big(\frac{k_1}{k}\Big)^{-\eta}\bigg|\sqrt{k_1}\Big\{\frac{1}{k_1}\sum_{t=1}^{n}I_{\big\{F_{+}(\varepsilon_{t,Y}^{+})>1-k_1y/n\big\}}-y\Big\} - \sqrt{\frac{k_1}{k}}W_R\Big(\infty,y\frac{k_1}{k}\Big)\bigg|\overset{a.s.}{\longrightarrow}0.
\end{equation*}
Again, since $k_1/k\to q\in(0,\infty)$, we may drop the term $(k_1/k)^{-\eta}$ to get that
\begin{equation*}
	\sup_{y\in(0,T]}y^{-\eta}\bigg|\sqrt{k_1}\Big\{\frac{1}{k_1}\sum_{t=1}^{n}I_{\big\{F_{+}(\varepsilon_{t,Y}^{+})>1-k_1y/n\big\}}-y\Big\} - \sqrt{\frac{k_1}{k}}W_R\Big(\infty,y\frac{k_1}{k}\Big)\bigg|\overset{a.s.}{\longrightarrow}0.
\end{equation*}
Due to the uniform continuity of the weighted Wiener process, 
\begin{equation*}
	\sup_{y\in(0,T]}y^{-\eta}\bigg|\sqrt{\frac{k}{k_1}}W_R\Big(\infty,y\frac{k_1}{k}\Big) - \frac{1}{\sqrt{q}}W_R(\infty,yq)\bigg|\overset{a.s.}{\longrightarrow}0.
\end{equation*}
The last two displays imply that
\begin{equation}\label{eq:(p.30(2))}
	\sup_{y\in(0,T]}y^{-\eta}\bigg|\sqrt{k_1}\Big\{\frac{1}{k_1}\sum_{t=1}^{n}I_{\big\{F_{+}(\varepsilon_{t,Y}^{+})>1-k_1y/n\big\}}-y\Big\} - q^{-1/2}W_R(\infty,qy)\bigg|\overset{a.s.}{\longrightarrow}0.
\end{equation}
With this and $s_n(y)\to y$, as $n\to\infty$, uniformly in $y\in(0,T]$ from \eqref{eq:sn}, it follows for any $0<T_1<T$ that
\begin{equation}\label{eq:(1.2)}
	\sup_{y\in(0,T_1]}s_n^{-\eta}(y)\bigg|\sqrt{k_1}\Big\{\frac{1}{k_1}\sum_{t=1}^{n}I_{\big\{\varepsilon_{t,Y}^{+}>y^{-\gamma_1}U_{+}(n/k_1)\big\}}-s_n(y)\Big\} - q^{-1/2}W_R\{\infty,q s_n(y)\}\bigg|\overset{a.s.}{\longrightarrow}0.
\end{equation}
Our next goal is to show that $s_n(y)$ can be replaced by $y$ at all three appearances in \eqref{eq:(1.2)}. From \eqref{eq:sn} it follows that, without changing the limit, $s_n(y)$ may be replaced by $y$ in \eqref{eq:(1.2)} at the first two appearances. Finally, the uniform continuity of the weighted Wiener process ensures that, as $n\to\infty$,
\[
	\sup_{y\in(0,T_1]}y^{-\eta}q^{-1/2}\big|W_R\{\infty,qs_n(y)\}-W_R(\infty,qy)\big|\overset{a.s.}{\longrightarrow}0.
\]
Thus, since $T$ and $T_1<T$ can be chosen arbitrarily large, we get for any $y_0>0$ that
\begin{equation*}
\sup_{y\geq y_0}y^{\eta/\gamma_1}\bigg|\sqrt{k_1}\Big\{\frac{1}{k_1}\sum_{t=1}^{n}I_{\big\{\varepsilon_{t,Y}^{+}>y U_{+}(n/k_1)\big\}}-y^{-1/\gamma_1}\Big\}-q^{-1/2}W_R\big(\infty,qy^{-1/\gamma_1}\big)\bigg|\overset{a.s.}{\longrightarrow}0.
\end{equation*}
Because $U_{+}(n/k_1)=U_{1}(n/k_1)>0$ for sufficiently large $n$, we get from this that
\begin{equation}\label{eq:(2.1)}
\sup_{y\geq y_0}y^{\eta/\gamma_1}\Big|\sqrt{k_1}\big\{T_n(y)-y^{-1/\gamma_1}\big\}-q^{-1/2}W_R\big(\infty,qy^{-1/\gamma_1}\big)\Big|\overset{a.s.}{\longrightarrow}0,
\end{equation}
where $T_n(y)=\frac{1}{k_1}\sum_{t=1}^{n}I_{\big\{\varepsilon_{t,Y}>yU_1(n/k_1)\big\}}$. 

Our next goal is to show that $T_n(y)$ in \eqref{eq:(2.1)} can be replaced by $\widehat{T}_n(y)=\frac{1}{k_1}\sum_{t=1}^{n}I_{\big\{\widehat{\varepsilon}_{t,Y}>yU_1(n/k_1)\big\}}$. To task this, let $\delta_n=n^{\iota-\xi}$ with $\iota>0$ chosen sufficiently small to ensure that $\sqrt{k_1}\delta_n=o(1)$ (which is possible due to the Assumption~\ref{ass:k} requirement that $\sqrt{k_1}=O(n^{\widetilde{\alpha}})$). By Proposition~\ref{prop:approx}, $\widehat{\varepsilon}_{t,Y}=\varepsilon_{t,Y}\big\{1+o_{\p}\big(n^{\iota-\xi}\big)\big\}$ uniformly in $t=1,\ldots,n$. Thus, since $U_1(n/k_1)>0$ for sufficiently large $n$ (by Assumption~\ref{ass:U}),
\[
	I_{\big\{\varepsilon_{t,Y}>(1+\delta_n)yU_1(n/k_1)\big\}} \leq I_{\big\{\widehat{\varepsilon}_{t,Y}>yU_1(n/k_1)\big\}} \leq I_{\big\{\varepsilon_{t,Y}>(1-\delta_n)yU_1(n/k_1)\big\}}
\]
holds for all $t=1,\ldots,n$ with probability approaching 1 (w.p.a.~1), as $n\to\infty$. Hence, for any $\varepsilon>0$ we can ensure that $\p\{\mathcal{W}_1\}>1-\varepsilon/3$ for sufficiently large $n$, where
\begin{equation*}
	\mathcal{W}_1:=\Big\{ T_n\big\{(1+\delta_n)y\big\}-T_n(y) \leq \widehat{T}_n(y)-T_n(y) \leq T_n\big\{(1-\delta_n)y\big\}-T_n(y)\Big\}.
\end{equation*}
We now show that
\begin{equation}\label{eq:(3.1)}
	\sup_{y\geq y_0}y^{\eta/\gamma_1}\Big|\sqrt{k_1}\Big[T_n\big\{(1\pm\delta_n)y\big\}-T_n(y)\Big]\Big|\overset{a.s.}{\longrightarrow}0.
\end{equation}
We only show \eqref{eq:(3.1)} for $T_n\big\{(1+\delta_n)y\big\}$, as the proof for $T_n\big\{(1-\delta_n)y\big\}$ is similar. Bound
\begin{align*}
	\sup_{y\geq y_0}&\,y^{\eta/\gamma_1}\Big|\sqrt{k_1}\Big[T_n\big\{(1+\delta_n)y\big\}-T_n(y)\Big]\Big|\\
	&\leq \sup_{y\geq y_0}y^{\eta/\gamma_1}\Big|\sqrt{k_1}\Big[T_n\big\{(1+\delta_n)y\big\}-(1+\delta_n)^{-1/\gamma_1}y^{-1/\gamma_1}\Big] \\
	&\hspace{7cm} - q^{-1/2}W_R\big(\infty,q(1+\delta_n)^{-1/\gamma_1}y^{-1/\gamma_1}\big)\Big|\\
	&\hspace{1cm}+  \sup_{y\geq y_0}y^{\eta/\gamma_1}\Big|\sqrt{k_1}\big[T_n(y)-y^{-1/\gamma_1}\big] - q^{-1/2}W_R\big(\infty, qy^{-1/\gamma_1}\big)\Big|\\
	&\hspace{1cm}+\sup_{y\geq y_0}y^{\eta/\gamma_1}q^{-1/2}\Big|W_R\big(\infty,q(1+\delta_n)^{-1/\gamma_1}y^{-1/\gamma_1}\big) - W_R\big(\infty, qy^{-1/\gamma_1}\big)\Big|\\
	&\hspace{1cm}+\sqrt{k_1}\sup_{y\geq y_0}y^{\eta/\gamma_1}\gamma_1^{-1}\big\{\delta_n+o(\delta_n)\big\}y^{-1/\gamma_1}\\
	&= o(1),
\end{align*}
where we used \eqref{eq:(2.1)}, the uniform continuity of the weighted Wiener process, and the fact that $\sqrt{k_1}\delta_n=o(1)$ by our choice of $\delta_n$. This proves \eqref{eq:(3.1)}.

Now, we may use \eqref{eq:(3.1)} to conclude that for any $\delta>0$ it holds for sufficiently large $n$ that
\begin{align*}
\p&\Big\{\sup_{y\geq y_0}y^{\eta/\gamma_1}\Big|\sqrt{k_1}\big[\widehat{T}_n(y)-T_n(y)\big]\Big|\geq\delta\Big\}\\
&\leq 
\p\Big\{\sup_{y\geq y_0}y^{\eta/\gamma_1}\Big|\sqrt{k_1}\big[\widehat{T}_n(y)-T_n(y)\big]\Big|\geq\delta,\ \mathcal{W}_1\Big\} + \p\big\{\mathcal{W}_1^{C}\big\}\\
&\leq \p\Big\{\sup_{y\geq y_0}y^{\eta/\gamma_1}\Big|\sqrt{k_1}\big[\widehat{T}_n\big\{(1+\delta_n)y\big\}-T_n(y)\big]\Big|\geq\delta,\ \mathcal{W}_1\Big\} \\
&\hspace{1cm}+ \p\Big\{\sup_{y\geq y_0}y^{\eta/\gamma_1}\Big|\sqrt{k_1}\big[\widehat{T}_n\big\{(1-\delta_n)y\big\}-T_n(y)\big]\Big|\geq\delta,\ \mathcal{W}_1\Big\} + \p\big\{\mathcal{W}_1^{C}\big\}\\
&\leq \frac{\varepsilon}{3} + \frac{\varepsilon}{3} + \frac{\varepsilon}{3} = \varepsilon.
\end{align*}
Combine this with \eqref{eq:(2.1)} to obtain that
\begin{equation*}
\sup_{y\geq y_0}y^{\eta/\gamma_1}\Big|\sqrt{k_1}\big[\widehat{T}_n(y)-y^{-1/\gamma_1}\big]-q^{-1/2}W_R\big(\infty,qy^{-1/\gamma_1}\big)\Big|=o_{\p}(1).
\end{equation*}
From this convergence it follows as in the proof of Corollary~1 in \citet{Hog14} that
\begin{multline*}
\sup_{y\geq y_0}y^{\eta/\gamma_1}\bigg|\sqrt{k_1}\bigg[\frac{1}{k_1}\sum_{t=1}^{n}I_{\big\{\widehat{\varepsilon}_{t,Y}>y\widehat{\varepsilon}_{(k_1+1),Y}\big\}}-y^{-1/\gamma_1}\bigg]\\
-q^{-1/2}\Big[W_R\big(\infty,qy^{-1/\gamma_1}\big)-y^{-1/\gamma_1}W_R(\infty,q)\Big]\bigg|=o_{\p}(1).
\end{multline*}
As in \citet[Example~4]{Hog14}, we then obtain that
\begin{align*}
	\sqrt{k_1}(\widehat{\gamma}_1-\gamma_1) &= \sqrt{k_1}\int_{1}^{\infty}\Big[\frac{1}{k_1}\sum_{t=1}^{n}I_{\big\{\widehat{\varepsilon}_{t,Y}>y\widehat{\varepsilon}_{(k_1+1),Y}\big\}}-y^{-1/\gamma_1}\Big]\frac{\D y}{y}\\
	&\overset{\p}{\underset{(n\to\infty)}{\longrightarrow}} q^{-1/2}\int_{1}^{\infty}W_R(\infty,qy^{-1/\gamma_1})\frac{\D y}{y} - q^{-1/2}\int_{1}^{\infty}y^{-1/\gamma_1-1}\D y\cdot W_R(\infty,q)\\
	&=q^{-1/2}\gamma_1\int_{0}^{q}s^{-1}W_R(\infty,s)\D s - q^{-1/2}\gamma_1 W_R(\infty,q)\\
	&=\Gamma,
\end{align*}
where we used the substitution $s=qy^{-1/\gamma_1}$ in the third step.
This ends the proof.
\end{proof}

The proof of Lemma~\ref{lem:MES final} builds on the preliminary Lemmas~\ref{lem:r n pm}--\ref{lem:help e n}. These require the following additional notation:
\begin{align}
	s_n^{\pm}(y)&=s_{n,k}\big(y\{1\pm\delta_n\}^{-1/\gamma_1}\big),\notag\\
	r_n^{\pm}(x)&=\frac{n}{k}\Big[1-F_0\big\{(1\pm\delta_n)U_0(n/[kx])\big\}\Big],\label{eq:rnpm}
\end{align}
where $\delta_n=n^{\iota-\xi}$ as in the above proof of Lemma~\ref{lem:gamma}. Also let $s_n(y)=s_{n,k}(y)$ for brevity from now on.

\begin{lem}\label{lem:r n pm}
Under Assumptions~\ref{ass:U} and~\ref{ass:k}, it holds that, as $n\to\infty$,
\[
	\sup_{x\in[1/2,2]}\big|r_n^{\pm}(x)-x\big|=o\big(1/\sqrt{k}\big).
\]
\end{lem}

\begin{proof}
By the regular variation condition \eqref{eq:U} and \citet[Theorem~1.2.1]{HF06},
\begin{equation}\label{eq:(SA1)}
	\frac{r_{n}^{\pm}(x)}{x}=\frac{\frac{n}{k}\big[1-F_0\{(1\pm\delta_n)U_0(n/[kx])\}\big]}{\frac{n}{k}\big[1-F_0\{U_0(n/[kx])\}\big]}\underset{(n\to\infty)}{\longrightarrow}(1\pm\delta_n)^{-1/\gamma_0}.
\end{equation}
Moreover, Assumption~\ref{ass:U} implies that there exists $t_0>0$ such that for $t>t_0$ and $u>1/2$,
\[
	\Bigg|\frac{u^{-\gamma_0}U_0(tu)/U_0(t)-1}{A_0(t)}\Bigg|<K.
\]
Inserting $t=\frac{n}{kx}$ and $u=x/r_n^{\pm}(x)$ in that relation it follows that
\[
	\Bigg|\frac{\Big(\frac{r_n^{\pm}(x)}{x}\Big)^{\gamma_0}(1\pm\delta_n)-1}{A_0\big(\frac{n}{kx}\big)}\Bigg|<K.
\]
Multiplying through with $|A_0\big(\frac{n}{kx}\big)|$ gives
\[
	\bigg|\Big(\frac{r_n^{\pm}(x)}{x}\Big)^{\gamma_0}-1 \pm \delta_n\Big(\frac{r_n^{\pm}(x)}{x}\Big)^{\gamma_0} \Big|<K\Big|A_0\Big(\frac{n}{kx}\Big)\Big|.
\]
By \eqref{eq:(SA1)}, this implies
\[
	\bigg|\Big(\frac{r_n^{\pm}(x)}{x}\Big)^{\gamma_0}-1\bigg|<K\bigg[\Big|A_0\Big(\frac{n}{kx}\Big)\Big|+\delta_n\bigg].
\]
Use the Taylor expansion
\[
	\Big(\frac{r_n^{\pm}(x)}{x}\Big)^{\gamma_0}-1 = \gamma_0\Big[\frac{r_n^{\pm}(x)}{x}-1\Big]+o\Big(\frac{r_n^{\pm}(x)}{x}-1\Big)
\]
to deduce that
\[
	\Big|\frac{r_n^{\pm}(x)}{x}-1\Big|<K\Bigg[\Big|A_0\Big(\frac{n}{kx}\Big)\Big|+\delta_n\Bigg].
\]
By Assumption~\ref{ass:k} and the Potter bounds \citep[Proposition~B.1.9(5)]{HF06}, it holds for any $\delta>0$ and $n$ sufficiently large,
\begin{align*}
	\sqrt{k}\Big|A_0\Big(\frac{n}{kx}\Big)\Big|&=\sqrt{k}\big|A_0(n/k)\big|\frac{A_0\big(n/[kx]\big)}{A_0(n/k)}\\
	&= o(1)x^{-\rho_0}\max\{x^{\delta},\ x^{-\delta}\}=o(1)
\end{align*}
uniformly in $x\in(0,T]$. Since also $\sqrt{k}\delta_n=o(1)$ by Assumption~\ref{ass:k}, the conclusion follows.
\end{proof}

\begin{lem}\label{lem:MES y}
Under the conditions of Theorem~\ref{thm:main result}, we have that, as $n\to\infty$,
\begin{equation}\label{eq:lem1}
	\sup_{x\in[1/2, 2]}\bigg|\frac{\sqrt{k}}{U_1(n/k)}\Big[x\widehat{\theta}^{\ast}_{kx/n}-x\widetilde{\theta}^{\ast}_{kx/n}\Big]\bigg|=o_{\p}(1).
\end{equation}
\end{lem}

\begin{proof}
Using the substitution $s=U_1(n/k)y^{-\gamma_1}$, we get that
\begin{align*}
	-U_1&(n/k)\int_{0}^{\infty}\widehat{T}_n\{x,s_n(y)\}\D y^{-\gamma_1} \\
	&=-U_1(n/k)\int_{0}^{\infty}\frac{1}{k}\sum_{t=1}^{n}I_{\big\{F_0(\widehat{\varepsilon}_{t,X})>1-kx/n,\ F_+(\widehat{\varepsilon}_{t,Y}^{+})>F_+(y^{-\gamma_1}U_+(n/k))\big\}}\D y^{-\gamma_1} \\
	&=-U_1(n/k)\frac{1}{k}\sum_{t=1}^{n}\int_{0}^{\infty}I_{\big\{\widehat{\varepsilon}_{t,X}>U_0(n/[kx]),\ \widehat{\varepsilon}_{t,Y}^{+}>y^{-\gamma_1}U_+(n/k)\big\}}\D y^{-\gamma_1} \\
	&= \frac{1}{k}\sum_{t=1}^{n}\int_{0}^{\infty}I_{\big\{\widehat{\varepsilon}_{t,X}>U_0(n/[kx]),\ \widehat{\varepsilon}_{t,Y}^{+}>s\big\}}\D s \\
	&= \frac{1}{k}\sum_{t=1}^{n}\int_{0}^{\widehat{\varepsilon}_{t,Y}^{+}}I_{\big\{\widehat{\varepsilon}_{t,X}>U_0(n/[kx])\big\}}\D s \\
	&= \frac{1}{k}\sum_{t=1}^{n}\widehat{\varepsilon}_{t,Y}^{+}I_{\big\{\widehat{\varepsilon}_{t,X}>U_0(n/[kx])\big\}} \\
	&=x\widehat{\theta}^{\ast}_{kx/n}.
\end{align*}
By similar arguments,
\begin{equation}\label{eq:decomp theta ast}
	x\widetilde{\theta}^{\ast}_{kx/n} = -U_1(n/k)\int_{0}^{\infty}T_n\big\{x,s_n(y)\big\}\D y^{-\gamma_1}.
\end{equation}
Thus, in view of \eqref{eq:lem1}, we only have to show that
\begin{equation}\label{eq:lem1 equiv}
	\sup_{x\in[1/2, 2]}\Big|\sqrt{k}\int_{0}^{\infty}\Big[\widehat{T}_n\big\{x,s_n(y)\big\}-T_n\big\{x,s_n(y)\big\}\Big]\D y^{-\gamma_1}\Big|=o_{\p}(1).
\end{equation}
By our choice $\delta_n=n^{\iota-\xi}$ and Proposition~\ref{prop:approx}, the following inequality holds w.p.a.~1, as $n\to\infty$,
\begin{align*}
	\widehat{T}_n\{x, s_n(y)\} &= \frac{1}{k}\sum_{t=1}^{n}I_{\Big\{\widehat{\varepsilon}_{t,X}>U_0\big(\frac{n}{kx}\big),\ \widehat{\varepsilon}_{t,Y}^{+}>U_+\big(\frac{n}{ks_n(y)}\big)\Big\}}\\
	&\leq \frac{1}{k}\sum_{t=1}^{n}I_{\Big\{\varepsilon_{t,X}>(1-\delta_n)U_0\big(\frac{n}{kx}\big),\ \varepsilon_{t,Y}^{+}>(1-\delta_n)U_{+}\big(\frac{n}{ks_n(y)}\big)\Big\}}\\
	&=\frac{1}{k}\sum_{t=1}^{n}I_{\Big\{\varepsilon_{t,X}>U_0\big(\frac{n}{k r_n^{-}(x)}\big),\ \varepsilon_{t,Y}^{+}>U_{+}\big(\frac{n}{ks_n^{-}(y)}\big)\Big\}}\\
	&=T_n\{r_n^{-}(x),\ s_n^{-}(y)\}.
\end{align*}
Similarly, 
\[
	\widehat{T}_n\{x, s_n(y)\}\geq T_n\{r_n^{+}(x), s_n^{+}(y)\}
\]
w.p.a.~1, as $n\to\infty$. Fix some arbitrary $\varepsilon>0$ and define
\begin{multline*}
	\mathcal{W}_2:=\Big\{T_n\big\{r_n^{+}(x), s_n^{+}(y)\big\}-T_n\big\{x, s_n(y)\big\}\leq \widehat{T}_n\big\{x, s_n(y)\big\}-T_n\big\{x, s_n(y)\big\}\\
	\leq T_n\big\{r_n^{-}(x), s_n^{-}(y)\big\}-T_n\big\{x, s_n(y)\big\}\Big\}.
\end{multline*}
Then, by the above, $\p\{\mathcal{W}_2\}>1-\varepsilon$ for sufficiently large $n$. Thus, to prove \eqref{eq:lem1 equiv} it suffices to show that
\[
	\sup_{x\in[1/2, 2]}\Big|\sqrt{k}\int_{0}^{\infty}\Big[T_n\big\{r_n^{\pm}(x), s_n^{\pm}(y)\big\}-T_n\big\{x, s_n(y)\big\}\Big]\D y^{-\gamma_1}\Big|=o_{\p}(1).
\]
We only do so for $T_n\{r_n^{+}(x), s_n^{+}(y)\}$, as the claim for $T_n\{r_n^{-}(x), s_n^{-}(y)\}$ can be established analogously. Decompose
\begin{align*}
	\sup_{x\in[1/2, 2]}&\Big|\sqrt{k}\int_{0}^{\infty}\Big[T_n\big\{r_n^{+}(x), s_n^{+}(y)\big\}-T_n\big\{x, s_n(y)\big\}\Big]\D y^{-\gamma_1}\Big|\\
	&\leq \sup_{x\in[1/2, 2]}\Big|\int_{0}^{\infty}\sqrt{k}\Big[T_n\big\{r_n^{+}(x), s_n^{+}(y)\big\}-R_n\big\{r_n^{+}(x), s_n^{+}(y)\big\}\Big]-W_R(x,y)\D y^{-\gamma_1}\Big|\\
	&\hspace{1cm}+ \sup_{x\in[1/2, 2]}\Big|\int_{0}^{\infty}\sqrt{k}\Big[T_n\big\{x,s_n(y)\big\}-R_n\big\{x,s_n(y)\big\}\Big]-W_R(x,y)\D y^{-\gamma_1}\Big|\\
	&\hspace{1cm}+ \sup_{x\in[1/2, 2]}\Big|\int_{0}^{\infty}\sqrt{k}\Big[R_n\big\{r_n^{+}(x),s_n^{+}(y)\big\}-R(x,y)\Big]\D y^{-\gamma_1}\Big|\\
	&\hspace{1cm}+ \sup_{x\in[1/2, 2]}\Big|\int_{0}^{\infty}\sqrt{k}\Big[R_n\big\{x, s_n(y)\big\}-R(x,y)\Big]\D y^{-\gamma_1}\Big|\\
	&=C_1+C_2+C_3+C_4.
\end{align*}
We show in turn that $C_1,\ldots,C_4$ are asymptotically negligible.

By \eqref{eq:decomp theta ast} and the fact that similarly $x\theta_{kx/n}=-U_1(n/k)\int_{0}^{\infty}R_n\{x, s_n(y)\}\D y^{-\gamma_1}$,
\begin{equation}\label{eq:C2}
	C_2=\sup_{x\in[1/2, 2]}\bigg|\frac{\sqrt{k}}{U_1(n/k)}\big[x\widetilde{\theta}_{kx/n}^{\ast}-x\theta_{kx/n}\big]-\int_{0}^{\infty}W_R(x,y)\D y^{-\gamma_1}\bigg|.
\end{equation}
Thus, $C_2=o_{\p}(1)$ follows from Proposition~2 in \citet{Cea15}. 

Carefully reading the proof of that proposition reveals that $s_n(y)$ in \eqref{eq:C2} (appearing in both $x\widetilde{\theta}_{kx/n}^{\ast}$ and $x\theta_{kx/n}$) can be replaced by $s_n^{+}(y)$ without changing the conclusion that the term is $o_{\p}(1)$. This exploits the fact that $\sup_{y\in(0,1]}s_n(y)/y^{(\gamma_1+\eta_1)}\to0$ for $\eta_1>\gamma_1$ and $\sup_{y\in(0,T]}|s_n(y)-y|\to0$ continue to hold for $s_n^{+}(y)$. Additionally using Lemma~\ref{lem:r n pm} we see that $x$ in \eqref{eq:C2} can be replaced by $r_n^{+}(x)$, such that
\[
	\sup_{x\in[1/2, 2]}\Big|\int_{0}^{\infty}\sqrt{k}\big[T_n\{r_n^{+}(x),s_n^{+}(y)\}-R_n\{r_n^{+}(x), s_n^{+}(y)\}\big]-W_R\{r_n^{+}(x), y\}\D y^{-\gamma_1}\Big|=o_{\p}(1).
\]
Hence, $C_1=o_{\p}(1)$ follows if we can show that
\begin{equation}\label{eq:C1}
	\sup_{x\in[1/2, 2]}\Big|\int_{0}^{\infty}W_R(x,y)-W_R\big\{r_n^{+}(x), y\big\}\D y^{-\gamma_1}\Big|=o_{\p}(1).
\end{equation}
By Corollary~1.11 of \citet{Adl90}, $(x,y)\mapsto W_R(x,y)$ is continuous on $[1/2, 2]\times(0,\infty)$. This implies that 
\[
	[1/2, 2]\ni x\mapsto\int_{0}^{\infty}W_R(x,y)\D y^{-\gamma_1}
\]
is continuous and, hence, uniformly continuous on the bounded interval $[1/2, 2]$. Thus, \eqref{eq:C1} follows from Lemma~\ref{lem:r n pm}, showing that $C_1=o_{\p}(1)$.

The fact that $C_4=o(1)$ follows directly from \citet[p.~439]{Cea15}. 

To prove $C_3=o(1)$, consider the bound
\begin{align*}
C_3 &\leq \sup_{x\in[1/2, 2]}\Big|\int_{0}^{\infty}\sqrt{k}\Big[R_n\big\{r_n^{+}(x), s_n^{+}(y)\big\}-R\big\{r_n^{+}(x), y(1+\delta_n)^{-1/\gamma_1}\big\}\Big]\D y^{-\gamma_1}\Big|\\
&\hspace{1cm} + \sup_{x\in[1/2, 2]}\Big|\int_{0}^{\infty}\sqrt{k}\Big[R\big\{r_n^{+}(x), y(1+\delta_n)^{-1/\gamma_1}\big\}-R\big\{r_n^{+}(x),y\big\}\Big]\D y^{-\gamma_1}\Big|\\
&\hspace{1cm} + \sup_{x\in[1/2, 2]}\Big|\int_{0}^{\infty}\sqrt{k}\Big[R\big\{r_n^{+}(x), y\big\}-R(x,y)\Big]\D y^{-\gamma_1}\Big|\\
&=C_{31}+C_{32}+C_{33}.
\end{align*}
Recalling that $s_n^{+}(y)=s_n\{y(1+\delta_n)^{-1/\gamma_1}\}$, we obtain by a change of variables that
\begin{align*}
	C_{31} &= \sup_{x\in[1/2, 2]}\Big|\int_{0}^{\infty}\sqrt{k}\Big[R_n\big\{r_n^{+}(x), s_n^{+}(y)\big\}-R\big\{r_n^{+}(x), y(1+\delta_n)^{-1/\gamma_1}\big\}\Big]\D y^{-\gamma_1}\Big|\\
	&=(1+\delta_n)^{-1}\sup_{x\in[1/2, 2]}\Big|\int_{0}^{\infty}\sqrt{k}\Big[R_n\big\{r_n^{+}(x), s_n(y)\big\}-R\big\{r_n^{+}(x), y\big\}\Big]\D y^{-\gamma_1}\Big|\\
	&=o(1),
\end{align*}
where we used the fact that $C_4=o(1)$ together with Lemma~\ref{lem:r n pm} in the final step. For $C_{32}$ we again use a change of variables to obtain
\[
	\int_{0}^{\infty}R\big\{r_n^{+}(x), y(1+\delta_n)^{-1/\gamma_1}\big\}\D y^{-\gamma_1}=(1+\delta_n)^{-1}\int_{0}^{\infty}R\big\{r_n^{+}(x), y\big\}\D y^{-\gamma_1}.
\]
Thus, since $(1+\delta_n)^{-1}=1-\delta_n+o(\delta_n)$ from a Taylor expansion,
\[
	C_{32}\leq K\sqrt{k}\delta_n\sup_{x\in[1/2, 2]}\int_{0}^{\infty}R\{r_n^{+}(x), y\}\D y^{-\gamma_1}=o(1).
\]
Finally, \citet[Theorem~1~(ii)]{SS06} establish homogeneity of the $R$-function, i.e., $R(sx, sy)=sR(x,y)$ for all $s>0$ and $x,y\geq0$. Using this and a change of variables,
\begin{align*}
\int_{0}^{\infty}\sqrt{k}&\Big[R\big\{r_n^{+}(x), y\big\}-R(x,y)\Big]\D y^{-\gamma_1}\\
&= \sqrt{k}\bigg[\int_{0}^{\infty}R\big\{r_n^{+}(x), y\big\}\D y^{-\gamma_1}-\int_{0}^{\infty}R(x,y)\D y^{-\gamma_1}\bigg]\\
&= \sqrt{k}\bigg[\int_{0}^{\infty}r_n^{+}(x)R\big\{1, y/r_n^{+}(x)\big\}\D y^{-\gamma_1}-\int_{0}^{\infty}xR(1,y/x)\D y^{-\gamma_1}\bigg]\\
&=\sqrt{k}\big[r_n^{+}(x)^{1-\gamma_1} - x^{1-\gamma_1}\big]\int_{0}^{\infty}R(1,y)\D y^{-\gamma_1}\\
&=o(1)
\end{align*}
uniformly in $x\in[1/2, 2]$ by Lemma~\ref{lem:r n pm}. Thus, $C_{33}=o(1)$, and the conclusion follows.
\end{proof}

\begin{lem}\label{lem:help e n}
Under the conditions of Theorem~\ref{thm:main result}, we have that, as $n\to\infty$,
\begin{align}
	\sqrt{k}(e_n-1)&\overset{\p}{\longrightarrow}-W_R(1,\infty),\label{eq:help e n}\\
	\sqrt{k}(\widehat{e}_n-1)&\overset{\p}{\longrightarrow}-W_R(1,\infty),\label{eq:help hat e n}
\end{align}
where $e_n$ and $\widehat{e}_n$ are defined in \eqref{eq:(e.n)} and \eqref{eq:(e.hat.n)}, respectively.
\end{lem}

\begin{proof}
From \eqref{eq:help e n 1},
\begin{equation}\label{eq:(p.37.0)}
\sup_{x\in(0,T]}\bigg|\sqrt{k}\Big[\frac{1}{k}\sum_{t=1}^{n}I_{\big\{\varepsilon_{t,X}> U(n/[kx])\big\}}-x\Big]-W_R(x,\infty)\bigg|\overset{a.s.}{\longrightarrow}0.
\end{equation}
Put $T_{n,1}(x)=\frac{1}{k}\sum_{t=1}^{n}I_{\big\{\varepsilon_{t,X}> U(n/[kx])\big\}}$ for short. The proof draws heavily on Example~A.0.3 in \citet{HF06}. Define
\[
	T_{n,1}^{\leftsquigarrow}(x):=\frac{n}{k}\Big[1-F_0\big(\varepsilon_{(\lfloor kx\rfloor+1),X}\big)\Big].
\]
Then, 
\[
	\sup_{x\in(0,T]}\Big|T_{n,1}\big\{T_{n,1}^{\leftsquigarrow}(x)\big\} - x\Big|=\sup_{x\in(0,T]}\Big|\frac{\lfloor kx\rfloor}{k}-x\Big|\leq\frac{1}{k},
\]
which implies
\begin{equation}\label{eq:(p.37.1)}
\sup_{x\in(0,T]}\big|T_{n,1}^{\leftsquigarrow}(x) - T_{n,1}(x)\big|=o\big(1/\sqrt{k}\big).
\end{equation}
From \eqref{eq:(p.37.0)} we obtain via \citeauthor{Ver72}'s \citeyearpar{Ver72} lemma that
\[
	\sup_{x\in[1/2,2]}\Big|\sqrt{k}\big[T_{n,1}^{\leftarrow}(x)-x\big]+W_R(x,\infty)\Big|\overset{a.s.}{=}o(1).
\]
Due to \eqref{eq:(p.37.1)} this implies
\[
	\sup_{x\in[1/2,2]}\Big|\sqrt{k}\big[T_{n,1}^{\leftsquigarrow}(x)-x\big]+W_R(x,\infty)\Big|\overset{a.s.}{=}o(1).
\]
Putting $x=1$ in that expression, \eqref{eq:help e n} follows.

Arguing similarly as in the proof of Lemma~\ref{lem:gamma}, we may show that \eqref{eq:help e n 1} remains valid if the $\varepsilon_{t,X}$ in $T_n(x,\infty)$ are replaced by $\widehat{\varepsilon}_{t,X}$. Hence, \eqref{eq:help hat e n} follows as before.
\end{proof}

Now, we are in a position to prove Lemma~\ref{lem:MES final}:

\begin{proof}[\textbf{Proof of Lemma~\ref{lem:MES final}}]
Recall that $\widehat{e}_n\widehat{\theta}_{k\widehat{e}_n/n}^{\ast}=\widehat{\theta}_{k/n}$ and $e_n\widetilde{\theta}^{\ast}_{ke_n/n}=\widetilde{\theta}_{k/n}$. With this, write
\begin{align*}
 \frac{\sqrt{k}}{U_1(n/k)}\Big[\widehat{\theta}_{k/n} - \widetilde{\theta}_{k/n}\Big] &=\frac{\sqrt{k}}{U_1(n/k)}\Big[\widehat{e}_n\widehat{\theta}_{k\widehat{e}_n/n}^{\ast} - e_n\widetilde{\theta}^{\ast}_{ke_n/n}\Big]\\
	&= \frac{\sqrt{k}}{U_1(n/k)}\Big[\widehat{e}_n\widehat{\theta}_{k\widehat{e}_n/n}^{\ast} - \widehat{e}_n\widetilde{\theta}_{k\widehat{e}_n/n}^{\ast} + \widehat{e}_n\widetilde{\theta}_{k\widehat{e}_n/n}^{\ast}	- e_n\widetilde{\theta}^{\ast}_{ke_n/n}\Big]\\
	&= o_{\p}(1)+\frac{\sqrt{k}}{U_1(n/k)}\Big[\widehat{e}_n\widetilde{\theta}_{k\widehat{e}_n/n}^{\ast}	- e_n\widetilde{\theta}^{\ast}_{ke_n/n}\Big],
\end{align*}
where the final line follows from Lemmas~\ref{lem:MES y} and~\ref{lem:help e n}. Decompose the remaining term as follows:
\begin{align*}
\frac{\sqrt{k}}{U_1(n/k)}\Big[\widehat{e}_n\widetilde{\theta}_{k\widehat{e}_n/n}^{\ast}	- e_n\widetilde{\theta}^{\ast}_{ke_n/n}\Big]&= \frac{\sqrt{k}}{U_1(n/k)}e_n\Big[\widetilde{\theta}_{k\widehat{e}_n/n}^{\ast}	- \widetilde{\theta}^{\ast}_{ke_n/n}\Big] + \frac{\sqrt{k}}{U_1(n/k)}\big[\widehat{e}_n-e_n\big]\widetilde{\theta}^{\ast}_{k\widehat{e}_n/n}\\
&=:D_1 + D_2.
\end{align*}
Consider $D_1$ and $D_2$ separately. For $D_1$, we get that
\begin{align}
D_1 &= e_n\bigg\{\frac{\sqrt{k}}{U_1(n/k)}\big[\widetilde{\theta}_{k\widehat{e}_n/n}^{\ast}-\theta_{k\widehat{e}_n/n}\big]+\int_{0}^{\infty}W_R(1, s)\D s^{-\gamma_1}\bigg\}\notag\\
&\hspace{1cm}- e_n\bigg\{\frac{\sqrt{k}}{U_1(n/k)}\big[\widetilde{\theta}_{ke_n/n}^{\ast}-\theta_{ke_n/n}\big]+\int_{0}^{\infty}W_R(1,s)\D s^{-\gamma_1}\bigg\}\notag\\
&\hspace{1cm}+e_n\frac{\sqrt{k}}{U_1(n/k)}\big[\theta_{k\widehat{e}_n/n}- \theta_{ke_n/n} \big]\notag\\
&=e_n\frac{\sqrt{k}}{U_1(n/k)}\big[\theta_{k\widehat{e}_n/n}- \theta_{ke_n/n} \big] + o_{\p}(1),\label{eq:An}
\end{align}
where we have used that $e_n=1+o_{\p}(1)$ from Lemma~\ref{lem:help e n} and $\frac{\sqrt{k}}{U_1(n/k)}\big[\widetilde{\theta}_{ke_n/n}^{\ast}-\theta_{ke_n/n}\big]\overset{\p}{\longrightarrow}-\int_{0}^{\infty}W_R(1,s)\D s^{-\gamma_1}$ from the proof of Proposition~3 in \citet{Cea15}. The fact that this convergence also holds when replacing $e_n$ with $\widehat{e}_n$ follows similarly, since $e_n$ and $\widehat{e}_n$ are asymptotically equivalent by Lemma~\ref{lem:help e n}. For the remaining term, we need the following result from \citet[p.~439]{Cea15}:
\[
	e_n\theta_{ke_n/n}=e_n^{1-\gamma_1}\theta_{k/n} + o_{\p}\big\{U_1(n/k)/\sqrt{k}\big\}.
\]
Again from Lemma~\ref{lem:help e n}, it follows that this continues to hold upon replacing $e_n$ with $\widehat{e}_n$, such that
\[
	\widehat{e}_n\theta_{k\widehat{e}_n/n}=\widehat{e}_n^{1-\gamma_1}\theta_{k/n} + o_{\p}\big\{U_1(n/k)/\sqrt{k}\big\}.
\]
Using these two results, we deduce that
\begin{align}
	e_n	\frac{\sqrt{k}}{U_1(n/k)}\big[\theta_{k\widehat{e}_n/n}- \theta_{ke_n/n} \big] 
			&=\frac{\sqrt{k}}{U_1(n/k)}\Big[\frac{e_n}{\widehat{e}_n}\big\{\widehat{e}_n\theta_{k\widehat{e}_n/n}\big\}- \theta_{k/n} \Big] \notag\\
			&\hspace{1cm}-  \frac{\sqrt{k}}{U_1(n/k)}\Big[e_n\theta_{ke_n/n}- \theta_{k/n} \Big]\notag\\
			&= \frac{\sqrt{k}}{U_1(n/k)}\Big[\frac{e_n}{\widehat{e}_n}\Big\{\widehat{e}_n^{1-\gamma_1}\theta_{k/n} + o_{\p}\big(U_1(n/k)/\sqrt{k}\big)\Big\} - \theta_{k/n} \Big] \notag\\
			&\hspace{1cm}-  \frac{\sqrt{k}}{U_1(n/k)}\Big[e_n^{1-\gamma_1}\theta_{k/n} + o_{\p}\big(U_1(n/k)/\sqrt{k}\big) - \theta_{k/n} \Big]\notag\\
			&= \frac{\sqrt{k}}{U_1(n/k)}\Big[\frac{e_n}{\widehat{e}_n}\big(\widehat{e}_n^{1-\gamma_1}-1\big)\theta_{k/n}\Big]  - \frac{\sqrt{k}}{U_1(n/k)}\Big[\frac{e_n}{\widehat{e}_n}-1\Big]\theta_{k/n}\notag\\
			&\hspace{1cm}-  \frac{\sqrt{k}}{U_1(n/k)}\Big[\big(e_n^{1-\gamma_1}-1\big)\theta_{k/n} \Big]+o_{\p}(1)\notag\\
			&= \frac{e_n}{\widehat{e}_n}\cdot\sqrt{k}\big(\widehat{e}_n^{1-\gamma_1}-1\big)\cdot\frac{\theta_{k/n}}{U_1(n/k)}-\frac{1}{\widehat{e}_n}\cdot\sqrt{k}\big[e_n-\widehat{e}_n\big]\cdot\frac{\theta_{k/n}}{U_1(n/k)}\notag\\
			&\hspace{1cm}-\sqrt{k}\big(e_n^{1-\gamma_1}-1\big)\cdot\frac{\theta_{k/n}}{U_1(n/k)}+o_{\p}(1)\notag\\
			&=o_{\p}(1),\label{eq:An help}
\end{align}
where the final step follows from 
\begin{align*}
\sqrt{k}\big(\widehat{e}_n^{1-\gamma_1}-1\big)&\overset{\p}{\longrightarrow}(\gamma_1-1)W_R(1,\infty),\\
\sqrt{k}\big(e_n^{1-\gamma_1}-1\big)&\overset{\p}{\longrightarrow}(\gamma_1-1)W_R(1,\infty)
\end{align*}
(as a consequence of the delta method applied to Lemma~\ref{lem:help e n}) and $\theta_{k/n}/U_1(n/k)\to\int_{0}^{\infty}R(1,s^{-1/\gamma_1})\D s$ \citep[from Proposition~1 of][]{Cea15}. From \eqref{eq:An} and \eqref{eq:An help}, $D_1=o_{\p}(1)$. 

For $D_2$, write
\begin{align*}
	D_2 &= \big[\sqrt{k}(e_n-1)-\sqrt{k}(\widehat{e}_n-1)\big]\cdot\frac{\theta_{k/n}}{U_1(n/k)}\cdot\frac{1}{\widehat{e}_n}\cdot\frac{\widehat{e}_n\widetilde{\theta}^{\ast}_{k\widehat{e}_n/n}}{\theta_{k/n}}\\
	&=\big[-W_R(1,\infty)+W_R(1,\infty)+o_{\p}(1)\big]\cdot\Big[\int_{0}^{\infty}R(1,s^{-1/\gamma_1})\D s+o(1)\Big]\cdot\big[1+o_{\p}(1)\big]\\
	&\hspace{3cm}\cdot\Big[\frac{1}{\sqrt{k}}\frac{U_1(n/k)}{\theta_{k/n}}\frac{\sqrt{k}}{U_1(n/k)}\big(\widehat{e}_n\widetilde{\theta}^{\ast}_{k\widehat{e}_n/n}-\theta_{k/n}\big)+1\Big]\\
	&=o_{\p}(1)\cdot\Bigg[\frac{1}{\sqrt{k}}\Big\{\int_{0}^{\infty}R(1,s^{-1/\gamma_1})\D s+o(1)\Big\}^{-1}\Theta\int_{0}^{\infty}R(1,s^{-1/\gamma_1})\D s + 1\Bigg]\\
	&=o_{\p}(1),
\end{align*}
where we used in the second-to-last step that \eqref{eq:MES} holds and that
\begin{align*}
	\frac{\sqrt{k}}{U_1(n/k)}\big(\widehat{e}_n\widetilde{\theta}^{\ast}_{k\widehat{e}_n/n}-\theta_{k/n}\big) \overset{\p}{\longrightarrow}\Theta\int_{0}^{\infty}R(1,s^{-1/\gamma_1})\D s,
\end{align*}
which follows from the proof of Proposition~3 in \citet{Cea15} (together with Lemma~\ref{lem:help e n}). The conclusion of the lemma follows.
\end{proof}

\section{Proof of Theorem~\ref{thm:main result2}}
\label{sec:Proof of Theorem 2}

\renewcommand{\theequation}{E.\arabic{equation}}	
\setcounter{equation}{0}

\begin{proof}[\textbf{Proof of Theorem~\ref{thm:main result2}}]
It follows similarly as in the proof of Theorem~\ref{thm:main result} that 
\begin{equation*}
\min\big\{\sqrt{k}, \sqrt{k_d}/\log(d_n)\big\}\log\Bigg(\frac{\widehat{\theta}_{n,p,d}}{\theta_{n,p,d}}\Bigg)
=\min\big\{\sqrt{k}, \sqrt{k_d}/\log(d_n)\big\}\log\Bigg(\frac{\widehat{\theta}_{p,d}}{\theta_{p,d}}\Bigg) + o_{\p}(1).
\end{equation*}
From analogous arguments used in the proof of Proposition~\ref{prop:MES} (see in particular \eqref{eq:first Prop 2} and \eqref{eq:final Prop 2}), we obtain for all $d=1,\ldots,D$ that
\[
	\min\big\{\sqrt{k}, \sqrt{k_d}/\log(d_n)\big\}\log\Bigg(\frac{\widehat{\theta}_{p,d}}{\theta_{p,d}}\Bigg) = \sqrt{k_d}(\widehat{\gamma}_d-\gamma_d)+o_{\p}(1)= \Gamma_d + o_{\p}(1),
\]
where $\Gamma_d\sim N(0,\gamma_d^2)$.
Next, we show that, as $n\to\infty$,
\[
	\big(\sqrt{k_1}\{\widehat{\gamma}_1-\gamma_1\},\ldots,\sqrt{k_D}\{\widehat{\gamma}_D-\gamma_D\}\big)^\prime\overset{d}{\longrightarrow}(\Gamma_1, \ldots, \Gamma_D)^\prime,
\]
where $\Cov(\Gamma_i,\Gamma_j) = \sigma_{i,j}$ for $i,j=1,\ldots,D$. To that end, we apply similar arguments as used in the proof of Proposition~3 in \citet{Hog17a+}. Therefore, we have to verify his conditions (M1)--(M4). Note that we cannot directly apply his Proposition~3, because it derives the joint limit of the Hill estimates only for a common intermediate sequence, whereas we allow for possibly distinct $k_1,\ldots,k_D$.

Since the $\vvarepsilon_t$ are i.i.d., the $\beta$-mixing condition (M1) is immediate for any $r_n\to\infty$ in the notation of \citet{Hog17a+}. In the following, we let $r_n=n$. 

For (M2), note by independence of the $\vvarepsilon_t$ that
\begin{align*}
	\frac{n}{r_n k}\Cov&\bigg(\sum_{t=1}^{r_n}I_{\big\{\varepsilon_{t,Y_i}>U_{i}(\frac{n}{k x})\big\}}, \sum_{t=1}^{r_n}I_{\big\{\varepsilon_{t,Y_j}>U_{j}(\frac{n}{k y})\big\}}\bigg)\\
	&=\frac{n}{r_n k}\sum_{t=1}^{r_n}\Cov\Big(I_{\big\{\varepsilon_{t,Y_i}>U_{i}(\frac{n}{k x})\big\}}, I_{\big\{\varepsilon_{t,Y_j}>U_{j}(\frac{n}{k y})\big\}}\Big)\\
	&= \frac{n}{r_n k}\sum_{t=1}^{r_n}\bigg[\p\Big\{\varepsilon_{t,Y_i}>U_{i}\Big(\frac{n}{k x}\Big),\ \varepsilon_{t,Y_j}>U_{j}\Big(\frac{n}{k y}\Big)\Big\} - \frac{kx}{n}\frac{ky}{n}\bigg]\\
	&= \frac{n}{k}\p\Big\{\varepsilon_{t,Y_i}>U_{i}\Big(\frac{n}{k x}\Big),\ \varepsilon_{t,Y_j}>U_{j}\Big(\frac{n}{k y}\Big)\Big\} + \frac{n}{k}\frac{k^2 xy}{n^2}\\
	&= R_{i,j}(x,y) + o(1) + o(k/n)\underset{(n\to\infty)}{\longrightarrow} R_{i,j}(x,y),
\end{align*}
where the final line uses Assumption~\ref{ass:R*}*. This establishes (M2).

For (M3) note that since the $\vvarepsilon_t$ are i.i.d., their $\rho$-mixing coefficients are trivially zero, such that by Lemma~2.3 of \citet{Sha93} (for $q=4$ in his notation)
\begin{align*}
	&\frac{n}{r_n k}\E\bigg[\sum_{t=1}^{r_n}I_{\big\{U_{d}(\frac{n}{k y})<\varepsilon_{t,Y_d}\leq U_{d}(\frac{n}{k x})\big\}}\bigg]^4\\
	&\leq \frac{n}{r_n k}K\bigg\{r_n^2 \E^2\Big[I_{\big\{U_{d}(\frac{n}{k y})<\varepsilon_{t,Y_d}\leq U_{d}(\frac{n}{k x})\big\}}^2\Big] + r_n \E\Big[I_{\big\{U_{d}(\frac{n}{k y})<\varepsilon_{t,Y_d}\leq U_{d}(\frac{n}{k x})\big\}}^4\Big]\bigg\}\\
	&\leq \frac{n}{r_n k}K\bigg\{r_n^2 \p^2\Big\{U_{d}\Big(\frac{n}{k y}\Big)<\varepsilon_{t,Y_d}\leq U_{d}\Big(\frac{n}{k x}\Big)\Big\} + r_n\p\Big\{U_{d}\Big(\frac{n}{k y}\Big)<\varepsilon_{t,Y_d}\leq U_{d}\Big(\frac{n}{k x}\Big)\Big\} \bigg\}\\
	&=\frac{n}{r_n k}K\bigg\{r_n^2 \frac{k^2}{n^2}(y-x)^2 + r_n\frac{k}{n}(y-x) \bigg\}\\
	&\leq K\frac{n}{r_n k}(y-x)^2 + K(y-x)\underset{(n\to\infty)}{\longrightarrow}K(y-x),
\end{align*}
since $r_n=n$ and $k\to\infty$. This implies (M3). 

Finally, condition (M4) follows from Assumption~\ref{ass:U*}*.

Having verified the conditions of Proposition~3 of \citet{Hog17a+}, we may follow the steps in that proof to derive that for any $\eta\in[0,1/2)$,
\begin{equation}\label{eq:(p.40)}
	\sup_{y\in(0,T]}y^{-\eta}\left\Vert\sqrt{k}\begin{pmatrix}T_{n,1}(y)-y\\ \vdots \\ T_{n,D}(y)-y\end{pmatrix} - \begin{pmatrix}W_1(y)\\ \vdots \\ W_D(y)\end{pmatrix}\right\Vert\overset{a.s.}{\longrightarrow}0,
\end{equation}
where $T_{n,d}(y):=\frac{1}{k}\sum_{t=1}^{n}I_{\big\{\varepsilon_{t,Y_d}^{+}>U_{d,+}(n/[ky])\big\}}$ with $U_{d,+}(\cdot)$ defined in the obvious way, and $\mW(y)=\big(W_1(y),\ldots, W_D(y)\big)^\prime$ is a $D$-variate continuous, zero-mean Gaussian process with covariance function
\[
	\Cov\big(\mW(y_1), \mW(y_2)\big) = \bigg[\frac{R_{i,j}(y_1,y_2) + R_{i,j}(y_2,y_1)}{2}\bigg]_{i,j=1,\ldots,D};
\]
see in particular Lemma~2 of \citet{Hog17a+}. Equation \eqref{eq:(p.40)} is the analog of \eqref{eq:(p.30)} in the proof of Lemma~\ref{lem:gamma}.

Applying the steps in the proof of Lemma~\ref{lem:gamma} to each of the components in \eqref{eq:(p.40)} gives, as $n\to\infty$,
\[
	\sqrt{k_d}(\widehat{\gamma}_d-\gamma_d) \overset{\p}{\longrightarrow}\gamma_d q_d^{-1/2}\bigg[\int_{0}^{1}u^{-1}W_d(q_d u )\D u - W_d(q_d)\bigg]=\Gamma_d
\]
for each $d=1,\ldots,D$. Since $\E[\Gamma_d]=0$, we obtain for $i,j=1,\ldots,D$ using Fubini's theorem that
\begin{align*}
	\Cov&(\Gamma_i, \Gamma_j)= \E[\Gamma_i\Gamma_j]\\
	&= \frac{\gamma_i\gamma_j}{\sqrt{q_i q_j}}\E\bigg[\Big\{\int_{0}^{1}u^{-1}W_i(q_i u )-W_i(q_i)\D u\Big\}\Big\{\int_{0}^{1}v^{-1}W_j(q_j v )- W_j(q_j)\D v \Big\}\bigg]\\
	&= \frac{\gamma_i\gamma_j}{\sqrt{q_i q_j}}\E\Big[\int_{0}^{1}\int_{0}^{1}\big\{u^{-1}W_i(q_i u ) - W_i(q_i)\big\}\big\{v^{-1}W_j(q_j v ) - W_j(q_j)\big\}\D u \D v\Big]\\
	&= \frac{\gamma_i\gamma_j}{\sqrt{q_i q_j}}\int_{0}^{1}\int_{0}^{1}\E\Big[\big\{u^{-1}W_i(q_i u ) - W_i(q_i)\big\}\big\{v^{-1}W_j(q_j v ) - W_j(q_j)\big\}\Big]\D u \D v\\
	&=\frac{\gamma_i\gamma_j}{\sqrt{q_i q_j}}\int_{0}^{1}\int_{0}^{1}\frac{\E\big[W_i(q_i u )W_j(q_j v )\big]}{uv} - \frac{\E\big[W_i(q_i u )W_j(q_j)\big]}{u} \\
	&\hspace{4cm}- \frac{\E\big[W_i(q_i)W_j(q_j v)\big]}{v} +\E\big[W_i(q_i)W_j(q_j)\big] \D u \D v\\
	&= \frac{\gamma_i\gamma_j}{\sqrt{q_i q_j}}\int_{0}^{1}\int_{0}^{1}\frac{R_{i,j}(q_i u, q_j v) + R_{i,j}(q_j v, q_i u)}{2uv} - \frac{R_{i,j}(q_i u, q_j) + R_{i,j}(q_j, q_i u)}{2u}\\
	&\hspace{4cm}- \frac{R_{i,j}(q_i, q_j v) + R_{i,j}(q_j v, q_i)}{2v} + \frac{R_{i,j}(q_i, q_j) + R_{i,j}(q_j, q_i)}{2}\D u \D v.
\end{align*}
Due to Theorem~1~(ii) of \citet{SS06}, the $R_{i,j}(\cdot,\cdot)$-function is homogeneous (i.e., $R_{i,j}(sx,sy)=sR_{i,j}(x,y)$ for all $s>0,x,y\geq0$). Hence,
\begin{align*}
\int_{0}^{1}&\int_{0}^{1}\frac{R_{i,j}(q_i u, q_j v)}{2uv}\D u\D v\\
 &= \int_{0}^{1}\Bigg[\int_{0}^{v}\frac{R_{i,j}(q_i u,q_j v)}{2uv}\D u\Bigg]\D v + \int_{0}^{1}\Bigg[\int_{0}^{u}\frac{R_{i,j}(q_i u,q_j v)}{2uv}\D v\Bigg]\D u\\
&=\int_{0}^{1}\Bigg[\int_{0}^{v}\frac{R_{i,j}(q_i u/v,q_j )}{2u}\D u\Bigg]\D v + \int_{0}^{1}\Bigg[\int_{0}^{u}\frac{R_{i,j}(q_i,q_j v/u)}{2v}\D v\Bigg]\D u\\
&=\int_{0}^{1}\Bigg[\int_{0}^{1}\frac{R_{i,j}(q_i u,q_j )}{2u}\D u\Bigg]\D v + \int_{0}^{1}\Bigg[\int_{0}^{1}\frac{R_{i,j}(q_i,q_j v)}{2v}\D v\Bigg]\D u\\
&=\int_{0}^{1}\frac{R_{i,j}(q_i u,q_j )}{2u}\D u + \int_{0}^{1}\frac{R_{i,j}(q_i,q_j v)}{2v}\D v.
\end{align*}
Similar arguments yield that
\[
	\int_{0}^{1}\int_{0}^{1}\frac{R_{i,j}(q_j v, q_i u)}{2uv}\D u\D v=\int_{0}^{1}\frac{R_{i,j}(q_j u,q_i)}{2u}\D u + \int_{0}^{1}\frac{R_{i,j}(q_j,q_i v )}{2v}\D v.
\]
Therefore, 
\[
	\Cov(\Gamma_i, \Gamma_j) = \frac{\gamma_i\gamma_j}{\sqrt{q_i q_j}}\frac{R_{i,j}(q_i,q_j) + R_{i,j}(q_j,q_i)}{2},
\]
as claimed. This finishes the proof.
\end{proof}

\section{Proof of Proposition~\ref{prop:cons}}
\label{sec:Proof of Proposition 1}

\renewcommand{\theequation}{F.\arabic{equation}}	
\setcounter{equation}{0}

Here, we only show that $\widehat{R}_{i,j}(q_i,q_j)\overset{\p}{\longrightarrow}R_{i,j}(q_i,q_j)$, since the other claim of the proposition can be established analogously. The proof requires the preliminary Lemmas~\ref{lem:Lem 1}--\ref{lem:3 tilde} for which we have to introduce some additional notation.

Fix $i,j\in\{1,\ldots,D\}$, and define 
\begin{align*}
	\widetilde{\varphi}_{k/n}(x,y)&=\frac{1}{k}\sum_{t=1}^{n}I_{\big\{\varepsilon_{t,Y_i}>\varepsilon_{(\lfloor kx\rfloor+1),Y_i},\ \varepsilon_{t,Y_j}>\varepsilon_{(\lfloor ky\rfloor+1),Y_j}\big\}},\\
	\widetilde{\varphi}_{k/n}^{\ast}(x,y)&=\frac{1}{k}\sum_{t=1}^{n}I_{\big\{\varepsilon_{t,Y_i}>U_i(n/[kx]),\ \varepsilon_{t,Y_j}>U_j(n/[ky])\big\}},\\	
	\widehat{\varphi}_{k/n}(x,y)&=\frac{1}{k}\sum_{t=1}^{n}I_{\big\{\widehat{\varepsilon}_{t,Y_i}>\widehat{\varepsilon}_{(\lfloor kx\rfloor+1),Y_i},\ \widehat{\varepsilon}_{t,Y_j}>\widehat{\varepsilon}_{(\lfloor ky\rfloor+1),Y_j}\big\}},\\
	\widehat{\varphi}_{k/n}^{\ast}(x,y)&=\frac{1}{k}\sum_{t=1}^{n}I_{\big\{\widehat{\varepsilon}_{t,Y_i}>U_i(n/[kx]),\ \widehat{\varepsilon}_{t,Y_j}>U_j(n/[ky])\big\}}
\end{align*}
and
\begin{align*}
		e_{n,i} 				&= \frac{n}{k}\Big[1-F_i\big(\varepsilon_{(k_i+1),Y_i}\big)\Big],\\
	\widehat{e}_{n,i} &= \frac{n}{k}\Big[1-F_i\big(\widehat{\varepsilon}_{(k_i+1),Y_i}\big)\Big].
\end{align*}
In analogy to $R_n(x,y)$ in Appendix~\ref{Proof of Proposition 2}, we also define
\[
	R_{n,i,j}(x,y)=\frac{n}{k}\p\big\{F_{i}(\varepsilon_{t,Y_i})>1-kx/n,\ F_j(\varepsilon_{t,Y_j})>1-ky/n\big\}.
\]

We have the following analog to Lemma~\ref{lem:Lemma 1}, which again follows from Proposition~3.1 of \citet{EHL06}.

\begin{lem}\label{lem:Lem 1}
Suppose that \eqref{eq:Rij} holds. Then, for any $\eta\in[0,1/2)$ and $T>0$, it holds that, as $n\to\infty$,
\begin{align}
	\sup_{x,y\in(0,T]}&y^{-\eta}\Big|\sqrt{k}\big\{\widetilde{\varphi}_{k/n}^{\ast}(x,y) - R_{n,i,j}(x,y)\big\} - W_{i,j}(x,y)\Big|\overset{a.s.}{\longrightarrow}0,\notag\\
		\sup_{x\in(0,T]}&x^{-\eta}\Big|\sqrt{k}\big\{\widetilde{\varphi}_{k/n}^{\ast}(x,\infty) - x\big\} - W_{i,j}(x,\infty)\Big|\overset{a.s.}{\longrightarrow}0,\label{eq:(C.1tilde)}
\end{align}
where $W_{i,j}(\cdot,\cdot)$ is a zero-mean Gaussian process with covariance structure given by
\[
	\E\big[W_{i,j}(x_1, y_1) W_{i,j}(x_2, y_2)\big]=R_{i,j}(x_1\wedge x_2, y_1\wedge y_2).
\]
\end{lem}

The first step is to prove the following analog of Lemma~\ref{lem:help e n}:

\begin{lem}\label{lem:6tilde}
Under the conditions of Theorem~\ref{thm:main result2}, we have that, as $n\to\infty$,
\begin{align*}
	\sqrt{k}\Big(e_{n,i} - \frac{k_i}{k}\Big) &\overset{\p}{\longrightarrow}-W_{i,j}(q_i,\infty),\\
		\sqrt{k}\Big(\widehat{e}_{n,i} - \frac{k_i}{k}\Big) &\overset{\p}{\longrightarrow}-W_{i,j}(q_i,\infty).
\end{align*}
\end{lem}

\begin{proof}
From \eqref{eq:(C.1tilde)},
\[
	\sup_{x\in(0,T]}x^{-\eta}\bigg|\sqrt{k}\Big\{\frac{1}{k}\sum_{t=1}^{n}I_{\big\{\varepsilon_{t,Y_i}>U_i(n/[kx])\big\}} - x\Big\} - W_{i,j}(x,\infty)\bigg|\overset{a.s.}{\longrightarrow}0.
\]
Following the steps leading up to \eqref{eq:(p.30(2))} this implies
\[
	\sup_{x\in(0,T]}x^{-\eta}\bigg|\sqrt{k_i}\Big\{\frac{1}{k_i}\sum_{t=1}^{n}I_{\big\{\varepsilon_{t,Y_i}>U_i(n/[k_ix])\big\}} - x\Big\} - q^{-1/2}W_{i,j}(q_ix,\infty)\bigg|\overset{a.s.}{\longrightarrow}0.
\]
From this relation we may argue as in the proof of Lemma~\ref{lem:help e n} to obtain that
\[
	\sqrt{k_i}\Big(\frac{k}{k_i}e_{n,i}-1\Big)\overset{\p}{\longrightarrow}-q_i^{-1/2}W_{i,j}(q_i,\infty)
\]
and, because $k_i/k\to q_i$,
\[
	\sqrt{k}\Big(e_{n,i}-\frac{k_i}{k}\Big)\overset{\p}{\longrightarrow}-W_{i,j}(q_i,\infty).
\]
The claim for $\widehat{e}_{n,i}$ follows similarly.
\end{proof}

\begin{lem}\label{lem:5}
Under the conditions of Theorem~\ref{thm:main result2}, it holds that for any $\varepsilon\in(0,1)$
\[
	\sup_{x,y\in[\varepsilon,\varepsilon^{-1}]}\big|\widehat{\varphi}_{k/n}^{\ast}(x,y) - \widetilde{\varphi}_{k/n}^{\ast}(x,y)\big|=o_{\p}(1).
\]
\end{lem}

\begin{proof}
In analogy to $r_n^{\pm}(x)$ from \eqref{eq:rnpm}, we define $r_{n,i}^{\pm}(x)=\frac{n}{k}\Big[1-F_i\big\{(1\pm\delta_n)U_i(n/[kx])\big\}\Big]$. Set $\delta_n=n^{\iota-\xi}$ for sufficiently small $\iota>0$ to ensure that $\sqrt{k}\delta_n=o(1)$. Then, (a straightforward analog of) Proposition~\ref{prop:approx} allows us to deduce that w.p.a.~1, as $n\to\infty$,
\begin{align*}
	\widehat{\varphi}_{k/n}^{\ast}(x,y)&\leq  \frac{1}{k}\sum_{t=1}^{n}I_{\big\{\varepsilon_{t,Y_i}>(1-\delta_n)U_i(n/[kx]),\ \varepsilon_{t,Y_j}>(1-\delta_n)U_j(n/[ky])\big\}}\\
	&= \frac{1}{k}\sum_{t=1}^{n}I_{\big\{\varepsilon_{t,Y_i}>U_i\big(n/[kr_{n,i}^{-}(x)]\big),\ \varepsilon_{t,Y_j}>(1-\delta_n)U_j\big(n/[kr_{n,j}^{-}(y)]\big)\big\}}\\
	&= \widetilde{\varphi}_{k/n}^{\ast}\big(r_{n,i}^{-}(x),r_{n,j}^{-}(y)\big).
\end{align*}
Similarly, $\widetilde{\varphi}_{k/n}^{\ast}\big(r_{n,i}^{+}(x),r_{n,j}^{+}(y)\big)\leq \widehat{\varphi}_{k/n}^{\ast}(x,y)$. Thus, the event
\begin{multline*}
	\mathcal{W}_3:=\Big\{\widetilde{\varphi}_{k/n}^{\ast}\big(r_{n,i}^{+}(x),r_{n,j}^{+}(y)\big) - \widetilde{\varphi}_{k/n}^{\ast}(x,y)\leq \widehat{\varphi}_{k/n}^{\ast}(x,y)-\widetilde{\varphi}_{k/n}^{\ast}(x,y)\\
	\leq \widetilde{\varphi}_{k/n}^{\ast}\big(r_{n,i}^{-}(x),r_{n,j}^{-}(y)\big)-\widetilde{\varphi}_{k/n}^{\ast}(x,y)\Big\}
\end{multline*}
occurs w.p.a.~1, as $n\to\infty$. Therefore it suffices to show that
\[
	\sup_{x,y\in[\varepsilon,\varepsilon^{-1}]}\Big|\widetilde{\varphi}_{k/n}^{\ast}\big(r_{n,i}^{\pm}(x),r_{n,j}^{\pm}(y)\big)-\widetilde{\varphi}_{k/n}^{\ast}(x,y)\Big|=o_{\p}(1).
\]
To task this, write
\begin{align*}
\sup_{x,y\in[\varepsilon,\varepsilon^{-1}]}&\Big|\widetilde{\varphi}_{k/n}^{\ast}\big(r_{n,i}^{\pm}(x),r_{n,j}^{\pm}(y)\big)-\widetilde{\varphi}_{k/n}^{\ast}(x,y)\Big|\\
&\leq \sup_{x,y\in[\varepsilon,\varepsilon^{-1}]}\Big|\widetilde{\varphi}_{k/n}^{\ast}\big(r_{n,i}^{\pm}(x),r_{n,j}^{\pm}(y)\big)- R_{n,i,j}\big(r_{n,i}^{\pm}(x),r_{n,j}^{\pm}(y)\big)\Big|\\
&\hspace{1cm} + \sup_{x,y\in[\varepsilon,\varepsilon^{-1}]}\Big|\widetilde{\varphi}_{k/n}^{\ast}(x,y)- R_{n,i,j}(x,y)\Big|\\
&\hspace{1cm} + \sup_{x,y\in[\varepsilon,\varepsilon^{-1}]}\Big|R_{n,i,j}\big(r_{n,i}^{\pm}(x),r_{n,j}^{\pm}(y)\big)- R_{i,j}\big(r_{n,i}^{\pm}(x),r_{n,j}^{\pm}(y)\big)\Big|\\
&\hspace{1cm} + \sup_{x,y\in[\varepsilon,\varepsilon^{-1}]}\Big|R_{n,i,j}(x,y) - R_{i,j}(x,y)\Big|\\
&\hspace{1cm} + \sup_{x,y\in[\varepsilon,\varepsilon^{-1}]}\Big|R_{i,j}\big(r_{n,i}^{\pm}(x),r_{n,j}^{\pm}(y)\big)- R_{i,j}(x,y)\Big|\\
&=:E_1+ E_2 + E_3 + E_4 + E_5.
\end{align*}
The fact that $E_1=o_{\p}(1)$ and $E_2=o_{\p}(1)$ follows from Lemma~\ref{lem:Lem 1} together with $\sup_{x\in[\varepsilon,\varepsilon^{-1}]}|r_{n,i}^{\pm}(x)-x|=o(1)$ for all $i=1,\ldots,D$ from (a straightforward generalization of) Lemma~\ref{lem:r n pm}. That $E_3=o(1)$ and $E_4=o(1)$ follows from Assumption~\ref{ass:R*}*, where the convergence in \eqref{eq:Rij} is uniform by Theorem~1~(v) of \citet{SS06}. Finally, $E_5=o(1)$ follows from $\sup_{x\in[\varepsilon,\varepsilon^{-1}]}|r_{n,i}^{\pm}(x)-x|=o(1)$ and the (Lipschitz) continuity of $R_{i,j}(\cdot,\cdot)$ by Theorem~1~(iii) of \citet{SS06}.
\end{proof}

\begin{lem}\label{lem:3 tilde}
Under the conditions of Theorem~\ref{thm:main result2}, it holds that, as $n\to\infty$,
\[
	\widehat{\varphi}_{k/n}(k_i/k,k_j/k)-\widetilde{\varphi}_{k/n}(k_i/k,k_j/k)\overset{\p}{\longrightarrow}0.
\]
\end{lem}

\begin{proof}
Write
\begin{align*}
 \widehat{\varphi}_{k/n}(k_i/k,k_j/k)-\widetilde{\varphi}_{k/n}(k_i/k,k_j/k) &= \widehat{\varphi}_{k/n}^{\ast}(\widehat{e}_{n,i},\widehat{e}_{n,j})-\widetilde{\varphi}_{k/n}^{\ast}(e_{n,i},e_{n,j})\\
&= \big[\widehat{\varphi}_{k/n}^{\ast}(\widehat{e}_{n,i},\widehat{e}_{n,j})-\widetilde{\varphi}_{k/n}^{\ast}(\widehat{e}_{n,i},\widehat{e}_{n,j})\big]\\
&\hspace{2cm} + \big[\widetilde{\varphi}_{k/n}^{\ast}(\widehat{e}_{n,i},\widehat{e}_{n,j})- R_{n,i,j}(\widehat{e}_{n,i},\widehat{e}_{n,j})\big]\\
&\hspace{2cm} + \big[R_{n,i,j}(\widehat{e}_{n,i},\widehat{e}_{n,j}) - R_{i,j}(\widehat{e}_{n,i},\widehat{e}_{n,j})\big]\\
&\hspace{2cm} + \big[R_{i,j}(\widehat{e}_{n,i},\widehat{e}_{n,j}) - R_{i,j}(e_{n,i},e_{n,j})\big]\\
&\hspace{2cm} +\big[R_{i,j}(e_{n,i},e_{n,j}) - R_{n,i,j}(e_{n,i},e_{n,j})\big]\\
&\hspace{2cm}+\big[R_{n,i,j}(e_{n,i},e_{n,j}) - \widetilde{\varphi}_{k/n}^{\ast}(e_{n,i},e_{n,j})\big]\\
&=: F_1+F_2+F_3 + F_4 + F_5 + F_6.
\end{align*}
By Lemmas~\ref{lem:6tilde}--\ref{lem:5}, $F_1=o_{\p}(1)$ follows. By Lemmas~\ref{lem:Lem 1}--\ref{lem:6tilde}, $F_2=o_{\p}(1)$ and $F_6=o_{\p}(1)$. That $F_4=o_{\p}(1)$ and $F_5=o_{\p}(1)$ follows from the fact that the convergence in \eqref{eq:Rij} is uniform (by Theorem~1~(v) of \citet{SS06}) together with Lemma~\ref{lem:6tilde}. Finally, $F_3=o_{\p}(1)$ follows from Lemma~\ref{lem:6tilde} and the continuity of $R_{i,j}(\cdot,\cdot)$ \citep[Theorem~1~(iii)]{SS06}.
\end{proof}

\begin{proof}[\textbf{Proof of Proposition~\ref{prop:cons}}]
Write 
\begin{align*}
\widehat{R}_{i,j}(q_i, q_j) - R_{i,j}(q_i, q_j) &= \widehat{\varphi}_{k/n}(k_i/k,k_j/k) - R_{i,j}(q_i, q_j)\\
&= \big[\widehat{\varphi}_{k/n}(k_i/k,k_j/k) - \widetilde{\varphi}_{k/n}(k_i/k,k_j/k) \big]\\
&\hspace{2cm} + \big[\widetilde{\varphi}_{k/n}(k_i/k,k_j/k) - R_{n,i,j}(k_i/k,k_j/k)\big] \\
&\hspace{2cm} + \big[R_{n,i,j}(k_i/k,k_j/k) - R_{i,j}(k_i/k,k_j/k)\big] \\
&\hspace{2cm} + \big[R_{i,j}(k_i/k,k_j/k) - R_{i,j}(q_i, q_j)\big] \\
&=:G_1 + G_2 +  G_3 + G_4.
\end{align*}

We show that each term is asymptotically negligible.
By Lemma~\ref{lem:3 tilde}, $G_1=o_{\p}(1)$.

To show that $G_2=o_{\p}(1)$, we write
\begin{align*}
	G_2&= \widetilde{\varphi}_{k/n}^{\ast}(e_{n,i}, e_{n,j}) - R_{n,i,j}(k_i/k,k_j/k)\\
	&= \big[\widetilde{\varphi}_{k/n}^{\ast}(e_{n,i}, e_{n,j}) - R_{n,i,j}(e_{n,i}, e_{n,j})\big]\\
	&\hspace{2cm} + \big[R_{n,i,j}(e_{n,i}, e_{n,j})  - R_{i,j}(e_{n,i}, e_{n,j})\big]\\
	&\hspace{2cm} + \big[R_{i,j}(e_{n,i}, e_{n,j}) - R_{i,j}(k_i/k,k_j/k)\big]\\
	&\hspace{2cm} + \big[R_{i,j}(k_i/k,k_j/k) - R_{n,i,j}(k_i/k,k_j/k)\big]\\
	&=:G_{21} + G_{22} + G_{23} + G_{24}.
\end{align*}
Together, Lemma~\ref{lem:Lem 1} and Lemma~\ref{lem:6tilde}  imply that $G_{21}=o_{\p}(1)$. That the remaining terms are also $o_{\p}(1)$ can be established similarly as for the proof of $F_4=o_{\p}(1)$, $F_5=o_{\p}(1)$ and $F_6=o_{\p}(1)$ in the proof of Lemma~\ref{lem:3 tilde}. Hence, $G_2=o_{\p}(1)$.

From Assumption~\ref{ass:R*}* and the uniform convergence in \eqref{eq:Rij} implied by \citet[Theorem~1~(v)]{SS06}, $G_3=o(1)$. 

That $G_4=o(1)$ follows from the continuity of $R_{i,j}(\cdot,\cdot)$ \citep[Theorem~1~(iii)]{SS06} and $k_i/k\to q_i$, $k_j/k\to q_j$, as $n\to\infty$.

Overall, the conclusion follows.
\end{proof}

\singlespacing

\bibliographystyle{jaestyle2}
\bibliography{thebib}

\end{document}